\documentclass[11pt, letterpaper]{article}

\usepackage{setspace}
\usepackage{makecell}
\usepackage{pifont}     
\newcommand{\cmark}{\ding{51}}   
\newcommand{\xmark}{\ding{55}}
\makeatletter
\newcommand\email[2][]%
   {\newaffiltrue\let\AB@blk@and\AB@pand
      \if\relax#1\relax\def\AB@note{\AB@thenote}\else\def\AB@note{\relax}%
        \setcounter{Maxaffil}{0}\fi
      \begingroup
        \let\protect\@unexpandable@protect
        \def\thanks{\protect\thanks}\def\footnote{\protect\footnote}%
        \@temptokena=\expandafter{\AB@authors}%
        {\def\\{\protect\\\protect\Affilfont}\xdef\AB@temp{#2}}%
         \xdef\AB@authors{\the\@temptokena\AB@las\AB@au@str
         \protect\\[\affilsep]\protect\Affilfont\AB@temp}%
         \gdef\AB@las{}\gdef\AB@au@str{}%
        {\def\\{, \ignorespaces}\xdef\AB@temp{#2}}%
        \@temptokena=\expandafter{\AB@affillist}%
        \xdef\AB@affillist{\the\@temptokena \AB@affilsep
          \AB@affilnote{}\protect\Affilfont\AB@temp}%
      \endgroup
       \let\AB@affilsep\AB@affilsepx
}
\makeatother

\usepackage{amsmath}
\makeatletter
  \g@addto@macro \normalsize{%
    \setlength\abovedisplayskip{5pt plus 0pt minus 0pt}%
    \setlength\belowdisplayskip{5pt plus 0pt minus 0pt}}%
\makeatother

\makeatletter
\renewcommand{\paragraph}{%
  \@startsection{paragraph}{4}%
  {\z@}{1.5ex \@plus 1ex \@minus .2ex}{-1em}%
  {\normalfont\normalsize\itshape\bfseries}%
}
\makeatother

\usepackage{paralist}
\let\OLDthebibliography\thebibliography
\renewcommand\thebibliography[1]{
  \OLDthebibliography{#1}
  \setlength{\parskip}{2pt}
  \setlength{\itemsep}{4pt plus 1pt}
}
\usepackage{newtxtext,newtxmath}

\usepackage{amsthm}

\usepackage{amssymb}
\usepackage{geometry}
\usepackage{empheq}
\DeclareMathAlphabet{\mathcal}{OMS}{zplm}{m}{n}

\usepackage{booktabs} 
\usepackage{tabularx}
\usepackage[ruled, linesnumbered]{algorithm2e} 

\usepackage{amsmath}
\usepackage{amssymb}
\usepackage{amsthm}
\usepackage{thm-restate}
\usepackage{array}
\usepackage{xcolor}
\usepackage{multicol}
\usepackage[colorlinks, pagebackref]{hyperref}
\def\tmp#1#2#3{%
  \definecolor{Hy#1color}{#2}{#3}%
  \hypersetup{#1color=Hy#1color}}
\tmp{link}{HTML}{800006}
\tmp{cite}{HTML}{2E7E2A}
\tmp{file}{HTML}{131877}
\tmp{url} {HTML}{8A0087}
\tmp{menu}{HTML}{727500}
\tmp{run} {HTML}{137776}
\def\tmp#1#2{%
  \colorlet{Hy#1bordercolor}{Hy#1color#2}%
  \hypersetup{#1bordercolor=Hy#1bordercolor}}
\tmp{link}{!60!white}
\tmp{cite}{!60!white}
\tmp{file}{!60!white}
\tmp{url} {!60!white}
\tmp{menu}{!60!white}
\tmp{run} {!60!white}

\definecolor{myred}{RGB}{230,0,0}

\usepackage{natbib}

\usepackage{backref}

\renewcommand*{\backref}[1]{}
\renewcommand*{\backrefalt}[4]{%
  \ifcase #1 
     (Not cited.)%
  \or       
     (p.~#2.)%
  \else     
     (pp.~#2.)%
  \fi%
}

\newcommand{\apprefproof}[1]{\hyperref[#1]{{\mbox{\footnotesize\textsc{\upshape{[proof]}}}}}}

\usepackage[scaled=.85]{helvet}

\usepackage{comment}
\usepackage{url}

\SetAlFnt{\small}
\SetAlCapFnt{\small}
\SetAlCapNameFnt{\small\scshape}
\SetAlgoSkip{smallskip}



\SetCommentSty{mycommfont}

\usepackage{mathtools}
\usepackage{xcolor}

\usepackage{tikz}
\usetikzlibrary{
    graphs,
    graphs.standard,
    shapes,
    arrows.meta,
    backgrounds,
    positioning
}

\usepackage{wrapfig}
\usepackage{enumitem}
\usepackage{graphicx}
\usepackage{titling}
\usepackage{subcaption}
\usepackage{color}
\usepackage[noblocks]{authblk}
\renewcommand\Affilfont{\normalsize}
\usepackage[normalem]{ulem}
\usepackage{nicefrac}
\usepackage{pgfplots}
\usepgfplotslibrary{fillbetween}
\pgfplotsset{compat=1.18}

\newtheorem{theorem}{Theorem}[section]
\newtheorem{lemma}{Lemma}[section]
\newtheorem{fact}{Fact}[section]
\newtheorem{proposition}{Proposition}[section]
\newtheorem{corollary}{Corollary}[section]
\newtheorem{example}{Example}[section]
\theoremstyle{definition}
\newtheorem{definition}{Definition}[section]

\allowdisplaybreaks

\newcommand{\gs}[1]{\textcolor{blue}{#1}}
\newcommand{\at}[1]{\textcolor{red}{AT: #1}}

\newcommand{\gsrem}[1]{\marginpar{\tiny\gs{GS: #1}}}
\newcommand{\atrem}[1]{\marginpar{\tiny\at{#1}}}

\usepackage{bm}
\renewcommand{\vec}[1]{\bm{#1}}
\newcommand\numberthis{\addtocounter{equation}{1}\tag{\theequation}}

\newcommand{\ppr}{\ensuremath{\mathbb{P}}}
\newcommand{\pe}{\mathbb{E}}
\DeclareMathOperator*{\argmax}{arg\,max}

\DeclareMathOperator{\EX}{\mathbb{E}}
\newcommand{\supp}{\mathrm{supp}}
\newcommand{\set}[1]{\ensuremath{\{ #1 \}}}
\newcommand{\sset}[2]{\ensuremath{\{ #1 \, \mid \, #2 \}}}

\newcommand{\ssetl}[2]{\ensuremath{\left\{ #1 \; \bigg\vert \; #2 \right\}}}

\newcommand{\sens}{\ensuremath{\sigma}}
\newcommand{\type}{\ensuremath{t}}
\newcommand{\z}{\ensuremath{g}}
\newcommand{\budget}{\ensuremath{\mathfrak{B}}}

\newcommand{\sw}{\ensuremath{\textit{LW}}}
\newcommand{\opt}{\ensuremath{\textit{OPT}}}

\newcommand{\poa}{\ensuremath{\textit{POA}}}
\newcommand{\poarmp}{\ensuremath{\text{POA-RMP}}}

\newcommand{\mech}{\ensuremath\mathcal{M}}
\newcommand{\fpa}{\ensuremath{\textit{FPA}}\xspace}

\newcommand{\fpar}[2]{%
\ifthenelse{\equal{#1}{}}%
    {\ifthenelse{\equal{#2}{}}%
        {\ensuremath{\textit{FPA}}\xspace}%
        {\ensuremath{\textit{FPA}(#2)}\xspace}}%
    {\ifthenelse{\equal{#2}{}}
        {\ensuremath{\textit{FPA}(#1)}\xspace}%
        {\ensuremath{\textit{FPA}(#1, #2)}\xspace}}%
}

\newcommand{\rw}{\ensuremath{\textsf{rw}}}
\newcommand{\aw}{\ensuremath{\textsf{aw}}}

\newcommand{\equ}{\ensuremath{\textit{EQ}}}
\newcommand{\cce}{\ensuremath{\textit{CCE}}}
\newcommand{\ce}{\ensuremath{\textit{CE}}}
\newcommand{\swap}{\ensuremath{h}}
\newcommand{\mne}{\ensuremath{\textit{MNE}}}
\newcommand{\pne}{\ensuremath{\textit{PNE}}}

\newcommand{\val}{\ensuremath{\textsc{val}}}
\newcommand{\add}{\ensuremath{\textsc{add}}}
\newcommand{\sub}{\ensuremath{\textsc{sub}}}
\newcommand{\xos}{\ensuremath{\textsc{xos}}}

\renewcommand{\b}{\ensuremath{\vec{b}}}
\newcommand{\B}{\ensuremath{\vec{B}}}
\newcommand{\bsingle}{\ensuremath{b}}
\newcommand{\Bsingle}{\ensuremath{B}}
\newcommand{\dist}{\ensuremath{\pi}} 

\begin{document}

\title{\bfseries
Optimal Type-Dependent Liquid Welfare Guarantees for Autobidding Agents with Budgets}

\author[1]{Riccardo Colini-Baldeschi}
\author[2,3]{Sophie Klumper}
\author[2,3]{Twan Kroll}
\author[4]{\\ Stefano Leonardi}
\author[2,3]{Guido Sch\"afer}
\author[5]{Artem Tsikiridis}

\affil[1]{Meta, UK}
\affil[2]{Centrum Wiskunde \& Informatica (CWI), The Netherlands}
\affil[3]{University of Amsterdam, The Netherlands}
\affil[4]{Sapienza University of Rome, Italy}
\affil[5]{Technical University of Munich, Germany}

\date{}

\maketitle
\begin{abstract}
\noindent
Online advertising systems have recently transitioned to autobidding, allowing advertisers to delegate bidding decisions to automated agents. Each advertiser directs their agent to optimize an objective function subject to \emph{return-on-investment (ROI)} and \emph{budget constraints}. Given their practical relevance, this shift has spurred a surge of research on the \emph{liquid welfare price of anarchy (POA)} of fundamental auction formats under autobidding, most notably \emph{simultaneous first-price auctions (FPA)}. 
One of the main challenges is to understand the efficiency of FPA in the presence of \emph{heterogeneous} agent types.

We introduce a \emph{type-dependent smoothness framework} that enables a unified analysis of the POA in such complex autobidding environments. 
In our approach, we derive type-dependent smoothness parameters which we carefully balance to obtain POA bounds. This balancing gives rise to a \emph{POA-revealing mathematical program}, which we use to determine tight bounds on the POA of coarse correlated equilibria (CCE). Our framework is versatile enough to handle heterogeneous agent types and extends to the general class of fractionally subadditive valuations. Additionally, we develop a novel reduction technique that transforms budget-constrained agents into budget-unconstrained ones. Combining this reduction technique with our smoothness framework enables us to derive tight bounds on the POA of CCE in the general hybrid agent model with both ROI and budget constraints. 
Among other results, our bounds uncover an intriguing threshold phenomenon showing that the POA depends intricately on the smallest and largest agent types. 
We also extend our study to FPAs with reserve prices, which can be interpreted as predictions of agents' values, to further improve efficiency guarantees. 
\end{abstract}

\section{Introduction}

\paragraph{Background and Motivation.}

Over the past decade, online advertising systems have undergone a major shift with the emergence of autobidding. This shift allows advertisers to delegate complex bidding decisions to automated agents that take various factors into account such as ad performance, campaign constraints, and market dynamics. As a result, advertisers can manage their campaigns more efficiently and aim to maximize their return on investment.
Autobidding is now the dominant paradigm: over 80\% of online ad traffic is managed by autobidding agents \citep{deng2024}. This widespread adoption has important implications for the behavior of advertisers, publishers, and ad exchanges.

In the autobidding world, advertisers specify high-level constraints for their campaigns---most notably \emph{return-on-investment (ROI) constraints} and \emph{budget constraints}. Basically, the ROI constraint caps the cost per conversion or impression, while the budget constraint limits the total spend for a given campaign. Unlike the traditional view, where agents have intrinsic values for outcomes, autobidding agents operate under ROI and budget constraints reflecting performance goals and financial limitations. In addition to these constraints, agents may differ in the objective they seek to optimize. The autobidding literature typically considers two types: \emph{value maximizers}, who aim to maximize outcome value subject to budget and ROI constraints, and \emph{utility maximizers}, who aim to maximize value minus payment. These two types represent the extremes of a spectrum capturing agents' trade-offs between value and payment. In this work, we allow agents' types to lie anywhere along this spectrum, modeling diverse behaviors in a unified way.

Alongside the rise of autobidding, the online advertising industry has also undergone a major shift in auction formats---most notably, the move from second-price to \emph{first-price auctions (FPA)}. This transition has been especially pronounced in display ad markets, culminating in Google Ad Exchange's adoption of first-price auctions in 2019 \citep{PST20}. The combination of constraint-driven autobidding agents and first-price payment schemes raises fundamental questions about the performance of FPA in this new environment. This paper is driven by one central question: \emph{What is the efficiency of simultaneous first-price auctions in the presence of heterogeneous autobidding agents?}

\paragraph{State of the Art.}

The study of the inefficiency of equilibria for simultaneous first-price auctions in the autobidding setting was initiated by \citet{LMP23}.
To evaluate the efficiency of FPAs in the autobidding setting, we study the \emph{Price of Anarchy (POA)} \citep{koutsoupias1999worst}, which compares the optimal liquid welfare to the welfare achieved in the worst-case equilibrium. \emph{Liquid welfare} \citep{dobzinski2014efficiency} is a widely used metric that extends standard social welfare to settings with budget constraints by capping an agent's contribution at their available budget. This provides a more meaningful benchmark for evaluating the social welfare achieved by a mechanism when agents face budget constraints. 
The POA of first-price auctions under autobidding was recently studied by \cite{deng2024}, both in the model with only value maximizers and in the \emph{mixed agent model} with both value and utility maximizers. They establish tight bounds of $2$ and $2.188$ on the POA of mixed Nash equilibria (MNE) for value maximizers and the mixed agent model, respectively. However, their analysis is limited to agents subject to ROI constraints only. Their proofs rely on structural properties of MNE and are tailored to their specific setting. 
\citet{LMZ24} studied budget-constrained value maximizers and derive a bound of $2$ on the POA of pure Nash equilibria. 
To the best of our knowledge, the work of \citet{LMZ24} is the only one that studies the inefficiency of simultaneous FPAs in the autobidding setting under both ROI and budget constraints.

A powerful framework for proving POA bounds is the \emph{smoothness technique}, first introduced by \citet{Rou15} for strategic games and later extended by \citet{ST13} to composable auctions. This framework allows POA bounds established for a `base mechanism' (e.g., single-item first-price auction) to extend to more complex compositions (e.g., simultaneous first-price auctions) and to broader equilibrium concepts, including correlated and coarse correlated equilibria. 
Its ease of application combined with the strong, general guarantees it provides has made smoothness the technique of choice for studying the price of anarchy. 
Despite its success, applying the smoothness framework to the autobidding setting---where constraints and heterogeneous types play a central role---remains a major open challenge. This paper closes that gap.

\begin{table}[t]
\centering\small
\begin{tabular}{lccc}
\toprule
\small{\textbf{Agent Types}}& \small{\textbf{Budgets}} & \small{ \textbf{POA (UB)}}  & \small{ \textbf{Statements}} \\
\midrule\midrule
\\[-2.2ex]
$\type_{\max}\in [0,1]$    &  \cmark     & $P(\type_{\max})$ & Thm.~\ref{theorem:full-hybrid-no-reserve}\\ [.5ex]
\midrule \\[-2.2ex]
$\type_{\min}=\type_{\max}=\type$    &  \xmark     & $Q(t)$ &Thm.~\ref{theorem:single-type:budgetfree}  \\ [.5ex]
$\type_{\min}\geq \beta$    &  \xmark     & $\left(\type_{\min}\left(1-e^{-\frac{1}{\type_{\min}}}\right)\right)$ & Thm.~\ref{theorem:not-so-hybrid} \\ [.5ex]
\bottomrule
\end{tabular}
\caption{Overview of POA upper bounds for CCEs of simultaneous FPA for fractionally subadditive valuations.
Here, $\type_{\max}$ and $\type_{\min}$ refer to the largest and smallest agent types, respectively.
The function $P$ is defined as $P(\type)=1 + \type(1+W_{0}(-e^{-\type -1}))^{-1}$ for $\type\geq \theta$ and $P(\type)=2$ for $\type<\theta$, where $W_0$ is the principal branch of the Lambert $W$ function and $\theta$ is a threshold value defined as $\theta = 1 + \frac12 W_{0}(-2e^{-2}) \approx 0.797$. Furthermore, the function $Q$ is defined as $Q(\type)=\frac{e}{e-1}$ for $\type\geq \frac{e-1}{e}$, $Q(0)=2$, and $Q(\type)=1-\frac{(1-\type)\ln(1-\type)}{\type}$, otherwise. Lastly, $\beta\approx 0.741$ and is the solution to $\beta=1-e^{-\frac{1}{\beta}}$}
\label{tab:poa-upperbounds}
\end{table}

\begin{table}[t]
\centering\small
\begin{tabular}{lcccc}
\toprule
\small{\textbf{Agent Types}}& \small{\textbf{Budgets}} & \small{ \textbf{POA (LB)}} & \small{\textbf{Equilibrium class}} & \small{ \textbf{Statements}} \\
\midrule\midrule
\\[-2.2ex]
\text{any}  & \cmark& 2 & 
\mne& Thm.~\ref{lem:universal-LB-budget} \\[.5ex]
$\type_{\min}=\type_{\max}=\type > \theta$ & \cmark & $P(\type)$ 
& \cce & Thm.~\ref{thm:LB-POA-budget-commontype} \\ [.5ex]
\midrule \\[-2.2ex]
$\type_{\min}=\type_{\max}=t\geq \frac{e-1}{e}$    &  \xmark     & $\frac{e}{e-1}$ & \cce& Cor.~\ref{corollary:high-common-type-no-reserve} \\ [.5ex]
$\type_{\min}=0$, $\type_{\max}\in [0,1]$ & \xmark & $P(\type)$ 
&\mne & Thm.~\ref{thm:LB-POA-budgetfree-hybrid}  \\ [.5ex]
\bottomrule
\end{tabular}
\caption{Overview of POA lower bounds of simultaneous FPA for additive valuations. Here, $\type_{\max}$ and $\type_{\min}$ refer to the largest and smallest agent types, respectively.
The functions $P$ and $Q$ and threshhold value $\theta$ are defined in the caption of Table~\ref{tab:poa-upperbounds}.}
\label{tab:poa-lowerbounds}
\end{table}

\paragraph{Our Contributions.}

Building on the challenges posed by heterogeneous autobidding agents and the shift to first-price auctions, we introduce a \emph{type-dependent smoothness framework} that enables a unified analysis of the POA in this complex environment. Our contributions extend prior results in several important directions:

\begin{enumerate}\itemsep0pt

    \item We consider a general \emph{hybrid agent model} in which agents differ in type and are subject to both ROI and budget constraints. 
    The \emph{type} $\type \in [0,1]$ of an agent reflects their aversion to payments.
    In particular, $\type = 0$ corresponds to a value maximizer, and $\type =1$ to a utility maximizer. By allowing arbitrary $\type \in [0,1]$, our model captures the full spectrum of agent behavior between these two extremes. This model strictly generalizes the mixed agent model \citep{deng2024} and has been suggested in \citep{balseiro2021robust,survey}.
    
    \item We significantly broaden the class of valuation functions that can be handled in the autobidding context by analyzing \emph{fractionally subadditive} (also known as \emph{XOS}) valuation functions \citep{ST13}. 
    This class notably includes monotone submodular functions as a special case. These functions are particularly relevant in multi-platform autobidding environments, where advertisers use multiple platforms and experience diminishing returns as their ads are shown across them (see, e.g., \citep{aggarwal2025multi}). To the best of our knowledge, prior work on autobidding has focused exclusively on additive valuations.

    \item We extend our analysis beyond pure and mixed Nash equilibria and derive (tight) POA bounds for more general solution concepts, including \emph{correlated equilibria (CE)} and \emph{coarse correlated equilibria (CCE)}, even in the full generality of our hybrid agent model. CCE are particularly relevant in autobidding settings when agents use regret-minimizing algorithms (see, e.g.,~\citet{lucier2024autobidders} and references therein). 
    Building on insights from \citet{KN22}, we show that CCE induced by such learning dynamics satisfy additional structural 
    essential for our framework to be applicable.
                
   \item We also study the POA of simultaneous first-price auctions with reserve prices. Prior work by \citet{balseiro2021robust} and \citet{deng2024} studied this setting in the absence of budgets and under additive valuations. We extend their analysis to encompass fractionally subadditive valuations, both with and without budget constraints. As observed by \citet{BDM21}, reserve prices can be interpreted as \emph{predictions} of the agents' values (e.g., derived from historical data through machine-learning techniques), which can be used to improve efficiency guarantees. By leveraging such predictions to set reserve prices, we obtain improved POA bounds. This connects to the broader agenda of \emph{mechanism design with predictions} (see also \citep{gkatzelis22, chen04}).

\end{enumerate}

\begin{figure}
\newlength{\mywidth}
\centering
\begin{subfigure}{7cm}
\setlength{\mywidth}{6cm}
\hspace*{-.9cm}
  \begin{tikzpicture}
  [ declare function={
    func(\x)= and(\x>= 0.632120558829,\x<=1) * (1.58197670687)   +
    and(\x>0, \x<0.632120558829) * (1 + ln(1-\x)- ln(1-\x)/\x)+
    (\x==0) * (2);}
    ]
    \begin{axis}[
      scale only axis,
      width=\mywidth,
      height=.75\mywidth,
      xmin=0, xmax=1.07,
      ymin=1.4, ymax=2.25,
      axis x line=bottom,
      axis y line=left,
      axis line style={line width=1pt},
      tick style={semithick},
      ticklabel style={font=\small},
      xlabel={$t$},
      xlabel style={yshift=.12\mywidth,xshift=.53\mywidth},
      ylabel={},
      ylabel style={xshift=2.2cm, yshift=-1.2cm, rotate=-90},
      xtick={0,0.632120558829,0.79681213,1},
      xticklabels={$0$,$\frac{e-1}{e}$,$\theta$, $1$},
      ytick={1.58197670687,2,2.18848737},
      yticklabels={$\frac{e}{e-1}$,$2$, $2.18$},
      clip=false,
    ]
     \addplot[red, domain=0:1, samples = 100, very thick]{func(x)};
      
      \addplot[very thick,blue] coordinates {
        (0,          2.00)
        (0.79681213, 2.00)
      };
      \addplot[very thick,blue,smooth] coordinates {
        (0.79681213, 2.00000000)
        (0.81713092, 2.01897860)
        (0.83744970, 2.03792135)
        (0.85776849, 2.05698248)
        (0.87808727, 2.07619123)
        (0.89840605, 2.09556944)
        (0.91872484, 2.11515834)
        (0.93904362, 2.13497837)
        (0.95936240, 2.15505003)
        (0.97968119, 2.17539109)
        (1.00000000, 2.18848737)
      };
      	
	\node at (0.82,2.15) {$P(t)$};  
    \node at (0.82,1.65) {$Q(t)$};  
    \end{axis}
  \end{tikzpicture}
  \caption{ }
  \label{fig:commontype-budget-vs-budgetfree}
\end{subfigure}
\begin{subfigure}{8cm}
\setlength{\mywidth}{2.73cm}
\centering
\pgfplotsset{
  my style/.style={
    scale only axis,
    width=\mywidth,
    height=.75\mywidth,
    xmin=0, xmax=1.1,
    ymin=0.5, ymax=2.3,
    axis x line=bottom,
    axis y line=left,
    axis line style={line width=0.8pt},
    tick style={semithick},
    ticklabel style={font=\scriptsize},
    xlabel={\scriptsize{$\eta$}},
    xlabel style={yshift=.2\mywidth, xshift=.55\mywidth},
    ylabel={},
    ylabel style={xshift=1pt},
    xtick={0,0.5,1},
    xticklabels={\scriptsize$0$, \scriptsize$0.5$, \scriptsize$1$},
    ytick={1,1.581, 2},
    yticklabels={\scriptsize$1$, \scriptsize$\frac{e}{e-1}$, \scriptsize$2$},
    clip=false,
    title style={at={(0.5,0.85)}, anchor=north, font=\small},
  }
}%
\begin{subfigure}{.45\textwidth}
\centering
\begin{tikzpicture}
\begin{axis}[my style, title={$t=0$}]
  \addplot[very thick, myred, domain=0:1, samples=20] {2 - x};
\end{axis}
\end{tikzpicture}
\end{subfigure}
\quad
\begin{subfigure}{.45\textwidth}
\centering
\begin{tikzpicture}
\begin{axis}[my style, title={$t=0.3$}]
  \addplot[very thick, myred, domain=0:1, samples=50] {1 + (0.7/0.3)*ln((1-0.3*x)/0.7)};
\end{axis}
\end{tikzpicture}
\end{subfigure}

\begin{subfigure}{.45\textwidth}
\centering
\begin{tikzpicture}
\begin{axis}[my style, title={ $t=0.7$}]
  \pgfmathsetmacro{\etabound}{(1 - exp(1)*(1-0.7))/0.7} 
  \addplot[very thick, myred, domain=0:\etabound, samples=20] {exp(1)/(exp(1)-1 + 0.7*x)};
  \addplot[very thick, myred, domain=\etabound:1, samples=50] {1 + (0.3/0.7)*ln((1-0.7*x)/0.3)};
\end{axis}
\end{tikzpicture}
\end{subfigure}
\quad
\begin{subfigure}{.45\textwidth}
\centering
\begin{tikzpicture}
\begin{axis}[my style, title={$t=1$}]
  \addplot[very thick, myred, domain=0:1, samples=50] {exp(1)/(exp(1)-1 + x)};
\end{axis}
\end{tikzpicture}
\end{subfigure}
\caption{}
\label{fig:overview-with-eta}
\end{subfigure}
  
  \caption{Illustration of POA bounds of simultaneous FPA with fractionally subadditive valuations. 
  (a) Bounds for the setting without reserve prices as a function of the type $\type$. 
  (b) Bounds for the setting with reserve prices as a function of $\eta$ (quality of the reserve prices) for $\type \in \set{0, 0.3, 0.7, 1}$.}
  \label{fig:overview}
\end{figure}

\paragraph{Overview and Significance of Our Results.}

An overview of the POA upper bounds that we derive for CCE of simultaneous FPA with fractionally subadditive valuations is given in Table~\ref{tab:poa-upperbounds}. The corresponding lower bounds that we derive can be found in Table~\ref{tab:poa-lowerbounds}. Below, we highlight some key implications of our results. Our bounds depend on the set of agent types $T \subseteq [0,1]$, where $\type_{\min} = \min(T)$ and $\type_{\max} = \max(T)$ denote the smallest and largest agent types, respectively. 

Among other results, our bounds uncover an intriguing threshold phenomenon:
In the hybrid agent model with budgets and heterogeneous types, the POA is exactly $2$ when the largest agent type satisfies $\type_{\max} \le \theta \approx 0.797$.
When $\type_{\max} > \theta$ and value maximizers are present (i.e., $\type_{\min} = 0$), the POA bound increases smoothly with $\type_{\max}$, following the function $P(\type_{\max})$, from $2$ up to $2.188$ as $\type_{\max}$ approaches $1$; see Figure~\ref{fig:commontype-budget-vs-budgetfree} (blue curve). This bound is tight even for mixed Nash equilibria.
This result unifies and generalizes the state-of-the-art POA bounds of \citet{deng2024} and \citet{LMZ24}, both of which address special cases of our model. 
The same bound also applies when agents are budget-constrained and have a \emph{uniform} (or \emph{homogeneous}) type $\type$, i.e., $\type_{\min} = \type_{\max} = \type$\footnote{When agents are utility maximizers i.e., $t=1$, our objective coincides with the notion of \emph{effective welfare} introduced by \citep{ST13} (see also \citep{caragiannis16}). Note that, for budget-constrained agents, \citet{ST13} establish a guarantee of $\frac{e}{e-1}$ for the weaker benchmark of the ratio between optimal effective welfare and the utilitarian social welfare at equilibrium, see also Further Related Work.}.

Interestingly, we obtain the exact same POA bounds without budget constraints, as long as value maximizers are present. As will become clear below, this is not a coincidence. 
In the budget-free setting with uniform agent type, we derive strictly improved bounds, illustrated in Figure~\ref{fig:commontype-budget-vs-budgetfree} (red curve).
This yields a natural separation result for uniform agents: for every $\type > 0$, the POA for budget-constrained agents (Figure~\ref{fig:commontype-budget-vs-budgetfree} (blue curve)) is \emph{strictly worse} than that for budget-free agents (Figure~\ref{fig:commontype-budget-vs-budgetfree} (red curve)).
Our results also reveal a second threshold phenomenon: In the budget-free setting, if agents are sufficiently close to utility maximizers, i.e., $\type_{\min} \ge \beta \approx 0.741$, the POA is in guaranteed to lie in the range $[\nicefrac{e}{(e-1)}, \nicefrac{1}{\beta^2}]$ (as a function of $\type_{\min}$), regardless of type heterogeneity; see Theorem~\ref{theorem:not-so-hybrid} for details.

We briefly comment on our results for simultaneous first-price auctions with reserve prices. In this setting, the POA bounds additionally depend on a parameter $\eta \in [0,1)$, which quantifies the quality of the reserve prices. Intuitively, higher values of $\eta$ correspond to better reserve prices; formally $\eta$ captures the worst-case ratio between the reserve price and the maximum extractable payment over all auctions (see below for a precise definition). We derive new bounds for this setting (not included in the overview table). For budget-free agents with a uniform type $\type$, the POA bound is shown in Figure~\ref{fig:overview-with-eta}. As expected, the POA improves with the quality of the reserve prices, i.e., it decreases as $\eta$ increases and converges to $1$ as $\eta$ goes to $1$. The bound is tight for both $\type = 0$ and $\type = 1$. 

In our setting with reserve prices, our framework requires that all items are sold in equilibrium outcomes---a property we refer to as \emph{well-supported equilibria}. We show that, for additive valuations, this property holds for equilibria up to and including correlated equilibria (CE), but fails for coarse correlated equilibria (CCE). However, building on \citep{KN22}, we prove that in repeated single-item first-price auctions with reserve prices, CCE arising from regret-minimizing agents are well-supported regardless of the agents’ type distribution. This result implies that CCE produced by such learning dynamics inherently possess the well-supported property.

\paragraph{Our Techniques.}

At the heart of our analysis lies a novel, \emph{type-dependent smoothness framework}---a non-trivial generalization of the original framework by \citet{ST13} for utility maximizers to heterogeneous agent types. A key technical challenge is incorporating these diverse agent types.
Prior smoothness-based works assume that agents are `alike', allowing the $(\lambda, \mu)$-smoothness parameters of the base mechanism to lift directly to composed mechanisms. In contrast, our setting requires handling heterogeneous types, and applying the original approach directly fails to yield meaningful bounds in the autobidding context.

We overcome this by proving a smoothness inequality for each type separately, yielding type-specific $(\lambda_{\type}, \mu_{\type})$ parameters. The core insight is to balance these parameters---crucially leveraging the ROI constraint---through carefully chosen \emph{calibration vectors} to obtain the best POA bound. This balancing leads to an optimization problem over feasible choices of the smoothness parameters, giving rise to what we term the \emph{$\poa$-revealing mathematical program} (\poarmp). Bounding the objective of this program then yields upper bounds on the POA.
With this machinery in place, we then prove smoothness for a single-item first-price auction with reserve prices across different agent types. By the Extension Theorem, we derive (often tight) POA bounds for coarse correlated equilibria of simultaneous first-price auctions with reserve prices and fractionally subadditive valuations.

To handle budget constraints, we introduce a novel reduction technique that transforms instances with budgets into equivalent \emph{budget-free proxy instances}. A key insight is that any budget-constrained agent can be simulated `almost perfectly' by a budget-unconstrained agent with a \emph{budget-capped valuation function}, i.e., one that caps their original valuation at their budget.
This transformation, however, comes with a caveat: it fails when an agent's valuation exceeds their budget. We resolve this by introducing, for each such agent, a budget-unconstrained \emph{value maximizer} with the corresponding budget-capped valuation function. The resulting instance is our \emph{budget-free proxy instance}.

In a nutshell, our strategy is to reduce an instance with budget-constrained agents to a budget-free proxy instance, and then apply our type-dependent smoothness framework to bound the POA of the proxy instance. Since the transformation preserves the POA, this yields bounds for the original instance as well.
This approach, however, requires special care:
(i) The proxy instance includes budget-capped valuation functions. For example, if the original valuations are additive, the capped versions become submodular.
(ii) The transformation introduces heterogeneity in agent types. Even if the original instance had a single type $\type \ne 0$, the proxy instance includes both type $\type$ and type $0$ agents.
Addressing both (i) and (ii) relies critically on the full power of our type-dependent smoothness framework. For (i), the framework supports XOS valuations, and crucially, budget-capped XOS functions remain XOS. For (ii), our framework is explicitly designed to accommodate heterogeneous agent types. Notably, any approach incapable of handling either aspect would render the reduction technique infeasible.

Finally, the augmented type set in (ii) explains a key pattern observed in our results: all POA bounds for instances with budgets and type set $T$ match those of their budget-free proxy instance, whose type set becomes $T^+ = T \cup \set{0}$.

We demonstrate the power of our type-dependent smoothness technique by analyzing the POA of simultaneous first-price auctions. However, the technique is broadly applicable, and we believe it extends to a wide range of autobidding environments---and potentially even beyond. Indeed, we already have evidence that it can be used to bound the POA of multi-unit auctions under autobidding, further underscoring its broader impact.

\paragraph{Further Related Work} 

\citet{ABM19} initiated the study of the inefficiency of equilibria for auctions in autobidding environments. Their result implies that the liquid welfare price of anarchy for pure Nash equilibria of the second-price auction is~$2$. This upper bound was later generalized by \citet{deng21}, who considered a more general autobidding environment and the VCG mechanism, while \citet{deng24efficiency} and \citet{deng2023autobidding} obtained POA bounds for the Generalized Second Price auction~(GSP).

The study of the inefficiency of equilibria for simultaneous FPAs was initiated by \citet{LMP23}, who showed that when all agents are value maximizers constrained only by ROI and have additive valuations, the POA of pure Nash equilibria is also~$2$. This result was then extended by \citet{deng2024} to MNE and ROI-constrained agents. \citet{deng2024} were also the first to study the mixed-agent model, i.e., the setting where agents can be either utility or value maximizers, which we capture as a special case. \citet{LMZ24} study the inefficiency of simultaneous FPAs for agents who are both ROI-constrained and budget-constrained, focusing on pure Nash equilibria and showing a POA of~$2$. To the best of our knowledge, the work of \citet{LMZ24} is the only one that studies the inefficiency of simultaneous FPAs under both ROI and budget constraints.

Other settings in the context of autobidding beyond simultaneous compositions of simple classical auction formats that have been considered in the literature include the inefficiency of randomized auction mechanisms (see, e.g., \citep{M22,LMP23}) and scenarios in which the platform is allowed to ``boost'' the budgets of agents and implement reserve prices (see, e.g., \citet{DMM21,balseiro2021robust}), with the latter having an interpretation as \emph{machine-learned advice}. We refer the reader to Section~\ref{sec:reserve-prices} for a discussion of this perspective. Finally, beyond the inefficiency of equilibria, other directions relevant to autobidding include the study of optimal bidding from the perspective of the agent (see, e.g., \citep{ABM19,balseiro2015repeated}), online learning (see, e.g., \citep{balseiro2019learning,feng2023online,castiglioni2022online, aggarwal2025noR, lucier2024autobidders}), auction design (\citep{golrezaei2021auction,balseiro2021landscape,lv2023auction}), and the emerging domain of multi-platform (multi-channel) autobidding (see, e.g., \citep{deng23,susan2023multi,aggarwal25}). For further details, we refer the interested reader to the survey by \citet{aggarwalSurvey}.

For the standard setting in which all agents are utility maximizers, the inefficiency of the first-price auction has been studied in the economics literature since the seminal work of \citet{vickrey}. Naturally, due to its simplicity, it has also been considered for simultaneous simple auctions. The $\poa$ of $\cce$ was shown by \citet{ST13} to be at most $\frac{e}{e-1}$ for simultaneous auctions with XOS valuations. This bound was shown to be tight even for a single auction by \citet{syrgkanis-thesis}, and for $\mne$ in simultaneous auctions with submodular valuations by \citet{christodoulou16}. Beyond the smoothness framework, \citet{feldmanfu} showed a $\poa$ upper bound of $2$ for simultaneous auctions with subadditive valuations, while \citet{feldman16} considered $\ce$ and $\cce$ and their properties more closely. More recently, \citet{JL23} showed that the \emph{Bayesian} $\poa$ for the single-item first-price auction is exactly $\frac{e^2}{e^{2}-1}$, a breakthrough result.
For an overview of classical results regarding the POA of auctions for utility maximizers, including compositions of other simple auctions, we refer to the survey by \citet{roughgarden2017price}.

Finally, we remark that budgeted settings have been considered for utility maximizers from a POA perspective prior to the emergence of autobidding. Since our model captures the setting where all agents are utility maximizers as a special case, our work can also be viewed as a follow-up to this line of research. In this context, the three most closely related works are those by \citet{ST13}, \citet{azar17}, and \citet{caragiannis16}. While \citet{caragiannis16} focus on a different setting, both \citet{ST13} and \citet{azar17} study simultaneous first-price auctions for XOS valuations. However, \citet{ST13} focus on the ratio of the expected \emph{social} welfare at equilibrium to the optimal liquid welfare, a weaker benchmark than the one we consider in this work (see also Section \ref{theorem:full-hybrid-no-reserve}). On the other hand, \citet{azar17} analyze the \emph{ex post} liquid welfare, which is a stronger benchmark, but their results require the items to be divisible into discretely many parts.\footnote{The stronger benchmark is unbounded in our setting of simultaneous first-price auctions, see Theorem C.2.\ in \cite{azar2015liquidpriceanarchy}.} Note that \citet{caragiannis16} adopt the same benchmark as we do and call it \emph{Effective Welfare}.

\section{Preliminaries} \label{sec:preliminariesPOA}

\paragraph{Simultaneous First-Price Auctions with Reserve Prices.}

We study \emph{simultaneous first-price auctions}, where a set $N = [n]$ of $n \ge 2$ agents simultaneously participate in a set $M = [m]$ of $m \ge 1$ single-item auctions. Each auction $j \in M$ implements a \emph{first-price auction ($\fpa$) with reserve price}, as detailed below. We use $j$ to denote both the auction and the respective item interchangeably. Each agent $i \in N$ submits a bid $b_{ij} \in \mathbb{R}_{\ge 0}$ to each auction $j \in M$. We use $\b_i = (b_{ij})_{j \in M}$ to denote the bid profile of agent $i$, and $D_i = \mathbb{R}_{\ge 0}^m$ to refer to the set of all bid profiles of $i$. 
The aggregated bid profile of all agents is denoted by $\b = (\b_i)_{i \in N} \in D = \times_{i \in N} D_i$.

We focus on simultaneous first-price auctions with reserve prices. More specifically, each auction $j \in M$ handles a reserve price $r_j \in \mathbb{R}_{\ge 0}$ that must be met to sell item $j$; we use $\fpar{}{r_j}$ to refer to this auction. Given the bid profile $\b_j = (b_{ij})_{i \in N}$ submitted to auction $j$, $\fpar{}{r_j}$ allocates the item to the highest bidder $i$ meeting the reserve price $r_j$, i.e., $b_{ij} \ge r_j$, and charges their respective bid $b_{ij}$ for the item. 
The agent who wins the item (if any) is called the \emph{actual winner}, denoted by $\aw(j) \in \arg\max_{i \in N: b_{ij} \ge r_j} b_{ij}$. In case of ties, the actual winner is chosen according to an arbitrary but fixed tie-breaking rule. If the reserve price $r_j$ is not met (i.e., $b_{ij} < r_j$ for all $i \in N$), we define $\aw(j) = \emptyset$. 
Let $\vec{x}_j(\b) = (x_{ij}(\b))_{i \in N}$ and $\vec{p}_j(\b) = (p_{ij}(\b))_{i \in N}$ be the respective allocation and payments of $\fpar{}{r_j}$, i.e., for $i = \aw(j)$ we have $x_{ij}(\b) = 1$ and $p_{ij}(\b) = b_{ij}$, and for $i \neq \aw(j)$ we have $x_{ij}(\b) = 0$ and $p_{ij}(\b) = 0$.\footnote{Both $\vec{x}_j$ and $\vec{p}_j$ only depend on the input profile $\b_j$. However, we often use $\vec{b}$ as the argument for notational convenience.}
We write $\vec{x}_j(\b) \neq \vec{0}$ to indicate that item $j$ is sold.

Our global mechanism, denoted by $\mech$, implements the above mechanisms with reserve prices $\vec{r} = (r_j)_{j \in M}$ simultaneously. 
That is, given a bid profile $\b$, the outcome $\mech(\vec{r}, \b) = (\vec{x}(\b), \vec{p}(\b))$ is determined by the allocation $\vec{x}(\b)= (\vec{x}_{j}(\b))_{j \in M}$ and the payments $\vec{p}(\b) = (\vec{p}_{j}(\b))_{j \in M}$ obtained by running the $m$ auctions (i.e., $\fpar{}{r_j}$ for each $j \in M$) simultaneously. We write $\vec{x}(\b) \neq \vec{0}$ to indicate that all items are sold under $\b$, i.e., $\vec{x}_j(\b) \neq \vec{0}$ for all $j \in M$.\footnote{Note that we slightly abuse notation as $\vec x(\b) \neq \vec 0$ here indicates that there is exactly one $1$ in each row of $\vec x(\b) \in \set{0,1}^{m \times n}$.}
We use $\vec{x}_i(\vec{b}) = (x_{ij}(\vec{b}))_{j \in M} \in \set{0,1}^{m}$ and $p_i(\b) = \sum_{j \in M} p_{ij}(\b)$ to denote the allocation and total payment of agent $i$, respectively. Further, we define $\vec{X}$ as the set of all feasible allocations, i.e., $\vec{X} = \set{\vec{x} = (\vec{x}_{i})_{i \in N} \in \set{0,1}^{m\times n} \mid \sum_{i \in N} x_{ij} \le 1\ \forall j \in M}$. 
We slightly overload notation and use $\vec{x}_i(\vec{b})$ also to refer to the set of items assigned to $i$, i.e., $\vec{x}_i(\vec{b}) = \sset{j \in M}{x_{ij}(\b) = 1} \subseteq M$. Additionally, we sometimes omit the argument $\vec{b}$ when it is clear from context.

We use $\fpar{m}{\vec{r}}$ and $\fpar{m}{}$, respectively, to refer to $m$ simultaneous first-price auctions with reserve prices $\vec{r}$ and without reserve prices. 
We use $\fpar{}{r}$ to indicate that we consider a single-item first-price auction (i.e., $m = 1$) with reserve price $r$.
If $m = 1$, we drop the auction index $j = 1$ from all our notation.

\paragraph{Valuation Functions.} 
Each agent $i \in N$ has a valuation function $v_i: 2^M \rightarrow \mathbb{R}_{\ge 0}$ over the subsets of the items, where $v_i(S)$ specifies the value that $i$ obtains when receiving the items in $S \subseteq M$. We assume w.l.o.g. that $v_i(\emptyset) = 0$. Also, we assume that $v_i$ is \emph{monotone}, i.e., $v_i(S) \le v_i(T)$ for all $S \subseteq T \subseteq M$.
We use $\mathcal{V}_i$ to denote the class of valuation functions of agent $i$ and let $\mathcal{V} = \times_{i \in N} \mathcal{V}_i$ be the set of all valuation functions of the agents.
We use $\vec{v} = (v_i)_{i \in N} \in \mathcal{V}$ to refer to the profile of valuation functions of the agents. 
We consider different classes of valuation functions:

\begin{definition}\label{def:val-classes}
Let $v_i: 2^M \rightarrow \mathbb{R}_{\ge 0}$ be a valuation function. 
\begin{compactitem}\itemsep0pt
\item $v_i$ is \emph{additive} if there exist additive valuations $(v_{ij})_{j \in M} \in \mathbb{R}^m_{\ge 0}$ such that for every subset $S \subseteq M$, it holds that $v_i(S) = \sum_{j \in S} v_{ij}$. 
\item $v_i$ is \emph{submodular} if $v_i(S \cup \set{j}) - v_i(S) \ge v_i(T \cup \set{j}) - v_i(T)$ for all $S \subseteq T \subseteq M$. 
\item $v_i$ is \emph{fractionally subadditive} (or, \emph{XOS}), if there exists a class $\mathcal{L}_{i} = \set{(v^\ell_{ij})_{j \in M} \in \mathbb{R}^m_{\ge 0}}$ of additive valuations such that for every subset $S \subseteq M$, it holds that $v_i(S) = \max_{\ell \in \mathcal{L}_{i}} \sum_{j \in S} v^\ell_{ij}$. 
\end{compactitem}
\end{definition}

Let $\mathcal{V}_{\add}$, $\mathcal{V}_{\sub}$ and $\mathcal{V}_{\xos}$ refer to the set of additive, submodular and fractionally subadditive (XOS) valuation functions, respectively. It is well-known (see e.g., \citep{lehmann01}) that $\mathcal{V}_{\add} \subset \mathcal{V}_{\sub} \subset \mathcal{V}_{\xos}$.

\paragraph{Random Bid Profiles.}
Each agent $i$ can randomize over their deterministic (or pure) bid profiles $\b_i$ in $D_i$. We define $\Delta_i$ as the space of random bid profiles of $i$ over $D_i$. Let $\dist$ be a probability distribution over the set of bid profiles in $D$; we use $\Delta$ to refer to the set of all such probability distributions.
We use $\B \sim \dist$ to denote a random bid profile that is drawn from $\dist$; we often omit the reference to $\dist$ and identify $\B$ with $\dist$. We use $f_{\B}$ and $F_{\B}$ to refer to the probability density function (PDF) and cumulative distribution function (CDF) of $\B$, respectively.
The support of $\B$ refers to the set of bid profiles that have positive density, i.e., $\supp(\B) = \sset{\b \in D}{f_{\B}(\b)>0}$.
If $\supp(\B) = \set{\b}$ then $\B$ chooses $\b$ deterministically and we write $\B = \b$.
We use $\supp_i(\B)$ to refer to the set of bid profiles $\b_i$ of agent $i$ that have positive density under $\B$. 
The \emph{marginal} $\vec{B}_{-i}$ of $\B$ is defined by the following PDF: 
\[
\forall \b_{-i} \in D_{-i}: \qquad
f_{\B_{-i}}(\b_{-i}) = \int_{D_i} f_{\B}(\b_i, \b_{-i}) d \b_i.
\]
Given a bid profile $\b'_i$ of agent $i$, we denote by $(\b'_i, \B_{-i})$ the random bid profile that we obtain from $\B$ when agent $i$ bids $\b'_i$ deterministically and the other agents bid according to the marginal $\B_{-i}$.
We say that a bid profile $\B$ is \textit{well-supported} with respect to reserve prices $\vec{r}$ if the items are always sold under $\B$, i.e., for each $\b \in \supp(\B)$, $\vec{x}(\b) \neq \vec{0}$.

\paragraph{Hybrid Agent Model.}
We consider the general \emph{hybrid agent model} (see, e.g., \citep{balseiro2021robust,survey}), where each agent $i \in N$ maximizes their \emph{gain function} $\z_i: D \rightarrow \mathbb{R}$ defined as 
\begin{equation}\label{eq:agent-objective}
\z_i(\b) = 
v_i(\vec{x}_i(\b)) - \sens_i \cdot p_{i}(\b).
\end{equation}
Here, $\sens_i \in [0,1]$ defines the \emph{type} of agent $i$. 
Intuitively, $\sens_i$ represents $i$'s sensitivity to payments: a higher value indicates that $i$ is more negatively affected by payments. In particular, agent $i$ is a value maximizer if $\sens_i = 0$, and a utility maximizer if $\sens_i = 1$. Our model thus allows us to capture a large spectrum of agents' types, ranging from value maximizers to utility maximizers. 
Most previous works focus on the special case of the \emph{mixed agent model} consisting of value and utility maximizers only, i.e., $\sens_i \in \set{0,1}$ for all $i \in N$. 

Each agent $i$  has a \emph{return-on-investment (ROI) constraint} and a \emph{budget contsraint} that must be satisfied (see \citep{survey}). 
Given a bid profile $\vec{B}$, the ROI constraint of an agent $i$ enforces that the expected total payment of $i$ is at most a factor $\tau_i \in \mathbb{R}_{>0}$ of their expected valuation for the received items, where $\tau_i$ is the so-called \emph{target parameter} of $i$, i.e.,
\begin{equation}\label{eq:ROI}
\pe \left[p_{i}(\B)\right]\le \tau_i \cdot \pe[v_i(\vec{x}_i(\B))].
\end{equation}
Additionally, the budget constraint of agent $i$ requires that the expected total payment of $i$ is at most $\budget_i \in \mathbb{R}_{> 0} \cup \set{\infty}$, i.e.,
\begin{equation}\label{eq:budget}
\pe \left[p_{i}(\B)\right]\le \budget_i.
\end{equation}
For an agent $i$, we define $\mathcal{R}_i$ as the set of bid profiles $\B$ that satisfy both the ROI contraint \eqref{eq:ROI} and the budget constraint \eqref{eq:budget}.
It is not hard to see that we can assume w.l.o.g.~that $\tau_i \sens_i \le 1$ for all agents $i \in N$; we refer to Appendix~\ref{app:prelim} for more details. 

Formally, we use $I = (N, M, \vec{r}, \vec{v},  \vec{\sens}, \vec{\tau}, \vec{\budget})$ to denote an instance. For ease of notation, we omit explicit references to $N$ and $M$ and simply write $I = (\vec{r}, \vec{v},  \vec{\sens}, \vec{\tau}, \vec{\budget})$. 
We say that an instance $I$ is \emph{budget-free} if $\budget_i = \infty$ for all agents $i \in N$.
The set of \emph{agent types} of an instance $I$ is defined as $T(I) = \sset{\type}{\text{$\exists i \in N$ with $\sens_i = \type$}}$. Given an instance $I$, we use $N_{\type}(I) \subseteq N$ to refer to the subset of agents having type $\type$, i.e., $N_{\type}(I) = \sset{i \in N}{\sens_i = \type}$. 
The following notion will turn out to be useful below. Given a type set $T \subseteq [0,1]$, we define the \emph{augmented type set} $T^+$ as follows: for instances with budgets, we define $T^+ = T \cup \set{0}$ as the type set obtained from $T$ by adding the value-maximizing type $0$; for budget-free instances, we define $T^+ = T$ simply.

We use $\mathcal{I}$ to refer to a class of instances. 
For $\val \in \set{\add, \sub, \xos}$, we use $\mathcal{I}_{\val}$ to refer to the set of instances with valuation functions from $\mathcal{V}_{\val}$, i.e., $v_i \in \mathcal{V}_{\val}$ for all $i \in N$.
We use $\mathcal{I}^{\infty}$ to refer to the set of all budget-free instances.
Finally, we also define a class of instances whose type set is restricted. Given a set of types $T \subseteq [0,1]$, we use $\mathcal{I}^{T}$ to denote the set of instances $I$ whose type set is restricted to $T$, i.e., $T(I) = T$.
Note that the above restrictions will also be combined. For example, $\mathcal{I}^{T, \infty}_{\val}$ denotes the class of all budget-free instances with valuations functions from $\mathcal{V}_{\val}$ and agent types being restricted to $T$.

\paragraph{Agent's Problem and Equilibrium Notions.}

The objective of each agent $i$ is to determine a random bid profile $\B_i$ that, given the bid profile $\B_{-i}$ of the other agents, maximizes their gain subject to their ROI and budget constraints.
More formally, each agent $i$ solves the following \emph{agent's problem}:
\[
\max_{\B_i} \ \pe[\z_i(\B_i, \B_{-i})] \quad \text{subject to \quad $(\B_i, \B_{-i}) \in \mathcal{R}_i$.}
\]
The resulting bid profile $\B$ constitutes an equilibrium if for each agent $i$ only the deviations satisfying the ROI and budget constraints are considered. We consider the following equilibrium notions in this paper.

\begin{definition} \label{def:equilibria}
Let $\vec{B} \in \Delta$ be a bid profile satisfying $\B \in \mathcal{R}_i$ for each agent $i \in N$.
\begin{compactitem}
\item 
$\vec{B}$ is a \emph{coarse correlated equilibrium (\cce)} if for every agent $i \in N$ we have:
\begin{equation}
\label{eq:cce-def}
\pe [\z_i(\vec{B})] \ge \pe [\z_i(\vec{B}'_i, \vec{B}_{-i})] \qquad \forall (\vec{B}'_i, \vec{B}_{-i}) \in {\mathcal{R}_i}.
\end{equation}

\item 
$\vec{B}$ is a \emph{correlated equilibrium (\ce)} if for every agent $i \in N$ and every swapping function $\swap: \supp_{i}(\vec{B}) \mapsto \Delta_i$ we have: 
\begin{equation}
\label{eq:ce-def}
\pe [\z_i(\vec{B}) ] \ge \pe [\z_i( \swap(\vec{B}_i), \vec{B}_{-i}) ] \qquad \forall (\swap(\vec{B}_i), \vec{B}_{-i}) \in {\mathcal{R}_i}.
\end{equation}

\item $\vec{B}$ is a \emph{mixed Nash equilibrium (\mne)} if $\vec{B} = \prod_{i \in [n]} \vec{B}_i$ and for every agent $i \in N$ we have: 
\begin{equation}
\label{eq:mne-def}
\pe [\z_i(\vec{B})] \ge \pe [\z_i(\vec{B}'_i, \vec{B}_{-i})] \qquad \forall (\vec{B}'_i, \vec{B}_{-i}) \in {\mathcal{R}_i}.
\end{equation}
\end{compactitem}
\end{definition}

Given an instance $I$, we use $\mne(I)$, $\ce(I)$ and $\cce(I)$ to refer to the sets of mixed, correlated and coarse correlated equilibria of $I$, respectively. Note that $\mne(I) \subseteq \ce(I) \subseteq \cce(I)$; a more elaborate discussion on the equilibrium hierarchy can be found in Appendix \ref{sec:appequilibria}. We use $\equ$ as a generic placeholder for an equilibrium notion with $\equ \in \set{\mne, \ce, \cce}$.

\paragraph{Liquid Price of Anarchy.}
We use \emph{liquid welfare} as the social welfare objective, which is also the standard benchmark in the autobidding literature (see, e.g., \citep{survey}).
Intuitively, the liquid welfare measures the maximum amount of payments one can extract from the agents. We refer to Appendix~\ref{app:prelim} for a more detailed discussion of the liquid welfare objective. 
Given an instance $I$ and a bid profile $\vec{B}$, the liquid welfare is defined as
$$
\sw(I, \vec{B}) = \sum_{i \in N}  \min (\EX[\tau_i v_i(\vec{x}_i(\vec{B}))], \budget_i). 
$$
Given an instance $I$, an optimal allocation $\vec{x}^*(I) \in \vec{X}$ maximizes the liquid welfare over all feasible allocations:
\[
\vec{x}^*(I) \in \arg\max_{\vec{x} \in \vec{X}}\sum_{i \in N}  \min (\tau_i v_i(\vec{x}_i), \budget_i).
\]
We use $\opt(I) =\sum_{i \in N}  \min (\tau_i v_i(\vec{x}^*_i(I)), \budget_i)$ to denote the optimal liquid welfare. 

In this paper, we study the \emph{price of anarchy} as introduced by \citet{koutsoupias1999worst} with respect to the liquid welfare objective: the price of anarchy is defined as the worst-case ratio of the optimal liquid welfare over the expected liquid welfare of any equilibrium. More formally, given a set of instances $\mathcal{I}$ and an equilibrium notion $\equ \in \set{\mne, \ce, \cce}$, we define the \emph{price of anarchy with respect to $\equ$} as:
\[
\equ\text{-}\poa(\mathcal{I}) =
\sup_{I \in \mathcal{I}} \sup_{\B \in \equ(I)}
\frac{\opt(I)}{\sw(I, \B)}.
\]

As we show in Proposition~\ref{prop:uniform-target} (see Appendix~\ref{app:prelim}), when studying the liquid price of anarchy we can assume without loss of generality that $\tau_i = 1$ for all $i \in N$. Subsequently, we therefore omit the target parameters (with the understanding that $\vec \tau = \vec 1$) and refer to an instances as $I = (\vec r, \vec v, \vec \sens, \vec \budget)$.

\paragraph{Reserve Prices.}
\citet{balseiro2021robust} and \citet{deng2024} studied the effect of reserve prices on the price of anarchy for budget-free instances and additive valuations. Their bounds depend on a parameter $\eta$ that measures the relative gap between the reserve price and the highest valuation of an agent over all auctions. 
We extend their model to instances with fractionally subadditive valuations (with or without budgets). Let $I \in \mathcal{I}_{\xos}$ be such an instance with reserve prices $\vec{r}$. We can then choose \emph{opt-induced additive representatives} $(v_{ij})_{j \in M}$ for each agent $i$ with respect to their allocation $\vec{x}^{*}_i$ in the \emph{optimal solution} $\vec{x}^{*}(I)$ (see Section~\ref{sec:extension-theorem} for more details).
An agent $i$ is said to be the \emph{rightful winner} of auction $j$, denoted by $\rw(j)$, if $i$ is an agent with maximum valuation for the item, i.e., $\rw(j) \in \arg\max_{i \in N} v_{ij}$. In case of ties, we let $\rw(j)$ denote the winner of auction $j$ in the considered optimal allocation. 
For each auction $j \in M$, define the relative gap $\eta_j \in [0,1)$ such that $r_j = \eta_j v_{\rw(j)j}$, and let $\eta = \min_{j \in M} \eta_j$ be the smallest relative gap. Note that, as in \citep{balseiro2021robust} and \citep{deng2024}, we only consider reserve prices $\vec{r}$ that satisfy $\eta_j \in [0,1)$ for each auction $j \in M$. In fact, it is not hard to see that otherwise the POA is unbounded (see Appendix~\ref{app:prelim} for more details). 

\paragraph{The Lambert $W$ Function.} 
In order to derive POA bounds analytically, we use the Lambert $W$ function, which is the multivalued inverse of the function $f(z) = z e^z$. 
In this work, we exclusively use the \emph{principal branch} of the Lambert $W$ function and denote it by $W_0$; more details are given in the appendix.

\medskip

Several proofs and further discussions are deferred to the appendix, with all omitted material organized by the corresponding sections; hyperlinks are provided in the main text.

\section{Type-Dependent Smoothness Framework} \label{sec:smoothness:extension}

We introduce a new type-dependent smoothness framework that enables us to bound the POA of coarse correlated equilibria in simultaneous first-price auctions, in the full generality of our hybrid model.
We begin by formalizing the type-dependent smoothness notion and deriving corresponding smoothness lemmas. To handle budget constraints, we apply a reduction technique that transforms instances with budgets into budget-free proxy instances. Our Extension Theorem then leverages this framework to establish upper bounds on the POA of coarse correlated equilibria. A key technical challenge lies in balancing the type-dependent smoothness parameters to obtain the best possible bounds. To this end, we formulate a mathematical program that facilitates the analysis of the price of anarchy.

\subsection{Type-Dependent Smoothness}

In this section, we focus on single-item instances (i.e., $m = 1$) and thus omit the auction index $j = 1$ from the notation. Note that in this case, fractionally subadditive valuation functions reduce to additive ones. In particular, each agent $i$ has a single additive value $v_i$ for the item (see Definition~\ref{def:val-classes}). 

We need the notion of ROI-restricted bid profiles. Let $B'_i \in \Delta_i$ be a bid profile of agent $i$. We say that $B'_i$ is \emph{ROI-restricted} if for each $\b_{-i} \in D_{-i}$, it holds that
    $
    \pe[p_i(B'_i, \b_{-i})] \le  \pe[v_i(\vec{x}_i(B'_i, \b_{-i}))].
    $
We can now introduce our new type-dependent smoothness notion:

\begin{definition}[Type-Dependent Smoothness]\label{def:smooth} 
Let $I = (r, \vec v, \vec \sens, \vec \budget)$ be a single-item instance with reserve price $r$.
Let the rightful winner $i = \rw$ be of type $\type$. 
Then, $\fpar{}{r}$ is \emph{$(\lambda_{\type}, \mu_{\type})$-smooth for type $\type$} with $\lambda_{\type}, \mu_{\type} >0$, if there is a ROI-restricted bid $B'_{i} = B'_{i}(\vec{v}) \in \Delta_i$ such that for every well-supported bid profile $\b$ we have 
\begin{equation}\label{eq:def:smoothness}
            \pe[\z_{i}(B'_{i}, \b_{-{i}})] \ge \lambda_{\type}  v_{i} - \mu_{\type}  p_{\aw(\vec{b})}(\b).
\end{equation}
\end{definition}

We remark that, crucially, the random deviation $B'_{i}$ of the rightful winner $i$ may depend on the valuations $\vec{v}$ but \emph{not} on the bid profile $\b$. 
Note that \eqref{eq:def:smoothness} needs to hold only for bid profiles $\b$ that are well-supported, i.e., when the item is sold; clearly, this condition is redundant in the setting without reserve prices, i.e., $r = 0$. 

\paragraph{Smoothness Lemmas.}

The following two lemmas establish type-dependent smoothness of $\fpar{r}{}$ for different types $\type \in [0,1]$. We start with a simple smoothness lemma for type-0 agents (i.e., value maximizers). 
Recall that $\eta \in [0,1)$ measures the relative gap between the reserve price and the valuation of the rightful winner, i.e., $r= \eta v_{\rw}$.

\begin{restatable}{lemma}{fpasmoothnesskeytwo}\label{lem:fpa-smoothness-key-2}
Consider a single-item instance $I = (r, \vec v, \vec \sens, \vec \budget)$ and let the rightful winner $i = \rw$ be of type $\type = 0$. 
Then, $\fpar{}{r}$ is $(\lambda_{\type}, \mu_{\type})$-smooth for type $\type$ with $\lambda_{\type} = \mu_{\type} = \mu$ for every $\mu \in (0,(1-\eta)^{-1}]$.
\end{restatable}
\begin{proof}
    Assume that the rightful winner $i$ is of type $\type=0$. 
We need to show that there exists a ROI-restricted random bid $B'_{i}$ such that for every well-supported bid profile $\b$ and $\aw = \aw(\b)$ the actual winner, it holds that:
\begin{equation}\label{eq:def:smoothness-repeat2}
    \pe\left[\z_{i}\left(B'_{i}, \b_{-{i}}\right)\right] \ge \mu v_{i} - \mu p_{\aw}(\b).
\end{equation}

Let $B'_{i} = B'_{i}(\vec{v})$ be a random unilateral deviation of $i$ drawn from $[\eta v_{i}, v_i]$ with CDF $F(z) =  F_{B'_{i}}(z) = \mu z / v_i + 1-\mu$. 
Note that the domain is well-defined as $\eta \in [0,1)$, and it is easy to verify that $F(\cdot)$ is non-negative and increasing over $[\eta v_{i}, v_i]$ and $F(v_{i}) = 1$.
Also, $B'_{i}$ is ROI-restricted as the condition is even pointwise satisfied, i.e., for every $z \in [\eta v_{i}, v_{i}]$ it holds that $p_{i}(z, \cdot) \le v_{i}(\vec{x}_{i}(z, \cdot))$. 

It remains to show that $B'_{i}$ satisfies \eqref{eq:def:smoothness-repeat2}.
Note that the expected gain of $i$ is always non-negative, as $i$ bids above $v_i$ with 0 probability. 
Thus, \eqref{eq:def:smoothness-repeat2} holds trivially if $v_{i} \le p_{\aw}(\b)$.
Therefore, assume that $v_{i} > p_{\aw}(\b)$ and define $\theta_{i} := \max( \eta v_{i}, \max_{j \neq i} b_j)$.
For every $z \ge \theta_{i}$, $i$ wins the item under the bid profile $(z, \vec{b}_{-i})$ and pays $p_{i}(z, \vec{b}_{-i}) = z$.
As the item is sold under the bid profile $\vec{b}$ by assumption, the actual winner under $\vec{b}$ either pays the reserve price or their maximum bid, i.e., $p_{\aw}(\b) = \max( \eta v_{i}, \max_{j} b_j)$.
We obtain: 
\[
\theta_{i} 
= \max \bigg( \eta v_{i}, \max_{j \neq i} b_j \bigg) 
\le \max \bigg ( \eta v_{i}, \max_j b_j \bigg)
= p_{\aw}(\b)
< v_{i}.
\]
This leads to the desired result as:
\begin{align*}
\pe\left[\z_{i}\left(B'_{i}, \vec{b}_{-{i}}\right)\right] 
& = v_{i} (1-F(\theta_{i})) 
= v_i \left(1- \left( \frac{\mu \theta_i}{v_i} + 1-\mu \right)\right) \\
&= \mu v_i - \mu \theta_i 
\ge \mu v_i - \mu p_{\aw}(\b).
\end{align*}
Note that the first equality holds because the sensitivity of $i$ is $\sens_{i} =t = 0$.
\end{proof}

A key insight used in the proof is that the reserve price allows us to increase the probability mass of the random deviation for larger bids. This provides a better trade-off in terms of the smoothness parameters. 
Our smoothness lemma for agent types $\type \in (0,1]$ follows the same approach, but is technically more challenging. 

\begin{restatable}{lemma}{fpasmoothnesskeyone} \label{lem:fpa-smoothness-key-1}
Consider a single-item instance $I = (r, \vec v, \vec \sens, \vec \budget)$ and let the rightful winner $i = \rw$ be of type $\type \in (0,1]$.
Then, $\fpar{}{r}$ is $(\lambda_{\type}, \mu_{\type})$-smooth for type $\type$ with 
\begin{equation}\label{eq:alpharestr}
\lambda_{\type} =  \frac{\mu}{\type}\left(1 - \frac{1 - {\type\eta}}{e^{\type/{\mu}}}\right) 
\quad\text{and}\quad
\mu_{\type} = \mu
\quad\text{for every}\quad
\begin{cases}
	\mu \ge \type \left( \ln\left(\frac{1-{\type\eta}}{1- \type} \right)\right)^{-1}, & \text{if \ $\type < 1$,}\\
    \mu > 0 , & \text{if \ $\type = 1$.}
\end{cases}
\end{equation}
\end{restatable}

We elaborate on the expression of \eqref{eq:alpharestr} in Lemma \ref{lem:fpa-smoothness-key-1}. First, note that $\type \in (0, 1]$, as $\type > 0$ and $\type \le 1$ by assumption. 
The $\ln(\cdot)$ expression decreases as $\eta$ increases and converges to $0$ (from above) as $\eta \rightarrow 1$; the lower bound restriction on $\mu$ thus increases as $\eta$ increases. 
As $\eta < 1$, the $\ln(\cdot)$ expression is well-defined for all combinations of $\type$ and $\eta$, except when $\type = 1$. In the latter case, we only impose the restriction that $\mu > 0$.

Given $\mu$, we define a parameter $\gamma = \gamma(\mu)$ as follows: 
\begin{equation}\label{def:gamma}
    \gamma(\mu) 
    = \frac{1}{t}\left(1 - \frac{1 - {t\eta}}{e^{t/{\mu}}}\right),
\end{equation}
which will be useful in the proof of Lemma \ref{lem:fpa-smoothness-key-1}.
Note that $\gamma$ is well-defined because $t > 0$ by assumption. 

The following corollary is an immediate consequence of the definitions above. 
\begin{corollary} \label{cor:gamma}
Let $\mu$ satisfy \eqref{eq:alpharestr} and let $\gamma$ be defined as in \eqref{def:gamma}. 
Then, $\gamma \in [\eta,1]$.
\end{corollary}

\begin{proof} 
Note that the interval $[\eta, 1]$ is well-defined because $\eta \in [0, 1)$ by assumption. 
We first prove the lower bound on $\gamma$. 
Note that $e^{t/\mu} > 1$ as $\mu > 0$ and $t >0$, and therefore:
\[
    \gamma = \frac{1}{t} \bigg(1 - \frac{1 - {t \eta}}{e^{{t}/{\mu}}}\bigg)
    > \frac{1}{t} \cdot t \eta
    = \eta.    
\]    
For the upper bound on $\gamma$, we have that:
\[  
    \gamma 
    = \frac{1}{t} \left(1 - \frac{1 - {t\eta}}{e^{{t}/{\mu}}}\right) 
    \leq \frac{1}{t} \left(1 - \frac{1 - {t\eta}}{e^{\ln \left(\frac{1 - {t\eta }}{1 - t}\right)}}\right) 
    = \frac{1}{t} \cdot  t = 1,
\]
where the inequality follows from \eqref{eq:alpharestr} and because $e^{x}$ is non-decreasing in $x$.
\end{proof}

We now continue with the proof of Lemma \ref{lem:fpa-smoothness-key-1}.
\medskip

\begin{proof}[Proof of Lemma~\ref{lem:fpa-smoothness-key-1}]
Assume that the rightful winner $i$ is of type $\type \in T$ with $\sens_i = t > 0$. 
We need to show that there exists a ROI-restricted random bid $B'_i$ such that for every well-supported bid profile $\b$ and $\aw = \aw(\b)$ the actual winner, it holds that:
\begin{equation}\label{eq:def:smoothness-repeat}
    \pe\left[\z_{i}\left(B'_{i}, \b_{-{i}}\right)\right] \ge \mu\gamma  v_{i} - \mu  p_{\aw}(\b).
\end{equation}

Let $B'_i = B'_i(\vec{v})$ be a random unilateral deviation of $i$ drawn from $[\eta v_i, \gamma v_{i}]$ with PDF $f(z) = f_{B'_i}(z) = \mu/(v_i - t z)$. 
Note that the domain is well-defined as $\gamma \in [ \eta, 1]$ by Corollary~\ref{cor:gamma}, and that $f(\cdot)$ is non-negative.
Also note that:
\begin{align*}
    \int_{\eta v_i}^{\gamma v_i}f(z) d z
    & 
    = \int_{\eta v_i}^{\gamma v_i}\frac{\mu}{v_{i}- t z} d z
    = \mu \int_{\eta v_i}^{\gamma v_i} \left(\frac{-\ln\left(v_i-t z \right)}{\sens}\right)' d z \\
    &=\frac{\mu}{t}\ln\left(\frac{1-t\eta }{1-t\gamma}\right)
    = \frac{\mu}{t}\ln\left(e^{{t}/{\mu}} \right)
    = 1,
\end{align*}
where the fourth equality follows from the definition of $\gamma$ in \eqref{def:gamma}.
Furthermore, $B'_i$ is ROI-restricted as the condition is even pointwise satisfied, i.e., for $z \in [\eta v_i, \gamma v_i]$ it holds that $p_i(z, \cdot) \le \gamma v_i(\vec{x}_i(z, \cdot)) \le v_i(\vec{x}_i(z, \cdot))$, as $\gamma \in [ \eta, 1]$ by Corollary~\ref{cor:gamma}. 

It remains to show that $B'_i$ satisfies \eqref{eq:def:smoothness-repeat}.
Note that the expected gain of $i$ is always non-negative, as $i$ bids above $v_i$ with 0 probability and $t \le 1$.
Thus, \eqref{eq:def:smoothness-repeat} holds trivially if $ \gamma v_i \le  p_{\aw}(\b)$.
Therefore, assume that $\gamma v_i >  p_{\aw}(\b)$ and define $\theta_i := \max( \eta v_i, \max_{j \neq i} b_j)$.
Then, for every $z \ge \theta_i$, $i$ wins the item under bid profile $(z, \vec{b}_{-i})$ and pays $p_i(z, \vec{b}_{-i}) = z$.
As the item is sold under bid profile $\vec{b}$ by assumption, the actual winner under $\vec{b}$ either pays the reserve price or their maximum bid, i.e., $p_{\aw}(\b) = \max( \eta v_i, \max_{j} b_j)$.
We obtain:
\[
\theta_i 
= \max \bigg ( \eta  v_i, \max_{j \neq i} b_j \bigg) 
\le \max \bigg ( \eta  v_i, \max_j b_j \bigg )
= p_{\aw}(\b)
< \gamma v_i.
\]
This leads to the desired result as:
\begin{align*}
\pe\left[\z_i\left(B'_i, \vec{b}_{-i}\right)\right] 
& 
= \int_{\theta_i}^{\gamma v_i} (v_i - t p_i(z, \vec{b}_{-i}) )f(z) d z 
= \int_{\theta_i}^{\gamma v_i} (v_i  - t z) f(z) d z  \\
& 
= \int_{\theta_i}^{\gamma v_i} \mu d z  
= \mu\gamma v_i - \mu\theta_i 
\ge \mu \gamma v_i - \mu p_{\aw}(\b).
\end{align*}
Note that the first equality holds because the sensitivity of $i$ is $\sens_i=t$ by the precondition of the Lemma. 
\end{proof}

\subsection{Extension Theorem} \label{sec:extension-theorem}

We present our Extension Theorem to derive bounds on the price of anarchy. Our bounds depend on the set of available agent types $T \subseteq [0,1]$. We consider the class of instances $\mathcal{I}_{\xos}^{T}$ with fractionally subadditive valuations and type set $T$.

We first introduce the notion of \emph{calibration vectors}, which will be crucial in the proof below. As we show in Lemma~\ref{lemma:optimal-xi} below, the set of calibration vectors $\mathcal{C}(\vec{\mu}, T)$ is always non-empty.

\begin{definition}[Calibration Vectors]\label{def:Xi}
Let $T$ be a set of agent types and let $\vec{\mu} = (\mu_{\type})_{\type \in T}$ be such that $\mu_{\type} >0$ for each $\type \in T$. We define the set of \emph{calibration vectors} $\mathcal{C}(\vec{\mu}, T)$ as follows: 
 \begin{equation}\label{equation:Xi-def}
     \mathcal{C}(\vec{\mu}, T) = 
     \ssetl{\vec{\delta} \in (0,1]^{|T|}}{  
     \max_{\type \in T} (\delta_{\type} \mu_{\type} )  + \max_{\type \in T} (\delta_{\type}(1-\type)) 
     \leq 1}.
 \end{equation}
\end{definition}

We can now state the main result of this section.
Recall that $T^+$ is the augmented type set of $T$, where $T^+ = T \cup \set{0}$ for instances with budgets and $T^+ = T$ for budget-free instances. 

\begin{theorem}[Extension Theorem]
\label{theorem:template}
Let $\mathcal{I}_{\xos}^T$ be the class of instances with fractionally subadditive valuations and type set $T$. Assume that $\fpa(r)$ is $(\lambda_{\type}, \mu_{\type})$-smooth for each type $\type \in T^+$. Then, the price of anarchy of well-supported coarse correlated equilibria bounded by
\[
\cce\text{-}\poa(\mathcal{I}^{T}_{\xos}) \le 
\left(\max_{\vec{\delta} \in \mathcal{C}(\vec{\mu}, T^+)}\min_{\type \in T^+} \delta_{\type}\lambda_{\type}\right)^{-1}.
\]
\end{theorem}

The remainder of this section is devoted to the proof of Theorem \ref{theorem:template}.

\subsubsection{Opt-Induced Additive Representatives and Budget-Free Proxy Instances}

We first introduce the notion of \emph{opt-induced additive representatives}. 
Intuitively, these representatives allow us to treat fractionally subadditive valuations as additive ones in the analysis. Let $I \in \mathcal{I}_{\xos}$ be an instance with XOS valuation functions $\vec v = (v_i)_{i \in N}$. 
Fix an optimal allocation $\vec{x}^*:=\vec{x}^*(I)$. By Definition~\ref{def:val-classes}, for each $i \in N$, there exist additive representatives $(v^*_{ij})_{j \in M}$ with respect to the optimal allocation $\vec{x}^*_i$.\footnote{Note that these representatives simply coincide with the input valuations if the valuation functions $\vec v$ are additive.} 
We refer to these representatives as \emph{opt-induced additive representatives}. We define $v^*_i$ as the additive valuation function obtained from these representatives, i.e., $v^*_i(\vec{x}_i) := \sum_{j \in M} v^*_{ij} x_{ij}$ for any allocation $\vec{x}_i \subseteq M$. 
The following two properties follow directly from Definition~\ref{def:val-classes}: 
\textbf{(XOS1)} $v_i(\vec{x}^*_i) = v^*_i(\vec{x}^*_i)$. \textbf{(XOS2)} For any allocation $\vec x_i \subseteq M$, $v_i(\vec{x}_i) \ge v^*_i(\vec{x}_i)$.%
\footnote{To see this, recall that the additive representatives $(v_{ij})_{j \in M}$ of agent $i$ with respect to $\vec{x}_i$ are chosen as maximizers from the class $\mathcal{L}_i \ni (v^*_{ij})_{j \in M}$, and thus $v_i(\vec{x}_i) = \sum_{j \in M} v_{ij} x_{ij} \ge \sum_{j \in M} v^*_{ij} x_{ij}$.}

Next, we define \emph{budget-capped valuations} that account for the budget constraints in \eqref{eq:budget}. For every $i \in N$, the \emph{$\budget_i$-capped valuation} $v_i^{\budget_i}: 2^M \mapsto \mathbb{R}_{\geq 0}$ is defined as $v_i^{\budget_i}(S)=\min(v_i(S), \budget_i)$ for all $S \subseteq M$.
A crucial observation that we use below is that capped XOS valuation functions remain XOS.
\begin{proposition}[Lemma C.7 in \citep{ST13}]\label{proposition:capping-xos}
    $v_i \in \mathcal{V}_{\xos} \Rightarrow v_i^{\budget_i} \in \mathcal{V}_{\xos}$.
\end{proposition}

We can now formalize the notion of \emph{budget-free proxy instances}. 

\begin{definition}[Budget-Free Proxy Instance]\label{def:proxy-instance}
    Given an instance $I=(\vec{r}, \vec{v}, \vec{\sens}, \vec{\budget})$ and a bid profile $\vec{B} \in \Delta$, 
    the \emph{budget-free proxy instance} of $I$ and $\B$ is defined as $\hat{I}(I, \vec{B})=(\vec{r}, \vec{v^{\budget_i}}, \vec{\hat{\sens}}(\vec{B}), \vec{\infty})$ with
\[
\hat{\sens}_i(\vec{B}) := 
\begin{cases}
    0 & \text{if } \budget_i < \EX[v_i(x_i(\vec{B}))], \\
    \sens_i & \text{otherwise}.
\end{cases}
\]
\end{definition}

Intuitively, the budget-free proxy instance $\hat{I}$ simply replaces the valuation function of each agent $i$ by its budget-capped counterpart and leaves $i$'s type intact, unless $i$'s valuation under $\B$ exceeds the budget $\vec \budget_i$. In the latter case, $i$'s type is mapped to $0$ and $i$ becomes a value maximizer instead. 
As we show below, it suffices to focus on these proxy instances to bound the price of anarchy of instances with budgets. 

Note that for every budget-free instance $I \in \mathcal{I}^{\infty}$ it holds that $\hat{I}(I, \vec{B}) = I$ for all $\vec{B} \in \Delta$. 
Moreover, it is not hard to show that the optimal solutions of an instance and each of its budget-free proxies coincide. 
\begin{restatable}{proposition}{propoptbudgetlessequivalence}
    \label{prop:opt-budgetless-equivalence}
    \apprefproof{app:sec:individual-lw-lb}
    Let $I=(\vec{r}, \vec{v}, \vec{\sens}, \vec{\budget})$ and $\vec{B} \in \Delta$. Then, $\opt(\hat{I}(I, \vec{B}))=\opt(I)$.
\end{restatable}
For the remainder of Section~\ref{sec:extension-theorem}, given a pair $(I, \vec{B})$, we use $(\hat{g}_i)_{i \in N}$ to denote the agents’ gain functions for the proxy instance $\hat{I}(I, \vec{B})$, i.e., for every $i \in N$ and every $\vec{b} \in D$, we define 
$
\hat{g}_i(\vec{b}) := v_i^{\budget_i}(x_i(\vec{b})) - \hat{\sens}_{i}(\vec{B}) p_i(\vec{b}).
$

\begin{restatable}{lemma}{lemmaindividuallwlb}\label{lemma:individual-lw-lb}
\apprefproof{app:sec:individual-lw-lb}
        Consider an instance $I \in \mathcal{I}_{\xos}$ and let $\vec{B} \in \cce(I)$. Then, for every agent $i \in N$, for every $\vec{B}_i'$ with $(\vec{B}_i', \vec{B}_{-i}) \in \mathcal{R}_i$ and every $\delta  \in [0,1]$ it holds
    \[
          \min\left(\EX[v_i(x_i(\vec{B}))], \budget_i \right) \geq \delta \cdot \EX\left[\hat{g}_i(\vec{B}_i', \vec{B}_{-i})\right] + (1 - \delta + \delta \hat{\sens}_{i}(\vec{B})) \cdot \pe \left[p_i(\vec{B})\right].
    \]
    
\end{restatable}

\subsubsection{Proof of Theorem \ref{theorem:template}}

The following Lifting Lemma for budget-free instances provides the final ingredient for the proof of our Extension Theorem. Basically, for any type $\type$, it lifts the smoothness property of $\fpa(r)$ for a single auction of type $\type$ (i.e., where the rightful winner is of type $\type$) to all auctions having the same type.

\begin{restatable}[Lifting Lemma]{lemma}{lemcomp}
    \label{lem:comp}
    \apprefproof{app:sec:ll}
    Let $T$ be a set of types. Consider an instance $I \in \mathcal{I}_{\xos}^{T, \infty}$ and let $\B \in \Delta$ be a well-supported bid profile.
    Let $(v^*_{ij})$ be some opt-induced additive representatives. Assume that $\fpa(r)$ is $(\lambda_{\type}, \mu_{\type})$-smooth for each type $\type \in T$.
   Then, there exists a bid $\vec{B}_i'$ for every $i \in N$ satisfying $(\vec{B}_{i}', \vec{B}_{-i}) \in \mathcal{R}_i$. Further, for each $\type \in T$, it holds 
\begin{equation}\label{eq:lambda-mu-per-groupMT}
\sum_{i \in N_{\type}(I)} \pe[g_i(\vec{B}_{i}', \vec{B}_{-i})] 
\ge 
\sum_{j \in M: \rw(j) \in N_{\type}(I)} \lambda_{\type} v^*_{\rw(j)j} - \mu_{\type} \pe[p_{\aw(j)j}(\vec{B})].
\end{equation}
\end{restatable}

We can now present the proof of our Extension Theorem.

\begin{proof}[Proof of Theorem~\ref{theorem:template}]
Consider an instance $I \in \mathcal{I}_{\xos}^T$ and let $\vec{B}$ be a well-supported CCE of $I$. 
Let $\vec{\delta} \in \mathcal{C}(\vec{\mu}, T^+)$ be an arbitrary calibration vector; note that such a vector exists by Lemma~\ref{lemma:optimal-xi} (given below). 
Let $\vec{x}^*(I)$ be an optimal allocation and $(v^*_{ij})$ the respective opt-induced additive representatives. 
Let $\hat{I}(I, \vec{B})$ be the budget-free proxy instance of $(I, \vec{B})$ as defined in Definition~\ref{def:proxy-instance}.
For ease of notation, we write $\hat{N}_t := N_t(\hat{I}(I, \vec{B}))$ for every $t \in T^+$.

The instance $\hat{I}$ is budget-free by construction. Further, since $\hat{I}$ is a proxy of $I \in \mathcal{I}_{\xos}$, it follows from Proposition \ref{proposition:capping-xos} that $\hat{I} \in \mathcal{I}_{\xos}^{\infty}$. 
Thus, we can apply Lemma~\ref{lem:comp} to $\hat{I}\in \mathcal{I}_{\xos}^{\infty}$ and $\vec{B}$, and obtain (from its first statement) that there exists a bid $\vec{B}_i'$ for each agent $i \in N$ such that $(\vec{B}_i', \vec{B}_{-i})$ satisfies the ROI constraint~\eqref{eq:ROI} for $\hat{I}$. Using this and the definition of budget-capped valuations, we obtain for each agent $i$:
\[
    \EX[p_i(\vec{B}_{i}', \vec{B}_{-i})] \leq \EX\left[v_i^{\budget_i}(x_i(\vec{B}_{i}', \vec{B}_{-i}))\right] = \EX\left[\min\left(v_i(x_i(\vec{B}_i', \vec{B}_{-i})), \budget_i\right)\right].
\]
In particular, this shows that each $(\vec{B}_i', \vec{B}_{-i})$ satisfies the ROI constraint~\eqref{eq:ROI} and the budget constraint~\eqref{eq:budget} for instance $I$, i.e., $(\vec{B}_i', \vec{B}_{-i}) \in \mathcal{R}_i$. 

By exploiting~\eqref{eq:lambda-mu-per-groupMT} for instance $\hat{I}$, we obtain that for every $t \in T^+$ we have
\begin{equation}
\sum_{i \in \hat{N}_{\type}} \pe\left[\hat{g}_i(\vec{B}_{i}', \vec{B}_{-i})\right] 
\ge 
\sum_{j \in M: \rw(j) \in \hat{N}_{\type}} \lambda_{\type} v^*_{\rw(j)j} - \mu_{\type} \pe\left[p_{\aw(j)j}(\vec{B})\right].
\end{equation}
Note that we exploit here that $\vec{x}^*(I)$ is also an optimal solution for $\hat{I} \in \mathcal{I}_{\xos}^{\infty}$ by Proposition~\ref{prop:opt-budgetless-equivalence}. Thus, the opt-induced additive representatives with respect to $\vec{x}^*(I)$ of $I$ and $\hat{I}$ are the same.

We therefore have:
\begin{align*}
    \sw(I, \vec{B}) 
    & \geq \sum_{\type \in T} \sum_{i \in \hat{N}_{\type}} \delta_{\type} \pe\left[\hat{g}_i(\vec{B}_i', \vec{B}_{-i})\right] + \left(1 - \delta_{\type} + \delta_{\type} \hat{\sens}_{i}(\vec{B})\right) \pe\left[p_i(\vec{B})\right] \\
    &= \sum_{\type \in T^+} \sum_{i \in \hat{N}_{\type}} \delta_{\type} \pe\left[\hat{g}_i(\vec{B}_i', \vec{B}_{-i})\right] + \left(1 - \delta_{\type} + \delta_{\type} t\right) \pe\left[p_i(\vec{B})\right] \\
    &\geq \sum_{\type \in T^+} \sum_{i \in \hat{N}_{\type}} \delta_{\type} \pe\left[\hat{g}_i(\vec{B}_i', \vec{B}_{-i})\right] + \left(1 - \max_{\type \in T^+} \left(\delta_{\type}(1 - t)\right)\right) \pe\left[p_i(\vec{B})\right] \\
    &= \sum_{\type \in T^+} \sum_{i \in \hat{N}_{\type}} \delta_{\type} \pe\left[\hat{g}_i(\vec{B}_i', \vec{B}_{-i})\right] + \left(1 - \max_{\type \in T^+} \left(\delta_{\type}(1 - t)\right)\right) \pe\left[\sum_{i \in N} p_i(\vec{B})\right]. \numberthis \label{eq:decomposition}
\end{align*}

The first inequality follows from applying Lemma~\ref{lemma:individual-lw-lb} to each agent $i \in N$, using the deviation $\vec{B}_i'$ given by Lemma~\ref{lem:comp}. Then, the first equality holds by the definition of $T^+$.

We now lower bound the first term in~\eqref{eq:decomposition}. Using Lemma~\ref{lem:comp}, we obtain:
\begin{align*}
\sum_{\type \in T} \sum_{i \in \hat{N}_{\type}} \delta_{\type} \pe\left[\hat{g}_i(\vec{B}_i', \vec{B}_{-i})\right]
&\ge \sum_{\type \in T^+} \sum_{{j \in M: \rw(j) \in \hat{N}_{\type}}} \delta_{\type} \lambda_{\type} v^*_{\rw(j)j} - \delta_{\type} \mu_{\type} \pe\left[p_{\aw(j)j}(\vec{B})\right] \\
&\ge \sum_{\type \in T^+} \sum_{{j \in M : \rw(j) \in \hat{N}_{\type}}} \min_{\type \in T^+} \left(\delta_{\type} \lambda_{\type}\right) v^*_{\rw(j)j} - \max_{\type \in T^+} \left(\delta_{\type} \mu_{\type}\right) \pe\left[p_{\aw(j)j}(\vec{B})\right] \\
&= \min_{\type \in T^+} \left(\delta_{\type} \lambda_{\type}\right) \opt(I) - \max_{\type \in T^+} \left(\delta_{\type} \mu_{\type}\right) \pe\left[\sum_{i \in N} p_i(\vec{B})\right], \numberthis \label{eq:smoothness-total}
\end{align*}
where the last equality follows from property \textbf{XOS1}.

Substituting~\eqref{eq:smoothness-total} into~\eqref{eq:decomposition}, we obtain:
\begin{align*}
\sw(I, \vec{B}) 
&\ge \min_{\type \in T^+} \left(\delta_{\type} \lambda_{\type}\right) \opt(I) 
+ \left(1 - \max_{\type \in T^+} \left(\delta_{\type} \mu_{\type}\right) - \max_{\type \in T^+} \left(\delta_{\type}(1 - t)\right)\right) \pe\left[\sum_{i \in N} p_i(\vec{B})\right] 
\ge \min_{\type \in T^+} \left(\delta_{\type} \lambda_{\type}\right) \opt(I),
\end{align*}
where the second inequality holds because $\vec{\delta} \in \mathcal{C}(\vec{\mu}, T^+)$ (see Definition~\ref{def:Xi}). By 
selecting a calibration vector $\vec{\delta} \in \mathcal{C}(\vec{\mu}, T^+)$ that maximizes $\min_{\type \in T^+} \left(\delta_{\type} \lambda_{\type}\right)$, we finally obtain
\begin{equation}\label{eq:final-ratio}
\sw(I, \vec{B}) \geq \max_{\vec{\delta} \in \mathcal{C}(\vec{\mu}, T^+)}\min_{\type \in T^+} \delta_{\type}\lambda_{\type} \opt(I).
\end{equation}

Since \eqref{eq:final-ratio} holds for every instance $I \in \mathcal{I}_{\xos}$ and every well-supported $\vec{B} \in \cce(I)$, the proof follows.
\end{proof}

Note that, in the proof above, the whole purpose of our calibration vector was to eventually lower bound the total payments in the final expression by $0$. This also explains the specific definition of the feasibility constraint of $\mathcal{C}(\vec{\mu}, T)$ in \eqref{equation:Xi-def}. 

We close this subsection with the following lemma stating the existence of calibration vectors. 

\begin{restatable}{lemma}{lemcalibrationcharacterization}\label{lemma:optimal-xi}
\apprefproof{proof:lemma:optimal-xi}
Let $T$ be a set of types, $\vec{\mu} = (\mu_{\type})_{\type \in T} \in \mathbb{R}_{>0}^{|T|}$, and $\vec{\lambda} = (\lambda_{\type})_{\type \in T} \in \mathbb{R}_{>0}^{|T|}$.
Then, 
\begin{equation}\label{eq:optimal-xi-lambda-mu}
    \max_{\vec{\delta} \in \mathcal{C}(\vec{\mu}, T)}\min_{\type \in T}\lambda_{\type}\delta_{\type}=O,
    \qquad\text{where}\qquad
    O = \min\left\{\min_{\type \in T} \lambda_{\type}, \left(\max_{\type \in T}\left(\frac{\mu_{\type}}{\lambda_{\type}}\right) 
    + \max_{\type \in T}\left(\frac{1-t}{\lambda_{\type}}\right) \right)^{-1}\right\}.
\end{equation}
\end{restatable}

\subsection{POA-Revealing Mathematical Program} \label{sec:POA-reveal}

We can now present our \emph{POA-revealing mathematical program (\poarmp)}, which facilitates bounding the price of anarchy as stated in our Extension Theorem (Theorem~\ref{theorem:template}). The program is parameterized by the set $T$ of available agent types. Recall that $\eta = \min_{j \in M} \eta_j \in [0,1)$.

\setlength{\fboxsep}{5pt}
\bigskip
\noindent
\!\!
\fbox{\parbox{.98\textwidth}{
\begin{alignat}{3}
\poarmp(T) \ = \ \max\ & \min \textstyle
        \bigg\{ 
            \min_{\type \in T} 
            \lambda_{\type}, 
            \Big( 
                \max_{\substack{\type \in T}} 
                && \textstyle
                \Big( \frac{\mu_{\type} }{\lambda_{\type}} \Big)    
             \, + \,
            \max_{\type \in T} 
                \Big (\frac{1 - \type }{\lambda_{\type}} \Big) 
            \Big)^{-1}         
        \bigg\} 
        && \notag
    \\
    \text{s.t.} \ \ \,
    & \textstyle\quad
    \lambda_{\type} = \mu_{\type}\left(1 - \frac{1 - {\eta}}{e^{1/{\mu_{\type}}}}\right)
    \quad 
    && \textstyle\ \ 
    \mu_{\type} > 0
    && \textstyle \ 
    \forall \type \in T \cap \set{1}
    \label{eq:key-lemma-0-constraint}   
    \\[-.5ex]    
    & \textstyle\quad
    \lambda_{\type} = \frac{\mu_{\type}}{\type}\left(1 - \frac{1 - {{\type}\eta}}{e^{{\type}/{\mu_{\type}}}}\right)
    \quad 
    && \textstyle\ \ 
    \mu_{\type} \ge 
    {\type}
    \left(    
    \ln\left(\frac{1-{{\type}\eta}}{1- {\type}} \right)\right)^{-1}   
    \quad
    && \textstyle \ 
    \forall \type \in T \cap (0,1)
    \label{eq:key-lemma-1-constraint}
    \\[.5ex]
    & \textstyle\quad
    \lambda_{\type} = \mu_{\type}
    \quad 
    && \textstyle\ \ 
    \mu_{\type} \in \left (0, {\frac{1}{1-\eta}} \right]
    \quad
    && \textstyle \ 
    \forall \type \in T \cap \set{0} 
    \qquad\quad
    \label{eq:key-lemma-2-constraint} 
\end{alignat}
}}

\bigskip

A crucial building block in deriving our mathematical program (\poarmp) is the characterization of optimal calibration vectors. 
By exploiting this characterization, the task of finding the best upper bound on the POA reduces to solving the above program. We summarize this result in the following theorem.

\begin{restatable}{theorem}{extensiontemplaterefined}
\label{theorem:templateII}
\apprefproof{app:templateII}
Let $\mathcal{I}_{\xos}^T$ be the class of instances with fractionally subadditive valuations and type set $T$. Assume that $\fpa(r)$ is $(\lambda_{\type}, \mu_{\type})$-smooth for each type $\type \in T^+$. Then, the price of anarchy of well-supported coarse correlated equilibria is upper bounded by $\poarmp(T^+)^{-1}$.
\end{restatable}

We use (\poarmp) in the next section to derive (tight) bounds on the POA of CCE.

\section{Liquid Welfare Guarantees Without Reserve Prices}
\label{sec:POA-without-reserve-prices}

In this section, we focus on simultaneous first price auctions without reserve prices, i.e., we assume that $\eta = 0$ in \eqref{eq:key-lemma-0-constraint}--\eqref{eq:key-lemma-2-constraint}. In Section~\ref{subsec:solving-poarmp}, we develop lower bounds on the optimal value of $\poarmp(T)$, which lead to POA upper bounds for various sets of types. Then, in Section~\ref{subsec:poa-bounds-no-reserve}, we present these liquid welfare guarantees upper bounds, and in multiple cases, complement our positive results with matching lower bounds.

\subsection{Bounding \poarmp$(T)$ by Partitioning Agent Types}\label{subsec:solving-poarmp}

In this section, we characterize a feasible solution to $\poarmp(T)$ for a given set of types $T$. 
The main technical challenge is to identify an analytical solution that yields strong POA upper bounds. 
Below, we describe a policy for defining a solution vector $\vec{\mu}$ for $\poarmp(T)$. 

Given a set of types $T$, the main idea is to partition them into two classes $H_{\omega}$ (high) and $L_{\omega}$ (low), where $\omega$ is a separation parameter.
We then define $\mu_{\type}$ depending on the class each $\type \in T$ belongs to. 
Intuitively, $H_{\omega}$ contains agent types that are structurally close to utility maximizers, while $L_{\omega}$ contains agent types that are structurally close to value maximizers.

\begin{definition}\label{def:mu:star}
Given $\omega \in (0,1)$ and a set of types $T$, define $H_{\omega}(T)= \left\{\type \in T \mid {\type} \geq \omega \right\}$ and $L_{\omega}(T) = \left\{\type \in T \mid {\type} < \omega \right\}$.
Define $\vec{\mu}^*(\omega, T) \in \mathbb{R}_{> 0}^{|T|}$ such that, for each $t \in T$,
\begin{align}
    \mu^*_{\type}\left(\omega, T\right) &=
    \begin{cases}
        \frac{{\type}}{-\ln\left(1-\omega\right)}, & \text{if } \type \in H_{\omega}(T), \\
        \frac{{\type}}{-\ln(1-{\type})}, & \text{if } \type \in L_{\omega}(T)\cap(0,1), \\
        1, & \text{if $\type \in L_{\omega}(T) \cap \{0\}$.}
    \end{cases}
   \label{eq:mu-star-params}
\end{align}
\end{definition}

The following corollary is easy to verify. We refer to Appendix \ref{app:identities} for more properties of $\mu^*(\omega, T)$.
\begin{corollary}\label{cor:mu:feasible}
Given $\omega \in (0,1)$ and a set of types $T$, $\vec{\mu}^*\left(\omega, T\right)$ is a feasible solution of $\poarmp(T)$.    
\end{corollary}

\begin{proof}
Clearly, $\vec{\mu}^*\left(\omega, T\right)$ satisfies \eqref{eq:key-lemma-0-constraint} for all utility maximizers $\type \in H_{\omega}(T) \cap \{1\}$. 
Also,  $\vec{\mu}^*\left(\omega, T\right)$ satisfies \eqref{eq:key-lemma-1-constraint} for each $\type \in H_{\omega}(T)\cap (0,1)$ because ${\type} \geq \omega$ and because the function $f(z)=-\ln(1-z)$ is non-negative and non-decreasing on $(0,1)$.
Further, it satisfies \eqref{eq:key-lemma-1-constraint} with equality for all types $\type \in L_{\omega} \cap (0,1)$.
Finally, \eqref{eq:key-lemma-2-constraint} holds for all value maximizing types $\type \in L_{\omega}\cap \{0\}$. 
\end{proof}

Lemma \ref{lemma:general-omega-lb} will be useful when proving some of the price anarchy bounds that follow. The proof is given in Appendix~\ref{app:lowerbound-poarmp-general}. 

\begin{restatable}{lemma}{generalomegalb}\label{lemma:general-omega-lb}\apprefproof{app:lowerbound-poarmp-general}
Let $T$ be a set of agent types. 
If $\max(T)>0$, then for every $\omega \in (0 ,\max(T)] \cap (0,1)$, 
    \[
       \poarmp(T) \geq \min \left\{\frac{\omega}{-\ln(1-\omega)}, \frac{\omega}{\omega + \max(T)} \right\}.
    \]  
    
\end{restatable}

\subsection{Price of Anarchy Bounds}\label{subsec:poa-bounds-no-reserve}

\subsubsection{Budget-Free Instances and Common Type}  We first investigate the price of anarchy for budget-free instances with XOS valuations, assuming agents have a single type  $t \in [0,1]$ i.e., the class $\mathcal{I}_{\xos}^{\{t\}, \infty}$. In Theorem~\ref{theorem:single-type:budgetfree}, we establish a liquid welfare guarantee for $\mathcal{I}_{\xos}^{\{t\}, \infty}$, which interpolates smoothly from $\nicefrac{e}{e-1} \approx 1.58$ when $t = 1$ (all agents are utility maximizers) to $2$ when $t=0$ (all agents are value maximizers). This interpolation is illustrated in Figure~\ref{fig:overview}(a).
\begin{restatable}{theorem}{singletypebudgetfree}
    \label{theorem:single-type:budgetfree}
Let $\mathcal{I}_{\xos}^{\{t\}, \infty}$ be the class of budget-free instances with fractionally subadditive valuations and a single type $t \in [0,1]$. Then,
\[
    \cce\textit{-}\poa\left(\mathcal{I}_{\xos}^{\{t\}, \infty}\right) \leq
   \begin{cases}\frac{e}{e-1}, & \text{if }\type \in \left[1-\nicefrac{1}{e},1\right], \\
    1-\frac{(1-\type)\ln(1-\type)}{\type}, & \text{if }\type \in \left(0, 1-\nicefrac{1}{e}\right),\\
    2, & \text{if }\type=0.\end{cases}
\]
\end{restatable}

\begin{proof}
    Set $\omega=1-\nicefrac{1}{e}$. Let $\mu:=\mu^*(\omega, \{t\})$ be as prescribed by Definition \ref{def:mu:star}. We distinguish three cases based on the value of $t$. 

    \medskip
    \noindent
    \textbf{Case 1:} $t=0$. In this case, $(\lambda, \mu)=(1,1)$ (by \eqref{eq:key-lemma-2-constraint}) and $ \poarmp(\{t\})\geq  \min \left(\lambda, \frac{\lambda}{\mu +1} \right)=\frac{1}{2}.$
    
    \smallskip
    \noindent
    \textbf{Case 2:} $t \in \left(0, 1-\nicefrac{1}{e}\right)$. Similarly, $(\lambda, \mu)=\left(t(-\ln(1-t))^{-1}, t(-\ln(1-t))^{-1} \right)$ (by \eqref{eq:key-lemma-1-constraint} and we obtain
    \[
        \poarmp(\{t\})\geq  \min \left(\lambda, \frac{\lambda}{\mu +1-{\type}}\right)=\frac{\lambda}{\mu+1-t}=\frac{t}{t-(1-t)\ln(1-t)}.
    \]
    Here, the first equality follows since $\mu =\frac{t}{-\ln(1-t)} \geq t$ holds for all $t \in \left(0, 1-\nicefrac{1}{e}\right)$.

    \noindent
    \textbf{Case 3:} $t \in \left[1-\nicefrac{1}{e}, 1\right]$. We have that $(\lambda, \mu)=\left(\omega(-\ln(1-\omega))^{-1}, t(-\ln(1-\omega))^{-1}\right)=\left(1-\nicefrac{1}{e}, t\right)$ (by \eqref{eq:key-lemma-1-constraint} for $t<1$ and \eqref{eq:key-lemma-0-constraint} for $t=1$). Furthermore,
    \[
        \poarmp(\{t\})\geq \min \left(\lambda, \frac{\lambda}{\mu +1 -t} \right)=\lambda=1-\frac{1}{e}.
    \]

    \smallskip
    \noindent
    Finally, by Theorem \ref{theorem:templateII}, we have that $\cce\textit{-}\poa(\mathcal{I}_{\xos}^{\{t\}, \infty}) \leq \left(\poarmp(\{t\})\right)^{-1}$. Combining this fact with the lower bounds on $\poarmp(\{t\})$ we obtained for each of the three cases above, the claim follows.
\end{proof}
Theorem~\ref{theorem:single-type:budgetfree} has a few important implications. Note that for $t= 1$, i.e., when agents are utility maximizers, the upper bound of Theorem~\ref{theorem:single-type:budgetfree} recovers the best possible price of anarchy bound of $\frac{e}{e-1}$ (due to ~\citet{ST13}). We show that the same bound holds as long as $\type \geq 1 - \nicefrac{1}{e}$. That is, somewhat surprisingly, the CCE-POA for this range of types does not get worse; in fact, this is true even for single-item first price auctions. 
The proof of Corollary \ref{corollary:high-common-type-no-reserve} is implied by the proof of Theorem~\ref{lemma:LB-BigTauSens-withReserve} (for $\eta=0$) in Section~\ref{sec:well-supported-eq}.
\begin{corollary}\label{corollary:high-common-type-no-reserve}
    Let $\mathcal{I}_{\add}^{\{t\}, \infty}$ be the class of budget-free instances with additive valuations and a single type $t \in [0,1]$. If $t \geq 1-\nicefrac{1}{e}$ then, $\cce\textit{-}\poa(\mathcal{I}_{\add}^{\{t\}, \infty}) \geq \frac{e}{e-1}$ and the bound holds even for single-item auctions.
\end{corollary}

Finally, for instances with value-maximizers only i.e., for $t=0$, Theorem \ref{theorem:single-type:budgetfree} extends the best possible upper bound of $2$ to coarse correlated equilibria and XOS valuation functions \citep{LMP23, deng2024}.

\subsubsection{Budget-Constrained Agents and Heterogeneous Types} In Theorem \ref{theorem:full-hybrid-no-reserve}, we make use of all the technical tools developed so far in this section to obtain a liquid welfare guarantee for coarse correlated equilibria and budget-constrained agents with XOS valuations for arbitrarily heterogeneous agent types. Recall that we denote by $W_0$ the principal branch of the Lambert $W$ function (see Appendix \ref{app:lambert}). Let $P:[0,1] \mapsto \mathbb{R}_{\geq 0}$ be defined as
\begin{equation}\label{eq:poa-main-ub}
    P(z)=
    \begin{cases}
        1 + \frac{z}{1+W_{0}\left(-e^{-z-1}\right)}, & \text{if } z > 1 + \frac{W_{0}(-2e^{-2})}{2}, \\
        2, & \text{otherwise.}
    \end{cases}.
\end{equation}

\begin{theorem} \label{theorem:full-hybrid-no-reserve}
Let $\mathcal{I}_{\xos}^T$ be the class of instances with fractionally subadditive valuations and type set $T\subseteq [0,1]$. Then, $  \text{\cce-\poa}\left(\mathcal{I}_{\xos}^T\right) \leq P(\max(T))$.
\end{theorem}

Theorem~\ref{theorem:full-hybrid-no-reserve} reveals an intriguing threshold phenomenon: the POA for a type set $T$ remains at most $2$ when $\max(T) < 0.797$, and increases from $2$ to $2.1885$ as $\max(T)$ approaches $1$ (see also Figure~\ref{fig:overview}(a)). Note that our liquid welfare guarantee unifies and generalizes two state-of-the-art POA bounds. Specifically, Theorem~\ref{theorem:full-hybrid-no-reserve} recovers the upper bound of $2.1885$ due to \citet{deng2024}, for budget-free instances with additive valuation functions and mixed Nash equilibria under the mixed-agent model (i.e., $T = \{0,1\}$). It also recovers the upper bound of $2$ due to \citet{LMZ24}, for budget-constrained value maximizers (i.e., $T=\{0\}$) with additive valuations. Theorem \ref{theorem:full-hybrid-no-reserve} generalizes both results to coarse correlated equilibria and budget-constrained agents with XOS valuations, while simultaneously extending the type set to the general model of \citet{aggarwalSurvey}.

We now prove Theorem~\ref{theorem:full-hybrid-no-reserve} using the following technical claim. Its proof can be found in Appendix~\ref{app:math-claims-theorem-no-reserve-sigma-max}.

\begin{restatable}{claim}{claimLWpropzero}\label{claim:LW:prop0}
\apprefproof{app:math-claims-theorem-no-reserve-sigma-max}
Let $f(z) = 1 + W_{0}\left(-e^{-z-1} \right)$. For every $z \in \left(1+ \nicefrac{W_{0}\left(-2e^{-2} \right)}{2}, 1\right]$, we have $f(z)<z$ and $f(z)+z=-\ln(1-f(z))$.
\end{restatable}

\begin{proof}[Proof of Theorem~\ref{theorem:full-hybrid-no-reserve}]
Observe that $\max(T^+)=\max(T \cup \{0\})=\max(T)$. We distinguish three cases, based on the value of $\max(T)$. 

\medskip
\noindent
\textbf{Case 1:} $\max(T) = 0$. In this case $T^+=T=\{0\}$. We have:
\begin{equation}\label{eq:2-bound-main-thm-all-valuemax}
    \poarmp(T^+)=\poarmp(\{0\})\geq \min \left\{ \lambda_{\type}, 
    \left( \frac{\mu_{\type}}{\lambda_{\type}} 
    + \frac{1}{\lambda_{\type}} 
    \right)^{-1} \right\} = \min\left\{1 , \left(1 + 1\right)^{-1}\right\}= \frac{1}{2}.
\end{equation}
Here, the second equality is due to \eqref{eq:key-lemma-2-constraint} and \eqref{eq:mu-star-params}.

\medskip
\noindent \textbf{Case 2:} $\max(T) \in \left(0, 1 + \nicefrac{W_{0}(-2e^{-2})}{2}\right]$. Set $\omega = \max(T)$. Then $\omega \in (0, \max(T)] \cap (0,1)$ (since $\omega=\max(T) < 1$), and we can invoke Lemma \ref{lemma:general-omega-lb} for $T^+$ with $\omega =\max(T)$.
Hence, we obtain that the value of the objective function of $\poarmp(T^+)$ is at least:
\begin{equation}\label{eq:2-bound-main-thm}
    \min \left\{\frac{\omega}{-\ln(1-\omega)}, \frac{\omega}{\omega + \max(T)} \right\} = \min\left\{\frac{\max(T)}{-\ln(1-\max(T))}, \frac{1}{2}\right\}= \frac{1}{2}.
\end{equation}
Here, the first and second equality follow  by choice of $\omega = \max(T)$ as $\frac{z}{-\ln(1-z)} \geq \frac{1}{2}$ for all $z \leq  1 + \nicefrac{W_{0}(-2e^{-2})}{2}$.

\medskip
\noindent \textbf{Case 3:} $\max(T)  \in  (1 + \nicefrac{W_{0}(-2e^{-2})}{2}, 1]$. 
Set $\omega =  1 + W_{0}\left(-e^{-\max(T) -1} \right)$. 
Using the first statement of Claim ~\ref{claim:LW:prop0}, we get that $\omega = f(\max(T)) < \max(T)$ whenever $\max(T) \in (1+ \nicefrac{W_{0}(-2e^{-2})}{2}, 1]$.     
Therefore, $\omega \in (0,\max(T)] \cap (0,1)$ as $\max(T) \le 1$ and, similarly to \textbf{Case 2}, we can invoke Lemma \ref{lemma:general-omega-lb} for $T^+$ with $\omega =  1 + W_{0}\left(-e^{-\max(T) -1} \right)$.
Hence, we obtain that the value of the objective function of $\poarmp(T^+)$ is at least:
\begin{align}
    \min \left\{\frac{\omega}{-\ln(1-\omega)}, \frac{\omega}{\omega + \max(T)} \right\} 
    = \frac{\omega}{\omega + \max(T)}=\frac{ 1 + W_{0}\left(-e^{-\max(T) -1} \right)}{ 1 + W_{0}\left(-e^{-\max(T) -1} \right) +\max(T)}.\label{eq:218-bound-main-thm}
\end{align}
Here, the first equality holds by the second statement of Claim~\ref{claim:LW:prop0} and the second equality follows by our choice of $\omega$. 
\medskip

\medskip
\noindent
We therefore conclude that
$   \cce\textsc{-}\poa\left(\mathcal{I}_{\xos}^T\right) \leq \left(\poarmp\left(T^+\right) \right)^{-1} \leq P(\max(T)).
$
The first inequality follows by Theorem \ref{theorem:templateII} and the second inequality by \eqref{eq:2-bound-main-thm-all-valuemax}, \eqref{eq:2-bound-main-thm} and \eqref{eq:218-bound-main-thm}.
\end{proof}
We now show that the liquid welfare guarantee of Theorem~\ref{theorem:full-hybrid-no-reserve} is best possible by providing two matching lower bounds for simple classes of instances. The lower bound of Theorem \ref{thm:LB-POA-budget-commontype} matches the upper bound of Theorem \ref{theorem:full-hybrid-no-reserve} for CCE, even when the agents have a single type $t \in [0,1]$ and additive valuation functions.
\begin{restatable}{theorem}{LBPOAbudgetcommontype}
    \label{thm:LB-POA-budget-commontype}
    \apprefproof{proof:lb-POA-budget-commontype}
Let $\mathcal{I}_{\add}^{\{t\}}$ be the class of instances with additive valuations and a single type $t \in [0,1]$. Then, $ \text{\cce-\poa}(\mathcal{I}_{\add}^{\{\type\}}) \geq P(t)$.
\end{restatable}
Theorem \ref{thm:LB-POA-budget-commontype} has one additional implication for the inefficiency of CCE of simultaneous first price auctions and fractionally subadditive valuations. By combining Theorem \ref{thm:LB-POA-budget-commontype} and Theorem \ref{theorem:single-type:budgetfree}, we obtain a \emph{separation} in terms of liquid welfare guarantees for environments with a single type $t>0$. Namely, we show that budget-constraints make the price of anarchy \emph{strictly worse} in this case (see Figure \ref{fig:overview}(a)).
\begin{corollary}
    For every type $t \in T$, let $\mathcal{I}_{\xos}^{\{t\}}$ be the class of instances with fractionally subadditive valuations and let $\mathcal{I}_{\xos}^{\{t\},\infty}$ be the  subclass of budget-free instances. For $t > 0$, it holds that $\cce\textsc{-}\poa(\mathcal{I}_{\xos}^{\{t\}}) > \cce\textsc{-}\poa(\mathcal{I}_{\xos}^{\{t\}, \infty})$.
\end{corollary}
We continue with our second negative result. Here, we show that the liquid welfare guarantee in Theorem \ref{theorem:full-hybrid-no-reserve} is also tight for the class of budget-free instances, additive valuations and budget-free instances, even for mixed Nash equilibria, as long as value maximizers are included in the type set $T$. This lower bound is a generalization of the lower bound presented by \cite{deng2024} for a mixed agent model with value maximizers and a different type.

\begin{restatable}{theorem}{LBPOAbudgetfreehybrid}
\label{thm:LB-POA-budgetfree-hybrid}
\apprefproof{proof:LB-POA-budgetfree-hybrid}
    Let $\mathcal{I}_{\add}^{\{0, t\}, \infty}$ be the class of budget-free instances with additive valuations and a set of agent types $\{0, t\}$ for $t \in (0,1]$. Then, $\mne{\textit{-}}\poa(\mathcal{I}_{\add}^{\{0, t\}, \infty}) \geq P(t)$. 
\end{restatable}

Note that Theorem \ref{thm:LB-POA-budgetfree-hybrid} together with Theorem \ref{theorem:full-hybrid-no-reserve} settle the price of anarchy of all equilibrium classes and all budget-constrained instances for agent type sets that include the value-maximizing type $t=0$.

\begin{corollary}\label{cor:poa-hybrid-budget-equ}
    For $\val \in \{\add,\sub,\xos\}$, let $\mathcal I_{\val}^T$ be the class of instances with valuations in $\mathcal V_{\val}$  and let $T$ be a set of types. If $0\in T$, then, for $\equ\in \{\mne,\ce,\cce\}$ we have
    $\equ\textit{-}\poa(\mathcal I^T_{\val}) =P(\max(T))$.
\end{corollary}
We conclude this section with Theorem~\ref{lem:universal-LB-budget}, which states a slightly weaker POA lower bound for MNE with budget-constrained agents and additive valuation functions. However, this bound holds for an arbitrary type set $T$ (not necessarily including $t = 0$).
\begin{restatable}{theorem}{lemuniversalLBbudget}
    \label{lem:universal-LB-budget}
    \apprefproof{app:LB-no-reserve-prices}
    For every type set $T$, it holds that $\mne\textit{-}\poa(\mathcal I^T_{\add})\geq 2$.
\end{restatable}
Closing the gap between $2$ and $P(\max(T))$ for MNE, for every type set $T$, is an intriguing open question. 

\subsubsection{Bounded Minimum Type} A qualitative interpretation of the upper bound in Theorem~\ref{theorem:full-hybrid-no-reserve} is that equilibria become more inefficient in the presence of agent types whose types resemble utility maximizers—i.e., as $t$ approaches $1$, the POA worsens. One contributing factor to this phenomenon is that, in the worst case, the competitors of these agents may, on the constrast, be structurally similar to value maximizers. 
A natural question is to study mixtures of types in which the minimum type is bounded away from $0$. In Theorem~\ref{theorem:not-so-hybrid}, we present a second threshold phenomenon revealed by our framework: whenever a type set $T$ satisfies $\min(T) \geq 0.74$, the liquid welfare guarantees for CCE of budget-free instances improve, no matter how heterogeneous the set of types is. In fact, for such type sets, $\cce\text{-}\poa(\mathcal{I}_{\xos}^{T, \infty}) \leq 1.83$! 
\begin{restatable}{theorem}{notsohybrid}\label{theorem:not-so-hybrid}
\apprefproof{app:not-so-hybrid}
    Let $\beta$ be the solution to $\beta = 1 - e^{-\frac{1}{\beta}}$, i.e., $\beta \approx 0.741$. Let $\mathcal{I}_{\xos}^{T, \infty}$ be the class of budget-free instances with fractionally subadditive valuations and type set $T$. If $\min(T) \geq \beta$, then:
\begin{align*}
     \text{\cce-\poa}(\mathcal{I}_{\xos}^{T, \infty} ) &\leq \left(\min(T) \left(1 -e^{-\frac{1}{\min(T)}}\right) \right)^{-1} \in \left[\frac{e}{e-1}, \frac{1}{\beta^2}\right]
\end{align*}
\end{restatable}

\section{Improved Liquid Welfare Guarantees with Reserve Prices}\label{sec:reserve-prices}
Our type-dependent smoothness framework and the \poarmp\ developed in Section~\ref{sec:smoothness:extension} allow us to study the inefficiency of instances with fractionally subadditive valuations and their induced well-supported equilibria when the auctioneer implements reserve prices for each individual auction (see also Section~\ref{sec:preliminariesPOA}). 
Namely, as already hinted by our formulation of $\poarmp(T)$ in Section~\ref{sec:POA-reveal}, our liquid welfare guarantees for instances with reserve prices depend on the minimum relative gap $\eta \in [0,1)$.

We remark that, as observed by \citet{BDM21}, a reserve price $r_{j}$ can be interpreted as a \emph{prediction} (see also \citep{gkatzelis22, chen04, deng24individual}) of the value of the rightful winner of an auction $j \in M$. These predictions can then be used to set reserve prices accordingly, with the goal of improving liquid welfare guarantees. 
In this context, the parameter $\eta$ can be viewed as an error measure of the prediction: $\eta = 0$ indicates a completely uninformative prediction,\footnote{Uninformative for improving liquid welfare guarantees: almost all $\eta_j$'s may be close to $1$, but one auction $j$ with $\eta_j = 0$ can prevent improved guarantees (see also the proof of Theorem \ref{lemma:LB-POA-Reserve-value-maximizers}).} whereas as $\eta \rightarrow 1$, the reserve prices approach the actual valuations of the rightful winners in all auctions. 

Section~\ref{sec:reserve-prices} is structured as follows. 
In Section~\ref{sec:poa-bounds-reserve-prices}, we devise upper bounds on the POA that are functions of $\eta$ for well-supported equilibria of a given class. Then, in Section~\ref{sec:well-supported-eq}, we examine when such equilibria are guaranteed to exist. 
Finally, in Section \ref{sec:Learning}, we demonstrate that when budget-free agents repeatedly participate in a first-price auction with reserve price using regret-minimizing algorithms, they converge to CCE that are well-supported.

\subsection{POA Bounds as Functions of the Minimum Relative Gap}\label{sec:poa-bounds-reserve-prices}

In the presence of the $\eta$ parameter, deriving analytical POA bounds for a general type set $T$ becomes significantly more challenging. 
In order to obtain upper bounds on the POA as functions of $\eta$, we can no longer rely on our approach from Section~\ref{subsec:solving-poarmp}. 
Therefore, in this section, we focus on more tractable settings, such as the budget-free single-type environment and the mixed-agent model with budget-constrained agents.

\subsubsection{Budget-Free Agents with One Type}

We present our POA upper bound for the single type environment $\{t\}$ for every $t\in[0,1]$, fractionally subadditive valuations and well-supported CCE in Theorem \ref{theorem:single-type}. As $\eta \to 1$, the liquid welfare of all such CCE tends to optimality; see Figure \ref{fig:overview-with-eta} for an illustration.

\begin{restatable}{theorem}{theoremsingletype}
    \label{theorem:single-type}
    \apprefproof{proof:theorem:single-type}
Let $\mathcal{I}_{\xos}^{\{t\}, \infty}$ be the class of budget-free instances with fractionally subadditive valuations, a single type $t \in [0,1]$, and let $\eta \in [0,1)$ be the smallest relative gap of the reserve prices. 
Also, let:
\begin{equation}\label{eq:poa-expression-commono-reserves}
     P_{t}(\eta) = \begin{cases}
        \frac{e}{e-1 +  t \eta}, & \text{ if } t \in (1-\nicefrac{1}{e},1] \text{ and } \eta \in \left[0, \frac{1-e(1-t)}{t} \right), \\[2pt]        
        1 + \frac{1}{t} \left( \ln\left(\frac{1- t \eta}{1- t}\right)(1 -t) \right), 
        & \text{ if } t \in (0,1)   \text{ and } \eta \in \left(\max\left(0, \frac{1-e(1-t)}{t} \right), 1\right) ,\\[2pt]
        2-\eta, &\text{ if } t = 0.
    \end{cases}
\end{equation}
For well-supported coarse correlated equilibria, it holds that $\cce\textit{-}\poa(\mathcal{I}_{\xos}^{\{t\}, \infty}) \leq P_{t}(\eta)$. Furthermore, $P_{t}$ is non increasing in $[0,1)$ with $\lim_{z \to 1}P_{t}(z)=1$ for every $\type \in [0,1]$.
\end{restatable}
We now present two lower bounds for CCE in budget-free instances with reserve prices in single-type environments. In Theorem~\ref{lemma:LB-POA-Reserve-value-maximizers}, we show that the bound of Theorem~\ref{theorem:single-type} is tight for all $\eta \in [0,1)$ in auctions with value maximizers only, i.e., when $t = 0$. 
Then, in Theorem~\ref{lemma:LB-BigTauSens-withReserve}, we establish a negative result for a different restricted range of $(\eta, t)$, including the case of utility maximizers ($t = 1$) for which we prove that the bound of Theorem~\ref{theorem:single-type} is tight for all $\eta \in [0,1)$. Interestingly, the worst-case instances used in the proof of these theorems involve a two-item and single-item auction. 

\begin{restatable}{theorem}{LBPOAReservevaluemaximizers}
    \label{lemma:LB-POA-Reserve-value-maximizers}
    \apprefproof{proof:LB-POA-Reserve-value-maximizers}
Let $\mathcal{I}_{\add}^{t, \infty}$ be the class of budget-free instances with  additive valuations for type $t=0$ and let $\eta \in [0,1)$ be the smallest relative gap from reserve prices in $\mathcal{I}_{\add}^{\{0\}, \infty}$. Then, $\cce\textit{-}\poa(\mathcal{I}_{\add}^{\{0\}, \infty}) \geq 2- \eta$ for  well-supported coarse correlated equilibria.
\end{restatable}

\begin{restatable}{theorem}{LBBigTauSenswithReserve}
    \label{lemma:LB-BigTauSens-withReserve}
    \apprefproof{proof:LB-BigTauSens-withReserve}
Let $\mathcal{I}_{\add}^{\{t\}, \infty}$ be the class of budget-free instances with additive valuations for type $t \in [1-\nicefrac{1}{e},1]$ and let $\eta \in [0,1)$ be the smallest relative gap from reserve prices. Then, if $\eta \leq \frac{1-(e-1)t}{t}$,  $\cce\textit{-}\poa(\mathcal{I}_{\add}^{\{0\}, \infty}) \geq \frac{e}{e-1+\eta}$ for well-supported coarse correlated equilibria.
\end{restatable}

\smallskip
\begin{corollary}\label{corollary:tight-single-type-reserves}
Let $\eta \in [0,1)$ be the smallest relative gap from reserve prices.
Then, for $t \in \{0,1\}$ it holds that
$ \cce\textsc{-}\poa(\mathcal{I}_{\xos}^{\{t\}, \infty})=\cce\textsc{-}\poa(\mathcal{I}_{\add}^{\{t\}, \infty})=P_t(\eta)$. 
\end{corollary}

\subsubsection{Budget-Constrained Agents in the Mixed Agent Model} 
We now focus on liquid welfare guarantees for $T = \{0,1\}$, i.e., the mixed-agent model in which agents are either utility or value maximizers, under budget constraints and with reserve prices. This setting has previously been studied by \citet{deng2024} for budget-free instances, mixed Nash equilibria, and additive valuations. In Theorem~\ref{thm:hybrid-reserve-prices}, we establish a POA upper bound that depends on the parameter~$\eta$ for budget-constrained agents, CCE, and fractionally subadditive valuations. As the parameter~$\eta$ increases from $0$ to $1$, the resulting bound interpolates between our liquid welfare guarantee of $2.1885$ from Theorem~\ref{theorem:full-hybrid-no-reserve} and the optimal value of~$1$.

\begin{restatable}{theorem}{thmhybridreserveprices}
    \label{thm:hybrid-reserve-prices}
    \apprefproof{app:thm:hybrid-reserve-prices}
Let $\mathcal{I}_{\xos}^{\{0,1\}, \infty}$ be the class of instances with fractionally subadditive valuations for the set of types $\{0, 1\}$. Let $\eta \in [0,1)$ be the smallest relative gap from reserve prices. Then, for well-supported equilibria, $\cce\textit{-}\poa(\mathcal{I}_{\xos}^{\{0,1\}}) \leq Q(\eta)$, where
$$
    Q(\eta) = \left(1-\eta\right) \cdot \frac{2-\eta + W_{0}\left(-(1-\eta)^2e^{\eta-2}\right)}{1-\eta + W_{0}\left(-(1-\eta)^2e^{\eta-2}\right)}.
$$
Furthermore, $Q(\eta)$ is non-increasing in $[0,1)$ with $Q(0)\approx 2.1885$ and $\lim_{z \to 1}Q(z)=1$.
\end{restatable}

\subsection{On the Existence of Well-Supported Equilibria} \label{sec:well-supported-eq}

Recall that for an instance $I$ with reserve prices $\vec{r}$, a bid profile $\B \in \Delta$ is well-supported if, for each $\b \in \supp(\B)$, it holds that $\vec{x}_j(\b) \neq \vec{0}$ for every item $j \in M$. 
Throughout our work, we have applied this refinement on equilibria when considering liquid welfare guarantees for instances with reserve prices, as, in the presence of reserve prices, it is a crucial precondition of our Extension Theorem (Theorem \ref{theorem:templateII}) in Section~\ref{sec:smoothness:extension}. In this section, we explore when such equilibria are guaranteed to exist.

\paragraph{Budget-Free Instances and Additive Valuations.} We begin with a positive result for the class $\mathcal{I}_{\add}^{T,\infty}$ given a type set $T$: namely, in Theorem~\ref{thm:CE-well-supported}, we show that all \emph{correlated} equilibria induced on instances of this class are well-supported.

\begin{restatable}{theorem}{thmCEwellsupported}
    \label{thm:CE-well-supported}
    \apprefproof{app:thm:CE-well-supported}
    Let $T$ be a set of agent types and let $I \in \mathcal{I}_{\add}^{T, \infty}$. Then, every $\vec{B} \in \ce(I)$ is well-supported. 
\end{restatable}

Intuitively, for agents with additive valuations, the rightful winner $i$ of an auction $j$ always has ``room'' for additional gain by bidding above the reserve price $r_j$ and competing for the item (unless of course some other agent has already submitted a sufficiently high bid above $r_j$). In Theorem~\ref{thm:CE-well-supported}, we confirm this intuition for correlated equilibria. Interestingly, we observe that it does not necessarily hold for coarse correlated equilibria; in Theorem~\ref{prop:CCEitemNotAlwaysSold}, we construct a simple single-item auction with two utility maximizers and a feasible reserve price, along with a CCE that is not well-supported.
\begin{restatable}{theorem}{propCCEitemNotAlwaysSold}
    \label{prop:CCEitemNotAlwaysSold}
    \apprefproof{proof:prop:CCEitemNotAlwaysSold}
Let $\mathcal{I}_{\add}$ be the class of budget-free instances with additive valuation functions and reserve prices. There exists an instance $I \in \mathcal{I}_{\add}$ and a $\vec{B} \in \cce(I)$ which is not well-supported for $I$.
\end{restatable}
The intuition behind why such CCE exist even in simple settings is that it can be more cost-effective for the highest-valued agent to remain \emph{coordinated} in their bidding (through the CCE probability distribution) which might mean that they sometimes bid below the reserve price with positive probability, rather than deviating unilaterally and always bidding at least the reserve price which can lead to a higher expected payment. Note that the instance in Theorem~\ref{prop:CCEitemNotAlwaysSold} also implies a lower bound on the POA of $\frac{e}{e-1}$ for CCE that are not necessarily well-supported. We therefore conclude that our inability to devise improved liquid welfare guarantees as parameters of $\eta \in [0,1)$ for such equilibria (similar to those of Theorem \ref{theorem:single-type} and Theorem \ref{thm:hybrid-reserve-prices}) is not an artifact of our analysis (e.g., in Theorem~\ref{theorem:template}), but rather an unavoidable structural property of CCE.

\begin{restatable}{corollary}{corollaryLBwithReservenotallitemssold}
\label{corollary:LB-withReserve-not-all-items-sold}
\apprefproof{proof:corollary:LB-withReserve-not-all-items-sold}
Let $\mathcal{I_{\add}}$ be the class of instances with additive valuations and let $\eta \in [0,1)$ be the smallest relative gap from reserve prices. Then, $\cce\textit{-}\poa(\mathcal{I}_{\add}) \geq \frac{e}{e-1}$.    
\end{restatable}

\paragraph{Beyond Additive Valuations.} We conclude the section with another negative result. We show that, for the class of budget-free instances with submodular valuations, even \emph{mixed Nash equilibria} are not guaranteed to be well-supported. Moreover, this fact paints a bleak picture for their liquid welfare guarantees: without refining the set of MNE to well-supported equilibria only, the POA for MNE for instances with reserve prices is unbounded.

\begin{restatable}{theorem}{propPNEsubmodnotwellsupported} \label{prop:PNE-submod-not-well-supported}
\apprefproof{proof:prop:PNE-submod-not-well-supported}
Let $\mathcal{I}_{\sub}$ be the class of budget-free instances with submodular valuation functions and reserve prices. There exists an instance $I \in \mathcal{I}_{\sub}$ and a $\vec{B} \in \mne(I)$ which is not well-supported for $I$. Furthermore, $\poa\textit{-}\mne(\mathcal{I}_{\sub})=\infty$.
\end{restatable}

\subsection{Mean-Based Algorithms Converge to Well-Supported CCE} \label{sec:Learning}

It is well known that regret-minimization dynamics in auctions lead to CCE (see, e.g., \citep{hannan57, blum07, young04}). In this section, motivated by the negative result in Theorem~\ref{prop:CCEitemNotAlwaysSold}, which shows that there exist CCE that are not well-supported even in single-item first-price auctions with agents with no budget constraints, we address the question of whether the CCE reached through such dynamics are well-supported.

We consider a model similar to the repeated auction setting studied by \citet{kolumbus22}. 
Specifically, we consider two budget-free agents\footnote{We focus on two agents for simplicity, though the result extends to more than two agents.} with arbitrary types who repeatedly participate in a first-price auction with a feasible reserve price $r$, where ties are broken uniformly at random. Each agent $i$ is assumed not to overbid, i.e., $b_i \le v_i$ (as we can still assume w.l.o.g.\ that $\tau = 1$). The agents' values and types remain fixed across all repetitions of the auction. We assume that values, bids, and the reserve price $r$ are all integer multiples of a minimum increment $\varepsilon > 0$. Each agent aims to maximize their cumulative gain $\z_i$ over time.

Agents use regret-minimization algorithms, where the regret of agent $i$ after $T$ rounds, given bids $(\b_1, \dots, \b_T)$, is defined as
\[
R_i^{T} = \sum_{t=1}^{T} \max_{b} \z_i(b, \b_{-i}^{t}) - \z_i(b_{i}^{t}, \b_{-i}^{t}),
\]
with $b$ denoting the optimal fixed bid in hindsight. We introduce the definition of a weakly domimated action below.

\begin{definition}
Let $A_1$ and $A_2$ be subsets of the action spaces of two agents in a two-agent game. An action $i \in A_1$ of agent 1 is \emph{weakly dominated in $A_2$} by another action $i' \in A_1$ if:
\begin{enumerate}
    \item $\forall j \in A_2: \z_1(i, j) \le \z_1(i', j)$, and
    \item $\exists j \in A_2: \z_1(i, j) < \z_1(i', j)$.
\end{enumerate}
\end{definition}

\citet{kolumbus22} introduce a specific subclass of CCE called \emph{co-undominated}. In such equilibria, no action in an agent’s support is weakly dominated relative to the other agent's support.

\begin{definition}[\citet{kolumbus22}]
Let $\B$ be a CCE of a (finite) two-agent game with action spaces $I_1$ and $I_2$.
Let its support be $(A_1, A_2)$, where $A_1 = \{ i \in I_1 \mid \exists j \in I_2 \text{ such that } B_{ij} > 0 \}$ and $A_2 = \{ j \in I_2 \mid \exists i \in I_1 \text{ such that } B_{ij} > 0 \}$.
The CCE is \emph{co-undominated} if, for every $i \in A_1$ and every $i' \in I_1$, action $i$ is not weakly dominated in $A_2$ by $i'$, and similarly for $A_2$.
\end{definition}

If regret-minimizing algorithms\footnote{In particular, \cite{kolumbus22} focus on a family of algorithms called \emph{mean-based}; see \citep{braverman2018} for a definition.} converge, \cite{kolumbus22} show that they converge to co-undominated CCE. This convergence result also holds in our setting of a first-price auction with a feasible reserve price and heterogeneous agent types. Moreover, co-undominated CCE possess the desired property of being well-supported, as we show in the next theorem.

\begin{theorem} \label{lem:itemSoldCoUnCCE}
Consider a first-price auction with feasible reserve price among two agents with arbitrary types and action spaces $I_1$ and $I_2$. Then, any co-undominated CCE is well-supported. 
\end{theorem}

\begin{proof}
Let $\B$ be a co-undominated CCE of a $\fpar{}{r}$ with a feasible reserve price among two agents with arbitrary types and action spaces $I_1$ and $I_2$. 
For contradiction, assume that $\B$ is not well-supported, i.e., the item is not sold with probability 1. 
Denote the support $(A_1, A_2)$ of $\B$ by $A_1 = \{ i \in I_1 \mid \exists j \in I_2 \text{ such that } B_{ij} > 0\}$ and $A_2 = \{j \in I_2 \mid \exists i \in I_1 \text{ such that } B_{ij} > 0\}$. 
Assume w.l.o.g.\ that $v_1 \ge v_2$.

Since the item is not sold with probability 1, there exist actions $i \in A_1$ and $j \in A_2$ such that $i < r$, $j < r$, and $B_{ij} > 0$. 
In this case, we observe that action $i$ of agent 1 is weakly dominated in $A_2$ by action $r$, since for all $j' \in A_2$ it holds that $0 = \z_1(i, j') \le \z_1(r, j')$, and $0 = \z_1(i, j) < \z_1(r, j) = v_1 - \sens_1 r$, as $r < v_1$ and $\sens_1 \le 1$. However, by assumption, $i \in A_1$ i.e., it is in the support of the CCE $\vec{B}$. This contradicts Definition \ref{lem:itemSoldCoUnCCE} and the claim follows.
\end{proof}

Theorem~\ref{lem:itemSoldCoUnCCE} indicates that, when autobidding agents converge to a CCE using regret-minimizing algorithms, they reach a well-supported CCE. Hence, for each such CCE, the liquid welfare guarantees we obtained in Section~\ref{sec:poa-bounds-reserve-prices} apply.

\section*{Acknowledgements}
S. Klumper was supported by the Dutch Research Council (NWO) through its Open Technology Program, proj.~no.~18938. A. Tsikiridis, who was at CWI during the time of this work,  was also partially supported by NWO through the Gravitation Project NETWORKS, grant no.~024.002.003 and by the European Union under the EU Horizon 2020 Research and Innovation Program, Marie Sk\l{}odowska-Curie Grant Agreement, grant no.~101034253.

\bibliographystyle{ACM-Reference-Format}
\bibliography{soda}

\section*{Appendix}
\appendix

\section{Equilibrium Notions} \label{sec:appequilibria}
\subsection{Hierarchy of Equilibrium Notions} 

To see that $\mne \subseteq \ce$, consider a \mne\ $\vec{B} = \prod_{i \in [n]} \vec{B}_i$. Note that $\B$ also satisfies the \ce\ requirements:
\[
\pe [\pe [\z_i( \swap(\vec{b}_i), \vec{B}_{-i}) \mid \vec{B}_{i} = \vec{b}_i] ] = \pe [\z_i(\vec{B}'_i, \vec{B}_{-i})] \le \pe [\z_i(\vec{B})], 
\]
where the equality holds as $\B_i$ and $\B_{-i}$ are independent random bid profiles and the swapping function $\swap$ can be written as a randomized bid vector $\B'_i$ in this case. The randomized deviations of $i$ that satisfy the ROI and budget constraints also coincide for both equilibrium notions by the same reasoning. 

To see that $\ce \subseteq \cce$, for contradiction, consider a \ce\ $\B$ that is not a \cce. 
Then, there is an agent $i \in N$ with randomized bid profile $\B'_i$ satisfying the ROI constraint as in \eqref{eq:ROI} and budget constraint as in \eqref{eq:budget} such that $\pe [\z_i(\vec{B})] < \pe [\z_i(\vec{B}'_i, \vec{B}_{-i})]$. However, this would contradict that $\B$ is a \ce, as a swapping function $\swap$ for agent $i$ that maps all bid profiles to the randomized bid vector $\B'_i$ is equivalent to playing $\B'_i$ independent of $\B_{-i}$.  
This contradicts that $\B$ is a \ce, as deviating to $\B'_i$ would strictly improve the gain of agent $i$ and $\B'_i$ satisfies the ROI and budget constraints as it coincides with the satisfied ROI and budget constraints for \cce\ in this case.

\subsection{Equilibrium Notions under ROI} 

As mentioned in Section \ref{sec:preliminariesPOA}, it is not true that a \pne\ w.r.t. deterministic deviations survives when randomization is allowed due to the ROI constraints. 
Therefore, given the following definition of a \pne\ for our setting, it is not true that $\pne \subseteq \mne$. We illustrate this with an example after introducing the definition of a \pne.

\begin{definition} \label{def:pureNash}
Let $\vec{b}$ be a deterministic bid profile satisfying the ROI and budget constraints of each agent, i.e., $\vec{b} \in \mathcal{R}_i$ for all $i \in N$. Then, $\vec{b}$ is a \emph{\pne} if for every agent $i \in N$ it holds that:
\[\z_i(\vec{b}) \ge \z_i(\vec{b}'_i, \vec{b}_{-i}) \quad \forall (\vec{b}'_i, \vec{b}_{-i}) \in {\mathcal{R}_i}.\]
\end{definition}

\begin{example} \label{ex:PNEdefinition}
Consider a budget-free \fpar{2}{} with items $i$ and $j$ among 2 agents with additive valuations. In case of ties, they are always broken in favor of agent 1.
Let $\vec{v}_1 = (v_{1i}, v_{1j}) = (0.5, 0.3)$ and let $\vec{v}_2 = (v_{2i}, v_{2j}) = (0.7, 0)$.
Further, let $\sens_1 = \sens_2 = 0$ and recall that we can assume w.l.o.g. that $\tau_1 = \tau_2 =1$. 
Consider the deterministic bid profile $\b = ((0, 0.2),(0.7, 0.2))$.

It is not hard to see that $\b$ is a \pne. Agent 2 always wins item $i$ and never wins item $j$, which it has 0 value for, and agent 2 satisfies the ROI constraint.
Agent 1 always wins item $j$ and never wins item $i$, and satisfies the ROI constraint. However, agent 1 does have a positive value for item $i$, but agent 1 must at least bid 0.7 to win item $i$, which would violate the ROI constraint, regardless of agent 1 winning item $j$ or not.

If we consider randomized deviations, we will see that $\b$ is not a \mne.
Namely, consider the randomized bid profile $\B'_1$ for agent 1 which draws both bid profiles $\b^{1}_1 = (0, 0.2)$ and $\b^{2}_1 = (0.7, 0.2)$ with probability $\frac{1}{2}$. 
This leads to an expected gain for agent 1 of $\frac{1}{2} \cdot 0.3 + \frac{1}{2} (0.5 + 0.3) = 0.55$, which is strictly better than 0.3, which is the gain of agent 1 under $\b$. 
Further, $\B'_1$ satisfies the ROI constraint as $\frac{1}{2} \cdot 0.2 + \frac{1}{2} (0.7 + 0.2) = 0.55$.
\end{example}

What Example \ref{ex:PNEdefinition} also shows is that, given some bid profile, it might be beneficial for an agent $i$ to sometimes pay more than $v_{ij}$ for an item $j$ if this still leads to a positive gain, and compensate this `overpaying' by `underpaying' for another item so that the ROI constraint is still satisfied. 
To also account for this behavior in the definition of a \ce, the ROI and budget constraints for a deviation is evaluated for the entire support, and not only for the conditioned bid profile. 
Further, a deviation may require multiple recommendations to be mapped to a different strategy, i.e., the swapping function is $\swap: \supp_{i}(\vec{B}) \rightarrow \Delta_i$. And finally, the gain of a deviation of an agent $i$ is evaluated for the entire support, not only for the conditioned bid profile. The following example illustrates why these three elements are relevant. 

\begin{example} \label{ex:CEdefinition} 
Consider a budget-free \fpar{2}{} with items $i$ and $j$ among 2 agents with additive valuations. In case of ties, they are broken in favor of agent 1. Further, let $\sens_1 = \sens_2 = 0$ (and $\tau_1 = \tau_2 =1$).

First, consider the valuation $\vec{v}$ with $\vec{v}_1 = (v_{1i}, v_{1j}) = (0.2, 0.5)$ and $\vec{v}_2 = (v_{2i}, v_{2j}) = (0.2,0)$.
Consider $\B$ for which bid profiles $\b^1 = ((0, 0.7),( 0.2, 0.7))$ and $\b^2 = ((0.2, 0.3), (0.2, 0.3))$ are both drawn with probability $\frac{1}{2}$. Note that both agents satisfy the ROI constraint under $\B$. If $\b^1$ is recommended, agent 1 only wins item $j$, but by deviating to 0.2 for item $i$, agent 1 could win both items and strictly improve their gain. However, this does not satisfy the ROI constraint when evaluated only for $\b^{1}$, but does satisfy the ROI constraint when evaluated for the entire support of $\B$. 

Secondly, consider valuation $\vec{v}$ with $\vec{v}_1 = (v_{1i}, v_{1j}) = (0.4, 0.2)$ and $\vec{v}_2 = (v_{2i}, v_{2j}) = (0.6, 0)$.
Consider $\B$ for which both bid profiles $\b^1 = ((0, 0.2),( 0.5, 0.2))$ and $\b^2 = ((0, 0.2), (0.7, 0.1))$ are both drawn with probability $\frac{1}{2}$. Note that both agents satisfy the ROI constraint for $\B$. 
If $\b^1$ is recommended, agent 1 only wins item $j$, but by deviating to 0.5 for item $i$, agent 1 could win both items and strictly improve their gain. However, this would not satisfy the ROI constraint. But if additionally, when $\b^{2}$ is recommended, agent 1 deviates to 0.1 for item $j$, agent 1 will still win item $j$ and the ROI will be satisfied.

Finally, consider valuation $\vec{v}$ with $\vec{v}_1 = (v_{1i}, v_{1j}) = (0.5, 0.1)$ and $\vec{v}_2 = (v_{2i}, v_{2j}) = (0, 0.2)$.
Consider $\B$ for which bid profiles $\b^1 = ((0.4, 0),( 0.4, 0.2))$ and $\b^2 = ((0.6, 0), (0.6, 0.2))$ are both drawn with probability $\frac{1}{2}$. Note that both agents satisfy the ROI constraint for $\B$. 
If $\b^1$ is recommended, agent 1 only wins item $i$, but by deviating to 0.2 for item $j$, agent 1 could win both items and strictly improve their gain. However, this would not satisfy the ROI constraint. But if additionally, when $\b^{2}$ is recommended, agent 1 deviates to 0 for item $i$, agent 1 will satisfy the ROI constraint. If in this case, the gain of agent 1 is only evaluated for the conditioned bid profile $\b^{1}$, the gain of agent 1 increases from 0.5 to 0.6. However, if the gain of agent 1 is evaluated for the entire support, the total gain of agent 1 decreases from 0.5 to 0.3.
\end{example}

\section{Missing Material of Section \ref{sec:preliminariesPOA}}
\label{app:prelim}

\subsection{Assumption that $\tau_i \sens_i \le 1$ for all agents $i \in N$}

It is not hard to see that we can assume w.l.o.g. that $\tau_i \sens_i \le 1$ for all agents $i \in N$. To see this, note that if an agent with sensitivity $\sens_i > 0$ has a target parameter $\tau_i > 1/\sens_i$, then $i$ prefers to withdraw whenever they are asked to pay more than $1/\sens_i$ their valuation. 
In other words, there are no equilibria in which the expected payment of agent $i$ exceeds $1/\sens_i$ times their expected value, as agent $i$ would have a negative expected gain and be better off by bidding 0 deterministically in each auction.
In essence, this has the same effect as capping the target parameter at $1/\sens_i$. Note that this argument applies only to agents with sensitivity $\sens_i > 0$. 
Note that the target parameter $\tau_i$ of each value maximizer $i$ (i.e., $\sens_i = 0$) remains unconstrained as $\tau_i \sens_i \le 1$ is always satisfied.

\subsection{Explanation of the Liquid Welfare Objective}

We use \emph{liquid welfare} as the social welfare objective, which is also the standard benchmark in the autobidding literature (see, e.g., \cite{survey}).
Intuitively, the liquid welfare measures the maximum amount of payments one can extract from the agents. To see that the formal definition of liquid welfare introduced in Section~\ref{sec:preliminariesPOA} aligns with this intuition, consider an instance $I$ and a random bid profile $\vec{B}$. 
Then, for an agent $i$ with $\sens_i > 0$, the liquid welfare is $\min(\EX[v_i(\vec{x}_i(\vec{B}))/\sens_i], \EX[ \tau_i v_i(\vec{x}_i(\vec{B}))], \budget_i)$, where the first term reflects that the expected gain $\z_i$ of agent $i$ always remains non-negative, the second term is due to the ROI constraint \eqref{eq:ROI} and the third term is due to the budget constraint \eqref{eq:budget}. 
Note that because $\sens_i \tau_i \le 1$ (as discussed above), the liquid welfare of $i$ reduces to $\min ( \EX[\tau_i v_i(\vec{x}_i(\vec{B}))], \budget_i)$. 
For an agent $i$ with $\sens_i = 0$, the liquid welfare is $\min (\EX[ \tau_i v_i(\vec{x}_i(\vec{B}))], {\budget}_i)$, due to the ROI and budget constraints only.
As a consequence, the liquid welfare of an agent evaluates to $\min(\EX[ \tau_i v_i(\vec{x}_i(\vec{B}))], \budget_i)$, independently of their type.

\subsection{Assumption that $\tau_i = 1$ for all agents $i \in N$}

All previous works studying the mixed agent model (i.e., $T= \{0, 1\}$) assume that the ROI constraint is imposed only on the value maximizers, while the utility maximizers remain unconstrained. As \citet{deng2024} argue, this is equivalent to assuming a uniform ROI target parameter of $\tau_i = 1$ for all agents $i \in N$. 
As we show in Proposition~\ref{prop:uniform-target}, this assumption extends without loss of generality to our hybrid agent model as well.

\begin{restatable}{proposition}{uniformtarget}\label{prop:uniform-target}
 Let $I=(\vec{r}, \vec v, \vec \sens,\vec \tau, \vec \budget)$ and define, for every $i\in N$, the valuation function $v_i'\colon 2^M \rightarrow \mathbb{R}_{\geq 0}$ by $v'_i(S)=\tau_iv_i(S)$ for all $S \subseteq M$. Furthermore, let $\vec{\sens}'=\vec{\sens}\vec{\tau}$, $\vec{\tau}'=\vec{1}$, and consider the instance $I'=(\vec{r}, \vec v',\vec \sens',\vec{\tau}'=\vec{1},\vec \budget)$. The following statements are true.
    
    \medskip
    \noindent
    {\qquad
    (i) $\vec{B} \in \mathcal{R}_i$ for $I$ $\Leftrightarrow$ $\vec{B} \in  \mathcal{R}_i$ for $I'$, \quad
    (ii) $\equ(I) = \equ(I')$, \quad
    (iii) $\displaystyle \sup_{\vec{B} \in \equ(I)}\frac{\opt(I)}{\sw(I,\vec{B})} = \sup_{\vec{B} \in \equ(I')}\frac{\opt(I')}{\sw(I',\vec{B})}$.
    }
\end{restatable}

\begin{proof}[Proof of Proposition~\ref{prop:uniform-target}]
        \smallskip
        \noindent
        \begin{compactenum}[(i)]
        \item ($\Rightarrow$) Every random bid profile $\vec{B} \in \mathcal{R}_i$ for instance $I$ satisfies both the ROI constraint \eqref{eq:ROI} and the budget constraint \eqref{eq:budget} for agent $i$. Since the budget profiles of $I$ and $I'$ are identical, \eqref{eq:budget} is clearly satisfied for $I'$ and $\vec{B}$. Observe that by the linearity of expectation,
        \begin{equation*}
            \EX[p_i(\vec{B})] \leq \tau_i\EX[v_i(x_i(\vec{B}))] = \EX[\tau_i v_i(x_i(\vec{B}))] = \tau_i' \EX[v_i'(\vec{B})],
        \end{equation*}
        which shows that \eqref{eq:ROI} also holds for $I'$ and $\vec{B}$. The argument for $(\Leftarrow)$ is analogous.

        \item Let $(g_i')_{i \in N}$ be the agents’ gain functions for instance $I'$, i.e., for every $\vec{b} \in D$ and every $i \in N$, we define $g_i'(\vec{b}) = v_i'(\vec{b}) - \sens_i' p_i(\vec{b})$. Observe that every agent $i \in N$ satisfies
        \begin{equation}\label{eq:agent-problem-equiv}
            \argmax_{\substack{\vec{B}_i \in D_i:\\ (\vec{B}_i, \vec{B}_{-i}) \in \mathcal{R}_i}} \EX[g_i(\vec{B}_i, \vec{B}_{-i})] = \argmax_{\substack{\vec{B}_i \in D_i:\\ (\vec{B}_i, \vec{B}_{-i}) \in \mathcal{R}_i}} \nicefrac{1}{\tau_i} \cdot \EX[g'_i(\vec{B}_i, \vec{B}_{-i})] = \argmax_{\substack{\vec{B}_i \in D_i:\\ (\vec{B}_i, \vec{B}_{-i}) \in \mathcal{R}_i}} \EX[g'_i(\vec{B}_i, \vec{B}_{-i})]
        \end{equation}
        for every random bid profile $\vec{B}_{-i} \in D_{-i}$, where the first equality follows from the linearity of expectation. By combining \eqref{eq:agent-problem-equiv} with property~(i), we conclude that the optimization problems faced by each agent $i$ in $I$ and $I'$ are identical. The claim follows as a corollary.

        \item Note that $I'$ satisfies $\tau_i'\sens_i' = \tau_i \sens_i \leq 1$ for every agent $i \in N$. Therefore,
        \begin{equation}\label{eq:opt-target-one}
            \opt(I) = \max_{\vec x \in \vec X} \sum_{i \in N} \min(\tau_i v_i(\vec x_i), \budget_i) = \max_{\vec x \in \vec X} \sum_{i \in N} \min(v_i'(\vec x_i), \budget_i) = \opt(I'),
        \end{equation}
        and similarly, for every random bid profile $\vec{B} \in D$,
        \begin{equation}\label{eq:lw-target-one}
            \sw(I, \vec{B}) = \sum_{i \in N} \min\left(\EX[\tau_i v_i(\vec{x}_i(\vec{B}))], \budget_i\right) = \sum_{i \in N} \min\left(\EX[v_i'(\vec{x}_i(\vec{B}))], \budget_i\right) = \sw(I', \vec{B}).
        \end{equation}
        We conclude that
        \begin{equation*}
            \sup_{\vec{B} \in \equ(I)}\frac{\opt(I)}{\sw(I, \vec{B})} = \sup_{\vec{B} \in \equ(I)}\frac{\opt(I')}{\sw(I', \vec{B})} = \sup_{\vec{B} \in \equ(I')}\frac{\opt(I')}{\sw(I', \vec{B})}.
        \end{equation*}
        The first equality follows from \eqref{eq:opt-target-one} and \eqref{eq:lw-target-one}, and the second equality from property (ii). The claim follows.
    \end{compactenum}
\end{proof}

\subsection{Unbounded Price of Anarchy for $\eta = 1$}

It is not hard to see that the price of anarchy is unbounded for $\eta = 1$, even for a single auction. To see this, consider a single auction with $r = v_{\rw}$. Suppose that agent $\rw$ is a utility maximizer and the only agent that can meet the reserve price. Then $\rw$ deterministically bidding 0 is a \mne\ with a liquid welfare of $0$, while the optimal liquid welfare is $v_{\rw} > 0$. The \poa\ is thus unbounded.

\subsection{Lambert $W$ Function}\label{app:lambert}

In order to derive POA bounds analytically, we use the Lambert $W$ function, which is the multivalued inverse of the function $f(z) = z e^z$. 
In this work, we exclusively use the \emph{principal branch} of the Lambert $W$ function and denote it by $W_0$. We present its definition and the expression for its derivative. For a more detailed treatment, we refer the interested reader to~\cite{corless96}.

\begin{definition}\label{def:lambert}
The principal branch of the Lambert $W$ function is the real function $W_{0} \colon \left[-\nicefrac{1}{e}, \infty\right) \rightarrow \left[-1, \infty\right)$ that satisfies $z = W_{0}(z) e^{W_{0}(z)}$, for every $z \in \left[-\nicefrac{1}{e}, \infty\right)$.
\end{definition}

\begin{fact}\label{fact:lambert-derivative}
For every $z > -\nicefrac{1}{e}$, the derivative of $W_{0}(z)$ is given by $\frac{W_{0}(z)}{z(1 + W_{0}(z))}$.
\end{fact}

\section{Missing Material of Section \ref{sec:smoothness:extension}}

\subsection{Proofs of Proposition~\ref{prop:opt-budgetless-equivalence} and Lemma~\ref{lemma:individual-lw-lb}}\label{app:sec:individual-lw-lb}

\propoptbudgetlessequivalence*
\begin{proof}[Proof of Propostion~\ref{prop:opt-budgetless-equivalence}]
    Using Definition \ref{def:proxy-instance} and the definition of budget-capped valuations, we obtain:
    \[
        \opt(\hat{I}(I, \vec{B}))= \max_{\vec{x} \in \vec{X}}\sum_{i \in N}v_i^{\budget_i}(\vec{x}_i)=\max_{\vec{x} \in \vec{X}}\sum_{i \in N}\min\left( v_i(\vec{x}_i), \budget_i\right)=\opt(I).\qedhere
    \]
\end{proof}

\lemmaindividuallwlb*
\begin{proof}[Proof of Lemma~\ref{lemma:individual-lw-lb}]
    Fix an agent $i \in N$. We distinguish two cases.

    \medskip
    \noindent
    \textbf{Case 1:} $\budget_i < \EX[v_i(x_i(\vec{B}))]$. We have:
    \begin{align*}
    \min\left(\EX[v_i(x_i(\vec{B}))], \budget_i \right) 
    &= \delta \budget_i + (1 - \delta) \budget_i \\
    &\geq \delta \EX\left[\min\left(v_i(x_i(\vec{B}_i', \vec{B}_{-i})), \budget_i \right)\right] + (1 - \delta)\budget_i \\
    &= \delta \EX\left[v_i^{\budget_i}(x_i(\vec{B}_i', \vec{B}_{-i}))\right] + (1 - \delta)\budget_i \\
    &\geq \delta \EX\left[v_i^{\budget_i}(x_i(\vec{B}_i', \vec{B}_{-i}))\right] + (1 - \delta)\EX\left[p_i(\vec{B})\right] \\
    &= \delta \EX\left[v_i^{\budget_i}(x_i(\vec{B}_i', \vec{B}_{-i})) - \hat{\sens}_i(\vec{B}) p_i(\vec{B}_i', \vec{B}_{-i})\right] + (1 - \delta + \hat{\sens}_i(\vec{B}) \delta)\EX\left[p_i(\vec{B})\right]\\
    &=\delta \EX\left[\hat{g}_i(\vec{B}_i', \vec{B}_{-i})\right] + (1 - \delta + \delta \hat{\sens}_{i}(\vec{B}))  \pe \left[p_i(\vec{B})\right].
    \end{align*}
    The first equality follows by the definition of \textbf{Case 1}, and the second equality follows by the definition of budget-capped valuations. The second inequality follows since, by assumption, $\vec{B}$ is a CCE for $I$ and therefore satisfies the budget constraint in \eqref{eq:budget} for instance $I$. Finally, the third equality follows since $\hat{\sens}_{i}(\vec{B}) = 0$ holds by Definition~\ref{def:proxy-instance}.

    \medskip
    \noindent
    \textbf{Case 2:} $\budget_i \geq \EX[v_i(x_i(\vec{B}))]$. In this case, we have:
    \begin{align*}
    \min\left(\EX[v_i(x_i(\vec{B}))], \budget_i \right) 
    &= \delta \EX[v_i(x_i(\vec{B}))] + (1 - \delta) \EX[v_i(x_i(\vec{B}))] \\
    &\geq \delta \EX[v_i(x_i(\vec{B}))] + (1 - \delta) \EX\left[p_i(\vec{B})\right] \\
    &= \delta \EX[g_i(\vec{B})] + (1 - \delta + \delta \sens_i)\EX\left[p_i(\vec{B})\right] \\
    &\geq \delta \EX[g_i(\vec{B}_i', \vec{B}_{-i})] + (1 - \delta + \delta \sens_i)\EX\left[p_i(\vec{B})\right] \\
    &= \delta \EX[v_i(x_i(\vec{B}_i', \vec{B}_{-i})) - \sens_i p_i(\vec{B}_{i}', \vec{B}_{-i})] + (1 - \delta + \delta \sens_i)\EX\left[p_i(\vec{B})\right] \\
    &\geq \delta \EX[\min\left(v_i(x_i(\vec{B}_i', \vec{B}_{-i})), \budget_i\right) - \sens_i p_i(\vec{B}_{i}', \vec{B}_{-i})] + (1 - \delta + \delta \sens_i)\EX\left[p_i(\vec{B})\right] \\
    &= \delta \EX\left[v_i^{\budget_i}(x_i(\vec{B}_i', \vec{B}_{-i})) - \sens_i p_i(\vec{B}_{i}', \vec{B}_{-i})\right] + (1 - \delta + \delta \sens_i)\EX\left[p_i(\vec{B})\right] \\
    &= \delta \EX\left[v_i^{\budget_i}(x_i(\vec{B}_i', \vec{B}_{-i})) - \hat{\sens}_i(\vec{B}) p_i(\vec{B}_{i}', \vec{B}_{-i})\right] + (1 - \delta + \delta \hat{\sens}_i(\vec{B}))\EX\left[p_i(\vec{B})\right]\\
    &=\delta \EX\left[\hat{g}_i(\vec{B}_i', \vec{B}_{-i})\right] + (1 - \delta + \delta \hat{\sens}_{i}(\vec{B}))  \pe \left[p_i(\vec{B})\right].
    \end{align*}

    Here, the first equality follows by the definition of \textbf{Case 2}, and the first inequality holds since, by assumption, $\vec{B}$ is a CCE of $I$ and therefore satisfies the ROI constraint in \eqref{eq:ROI}. Similarly, the second inequality follows from \eqref{eq:cce-def}, since $\vec{B}$ is a CCE of $I$, and it holds by assumption that $(\vec{B}_i', \vec{B}_{-i}) \in \mathcal{R}_i$. The fourth equality holds by the definition of budget-capped valuations, while the fifth follows since, by Definition~\ref{def:proxy-instance},  $\hat{\sens}_{i}(\vec{B}) = \sens_i$ holds.
\end{proof}

\subsection{Proof of Lemma~\ref{lem:comp}}
\label{app:sec:ll}
\begin{restatable}{corollary}{cordeviationindependent}\label{cor:deviation-independent}
       Consider a budget-free instance $I \in \mathcal{I}_{\xos}^{\infty}$. Fix an agent $i \in N$ and consider a bid profile $\B'_i \in \Delta_i$ that is ROI-restricted and let $\B_{-i} \in \Delta_{-i}$ be arbitrary.
        Then, $(\B'_i, \B_{-i}) \in \mathcal{R}_i$.
\end{restatable}
\begin{proof}[Proof of Corollary~\ref{cor:deviation-independent}]
We have:
\begin{align*}
\pe[p_{i}(\B'_i, \B_{-i})] 
& = \int_{D_{-i}} f_{\B_{-i}}(\b_{-i}) \cdot \pe[p_i(\B'_i, \b_{-i})] d \b_{-i} \\
& \le \int_{D_{-i}} f_{\vec{B}_{-i}}(\b_{-i}) \cdot \pe[ v_i(x_i(\B'_i, \b_{-i}))] d \b_{-i} 
= \pe
[v_i(\vec{x}_i(\B'_i, \vec{B}_{-i}))],
\end{align*}
where the inequality follows because $\vec{B}'_i$ is ROI-restricted. We have thus shown that \eqref{eq:ROI} is satisfied. Since $I$ is budget-free (and therefore \eqref{eq:budget} is trivially true), the proof follows.
\end{proof}

\lemcomp*
\begin{proof}[Proof of Lemma~\ref{lem:comp}]\label{proof:lem:comp}
Consider some agent $i \in N_{\type}(I)$.
If $i$ is not the rightful winner of any auction, we define $\B'_i$ such that $B'_{ij} = 0$ deterministically. Clearly, $(\vec{B}_{i}', \vec{B}_{-i}) \in \mathcal{R}_i$ holds. 
Otherwise, let $i$ is the rightful winner of auction $j$, i.e., $i = \rw(j)$. 
By assumption, $\fpa(r)$ is $(\lambda_{\type}, \mu_{\type})$-smooth for $\type$. Thus, for each auction $j \in M$, there exists an ROI-restricted bid $B'_{ij}$ such that, for each well-supported bid profile $\vec{b}_j$, we have:
\begin{equation}\label{eq:smooth-local-rw}
\pe[\z_{i}(B'_{ij}, (\vec{b}_{j})_{-i})] \geq \lambda_{\type}  v^*_{\rw(j)j} - \mu_{\type}  p_{\aw(j)j}(\vec{b}_j),
\end{equation}

We define the random deviation $\vec{B}'_i$ of agent $i$ for the global mechanism $\mech$ simply by drawing a bid $b'_{ij}$ for each auction $j \in M$ independently according to $B'_{ij}$ if $i = \rw(j)$, and letting $B'_{ij} = b'_{ij} = 0$ deterministically if $i \neq \rw(j)$. For each bid profile $\vec{b}_{-i}$, we have:
\begin{align*}
\pe[p_{i}(\vec{B}'_{i}, \vec{b}_{-i})] 
&= \pe\left[\sum_{j \in M} p_{ij}(B'_{ij}, (\vec{b}_j)_{-i})\right] 
\le \pe\left[\sum_{j \in M} v^*_{ij} x_{ij}(B'_{ij}, (\vec{b}_j)_{-i})\right] \\
&= \pe[v^*_{i}(x_{i}(\vec{B}'_{i}, \vec{b}_{-i}))] 
\le \pe[v_{i}(x_{i}(\vec{B}'_{i}, \vec{b}_{-i}))].
\end{align*}
Here the inequality holds because $B'_{ij}$ is ROI-restricted for each $j$ (which also holds trivially for all auctions $j$ with $i \neq \rw(j)$). The second equality follows by the definition of $v^*_i$, and the last inequality follows from property \textbf{XOS2}.
We conclude that $\vec{B}'_i$ is ROI-restricted for each agent $i$. Therefore, by Corollary~\ref{cor:deviation-independent}, $(\vec{B}_{i}', \vec{B}_{-i}) \in \mathcal{R}_i$, proving the first part of the lemma.

We continue with the second part. 
Fix a type $t \in T$. Given any bid profile $\vec{b} \in \supp(\vec{B})$, for every $i \in N_t(I)$ we have
\begin{align}
\pe[\z_{i}(\vec{B'}_{i}, \vec{b}_{-i})] 
&= \sum_{j \in M} \pe\left[ \z_{i}\left(B'_{ij}, (\vec{b}_{j})_{-i}\right)\right] \geq \sum_{j \in M: \rw(j) = i} \lambda_{\type} v^*_{\rw(j)j} - \mu_{\type} p_{\aw(j)j}(\vec{b}). \label{eq:lambda-mu-per-agent}
\end{align}
Here the equality follows by linearity of expectation. The inequality follows by applying~\eqref{eq:smooth-local-rw} to all auctions $j$ such that $i = \rw(j)$, and using that the expected gain of $i$ is non-negative for all $j$ with $i \neq \rw(j)$. Note that $\B$ is well-supported and thus $\vec{x}(\vec{b}) \neq 0$ for each $\vec{b} \in \supp(B)$.
Taking expectations over $\vec{B}$ on both sides of~\eqref{eq:lambda-mu-per-agent} and summing over all $i \in N_t(I)$ yields~\eqref{eq:lambda-mu-per-groupMT}. The claim follows.
\end{proof}

\subsection{Proof of Theorem~\ref{theorem:templateII}}
\label{app:templateII}

\extensiontemplaterefined*

A crucial building block in deriving our mathematical program (\poarmp) is the characterization of optimal calibration vectors.

\lemcalibrationcharacterization*

\begin{proof}[Proof of Lemma~\ref{lemma:optimal-xi}]\label{proof:lemma:optimal-xi}
Define $\vec{\delta}'$ such that $\delta'_{\type} = {O}/{\lambda_{\type}}$ for each $\type \in T$. We first show that $\vec{\delta}' \in \mathcal{C}(\vec{\mu}, T)$ and $\min_{\type \in T}\lambda_{\type}\delta'_{\type}=O$.

First, note that $\vec{\delta}' \in (0,1]^{|T|}$, since for each $\type \in T$ it holds that:
\[
    \delta'_{\type} 
    = \frac{O}{\lambda_{\type}} 
    \leq \frac{\min_{\type \in T}\lambda_{\type}}{\lambda_{\type}} \leq 1,
\]
and $\delta'_{\type} > 0$ as $O > 0$. 
Furthermore, $\vec{\delta}'$ satisfies:
\begin{align*}
    \max_{ \type \in T} \left(\delta'_{\type} \mu_{\type} \right) + \max_{\type \in T} \left(\delta'_{\type} \left(1 - t \right) \right) 
    &=\max_{ \type \in T} \left(\frac{O}{\lambda_{\type}} \mu_{\type} \right) 
    + \max_{\type \in T} \left(\frac{O}{\lambda_{\type}} \left(1 - t \right) \right) 
    = O \cdot \left(  \max_{ \type \in T} \left(\frac{\mu_{\type}}{\lambda_{\type}} \right) + \max_{\type \in T} \left(\frac{1 - t }{\lambda_{\type}}  \right) \right) \le 1,
\end{align*}
where the last inequality follows from \eqref{eq:optimal-xi-lambda-mu}. Hence, $\vec{\delta}' \in \mathcal{C}(\vec{\mu}, T)$, and therefore $\min_{\type \in T}\lambda_{\type}\delta'_{\type}=O$.

To complete the proof, we need to show that $\max_{\vec{\delta} \in \mathcal{C}(\vec{\mu},T)}\min_{\type \in T}\lambda_{\type}\delta_{\type}=O$. Towards a contradiction, assume that there exists $\bar{\vec{\delta}} \in \mathcal{C}(\vec{\mu}, T)$ with $\min_{\type \in T}\lambda_{\type}\bar{\delta}_{\type}>O$. We distinguish two cases for the value of $O$.

\medskip
\noindent \textbf{Case 1:} $\min_{\type \in T}\lambda_{\type} \leq \left(\max_{\type \in T}\left(\frac{\mu_{\type}}{\lambda_{\type}}\right) 
    + \max_{\type \in T}\left(\frac{1-t}{\lambda_{\type}}\right) \right)^{-1}$. In this case, we conclude that
\begin{equation}\label{eq:contra-1}
\min_{\type \in T}\lambda_{\type} = O < \min_{\type \in T}\lambda_{\type}\bar{\delta}_{\type} \leq \min_{\type \in T}\lambda_{\type}.
\end{equation}
The equality holds by the definition of \textbf{Case 1}, and the first inequality holds by assumption. The second inequality follows since $\bar{\vec{\delta}} \in \mathcal{C}(\vec{\mu}, T) \subseteq (0, 1]^{|T|}$, and therefore $\bar{\delta}_{\type} \le 1$ holds for all $\type \in T$. Thus, our analysis in \eqref{eq:contra-1} implies that \textbf{Case 1} cannot occur, and we move on to \textbf{Case 2}.

\noindent \textbf{Case 2:} $\min_{\type \in T}\lambda_{\type} > \left(\max_{\type \in T}\left(\frac{\mu_{\type}}{\lambda_{\type}}\right) 
    + \max_{\type \in T}\left(\frac{1-t}{\lambda_{\type}}\right) \right)^{-1}$. Let $\hat{\type}:=\argmax _{\type \in T} \mu_{\type} / \lambda_{\type}$ and $\tilde{\type}:=\argmax_{\type \in T} (1-t) / \lambda_{\type}$.  
    In this case,
    \begin{align*}
        1 < \left(\frac{\mu_{\hat{\type}}}{\lambda_{\hat{\type}}} 
    + \frac{1-{\tilde{\type}}}{\lambda_{\tilde{\type}}}\right) \cdot \min_{\type \in T}\lambda_{\type}\bar{\delta}_{\type}
    &= \frac{\mu_{\hat{\type}}}{\lambda_{\hat{\type}}}\cdot \min_{\type \in T}\lambda_{\type}\bar{\delta}_{\type} 
    + \frac{1-{\tilde{\type}}}{\lambda_{\tilde{\type}}} \cdot \min_{\type \in T}\lambda_{\type}\bar{\delta}_{\type} \\
    &\leq \mu_{\hat{\type}}\bar{\delta}_{\hat{\type}} + (1-{\tilde{\type}}) \bar{\delta}_{\tilde{\type}} 
    \leq \max_{\type \in T}\left( \mu_{\type}\bar{\delta}_{\type} \right) + \max_{\type \in T}\left( \left(1-t\right)\bar{\delta}_{\type} \right) 
    \leq 1.
    \end{align*}
Here, the first inequality holds by the definition of \textbf{Case 2}, and the last inequality holds since $\vec{\bar{\delta}} \in \mathcal{C}(\vec{\mu}, T)$. However, our analysis implies that \textbf{Case 2} also cannot occur.

As neither \textbf{Case 1} nor \textbf{Case 2} can occur, we have arrived at a contradiction. This concludes the proof.
\end{proof}

\begin{proof}[Proof of Theorem~\ref{theorem:templateII}]
We can now use Lemma \ref{lemma:optimal-xi} together with our Smoothness Lemmas (Lemma \ref{lem:fpa-smoothness-key-2} and Lemma \ref{lem:fpa-smoothness-key-1}) to derive the POA-revealing mathematical program (\poarmp) as defined in Section~\ref{sec:POA-reveal}. To obtain a bound on the POA, we determine a vector $\vec{\mu} = (\mu_{\type})_{\type \in T}$ that maximizes the expression in \eqref{eq:optimal-xi-lambda-mu} subject to the constraints \eqref{eq:key-lemma-0-constraint}--\eqref{eq:key-lemma-2-constraint}. This concludes the proof. 
\end{proof}

\section{Missing Material of Section \ref{sec:POA-without-reserve-prices}} \label{app:PoA-bounds-without-reserve}

\subsection{Properties of the Feasible Solution} \label{app:identities}

Given $\omega \in (0,1)$ and a set of types $T$, recall the feasible solution $\vec{\mu}^{*}$ in Definition \ref{def:mu:star}. (see also Figure \ref{fig:partition}).
\begin{figure}
  \centering
  \begin{tikzpicture}
    \begin{axis}[
        width=13cm,
        height=5cm,
        xlabel={},
        ylabel={$\lambda^*_{\type}(\omega, T)$},
        ylabel style={rotate=-90}, 
        xmin=0, xmax=1.0,
        ymin=0, ymax=1.15,
        xtick={0,0.5,1},
        ytick={0, {1/(2*ln(2))}, 1},
        grid=both,
        grid style={line width=0.1pt, draw=gray!10},
        major grid style={line width=0.2pt, draw=gray!50},
        tick label style={font=\small},
        label style={font=\small},
        clip=false,
      ]

      \draw[dashed, thick] (axis cs:0.5,0) -- (axis cs:0.5,1.15)
        node[pos=0.87, above right] {$\omega$};

      \addplot+[only marks, mark=*, color=blue, mark size=1.8pt] 
        plot [domain=0.001:0.49, samples=10] {x/(-ln(1-x))};

      \addplot+[only marks, mark=*, color=red, mark size=1.8pt] coordinates {
        (0.55, {1/(2*ln(2))})
        (0.62, {1/(2*ln(2))})
        (0.70, {1/(2*ln(2))})
        (0.75, {1/(2*ln(2))})
        (0.85, {1/(2*ln(2))})
        (1.0,  {1/(2*ln(2))})
      };
      
      \node at (axis cs:0.25,1.075) {$L_{\omega}(T)$};
      \node at (axis cs:0.75,1.075) {$H_{\omega}(T)$};
      
    \end{axis}
  \end{tikzpicture}
  \caption{Illustration of $\vec{\lambda}^*(\omega, T)$ for $\omega = \frac{1}{2}$ and the partitioning of agent type set $T$ into $L_{\omega}(T)$ (blue) and $H_{\omega}(T)$ (red). For all $\type \in H_{\omega}(T)$, the value $\lambda^*_{\type}(\omega, T)$ in Lemma \ref{lemma:technical-lemma} is given by $\lambda^*_{\type}(\omega, T) = \frac{\omega}{-\ln(1-\omega)} = \frac{1}{2\ln 2} \approx 0.72$. For all $\type \in L_{\omega}(T)$, the value $\lambda^*_{\type}(\omega)$ satisfies $\lambda^*_{\type}(\omega) \geq \frac{\omega}{-\ln(1-\omega)}$.}
  \label{fig:partition}
\end{figure}
\begin{restatable}{lemma}{technicallemma}
    \label{lemma:technical-lemma}
The following properties hold for every set of agent types $T$ with $\max(T) > 0$ and every $\omega \in \left(0, \max(T) \right] \cap (0,1)$:    
\begin{compactenum}[(i)]
   \item
   $
   \lambda^*_{\type}(\omega, T) =
            \begin{cases} 
                \frac{\omega}{-\ln(1-\omega)}, & \text{if } \type \in H_{\omega}(T), \\
                \frac{{\type}}{-\ln(1-{\type})}, & \text{if } \type \in L_{\omega}(T)\cap(0,1), \\
                1, & \text{if $\type \in L_{\omega}(T)\cap \{0\}$}.
            \end{cases} 
    $
    \item 
    $
    \min_{
    {\type \in T \cap (0,1] }}\lambda_{\type}^*(\omega, T) = \frac{\omega}{-\ln(1-\omega)}.
    $
    \item 
    $
    \max_{
    {\type \in T \cap (0,1]}} \frac{1-{\type} }{\lambda^*_{\type}(\omega, T)} =
    \begin{cases}
         \frac{-\ln(1-\omega)(1-\min \left(T \setminus \{0\}\right))}{\omega}, & \text{if } \min \left(T \setminus \{0\}\right) \ge \omega, \\
        \frac{-\ln(1-\min \left(T \setminus \{0\}\right))(1-\min \left(T \setminus \{0\}\right))}{\min \left(T \setminus \{0\}\right)}, & 
         \text{if } \min \left(T \setminus \{0\}\right) < \omega.  
    \end{cases}
    $
    \item 
    $
    \max_{
            {\type \in T \cap (0,1]}} \frac{\mu^*_{\type}(\omega, T)}{\lambda^*_{\type}(\omega, T)}= \frac{\max(T)}{\omega}.
    $
   \end{compactenum}
\end{restatable} 

\begin{proof}[Proof of Lemma~\ref{lemma:technical-lemma}]
For notational convenience, we use $\lambda_{\type} := \lambda^*_{\type}(\omega, T)$ and $\mu_{\type} := \mu^*_{\type}(\omega, T)$ for each $\type \in T$. Also, we use $L := L_{\omega}(T)$ and $H :=H_{\omega}(T)$. 
Note that $H$ is guaranteed to be non-empty by the range of $\omega$. All statements follow from the definition of $\vec{\mu^*}(\omega, T)$ in \eqref{eq:mu-star-params} and elementary calculus.
\medskip

\noindent \textit{(i)} For every $\type \in L\cap \{0\}$, $\lambda_{\type} = \mu_{\type} = 1$ holds. Then, for all $\type \in L\cap (0,1)$, we have that:
\begin{equation*}
    \lambda_{\type} = \frac{\mu_{\type}}{{\type}}\left(1-e^{-\frac{{\type}}{\mu_{\type}}} \right)= \frac{1-e^{\ln\left(1-{\type}\right)}}{-\ln(1-{\type})}=\frac{{\type}}{-\ln(1-{\type} )}.  
\end{equation*}
Similarly, for all $\type \in  H$ we have that:
\begin{equation*}
    \lambda_{\type} = \frac{\mu_{\type}}{{\type}}\left(1-e^{-\frac{{\type}}{\mu_{\type}}} \right)= \frac{1-e^{\ln\left(1-\omega\right)}}{-\ln(1-\omega)}=\frac{\omega}{-\ln(1-\omega)}.  
\end{equation*}
\medskip

\noindent \textit{(ii)} First, assume that $L \cap (0,1) = \emptyset$. Then, by identity \textit{(i)} we have:
\[
\min_{\type \in T \cap (0,1]} \lambda_{\type} 
= 
\min_{\type \in H} \lambda_{\type} 
= \frac{\omega}{-\ln(1- \omega)}.
\]
Next, assume that $L \cap (0,1) \neq \emptyset$. By the definition of $L$, the fact that $f(x)=\frac{x}{-\ln(1-x)}$ is non-increasing, and statement \textit{(i)}, we obtain:
\begin{equation*}
    \min_{
              {\type \in L \cap (0,1)}}\lambda_{\type} =  \min_{
              {\type \in L \cap (0,1)}}\frac{{\type}}{-\ln(1-\type)} \geq \frac{\omega}{-\ln(1- \omega)}= \min_{\type \in H} \lambda_{\type} .
\end{equation*}
Therefore:
\[\min_{\type \in T \cap (0,1]}\lambda_{\type} = \min \left\{ \min_{\type \in H}\lambda_{\type}, \min_{\type \in L\cap (0,1)} \lambda_{\type} \right\} =\min_{\type \in H}\lambda_{\type} = \frac{\omega}{-\ln(1-\omega)}.\]
\medskip

\noindent \textit{(iii)} First, consider the case of $\min(T \setminus \{0\}) \ge \omega$. 
Then $L \cap (0,1)=\emptyset$.
By statement \textit{(i)} and since $ \min H=\min(T\setminus \{0\})$, it holds that:
\begin{equation*}
    \max_{
    {\type \in T: \cap (0,1] }}\frac{1-{\type} }{\lambda_{\type}}= {\max_{\type \in  H} \frac{1-{\type}}{\lambda_{\type}}}=\frac{-\ln(1-\omega)(1-\min(T \setminus \{0\}))}{\omega}.
\end{equation*}
Next, assume that $\min(T \setminus \{0\}) < \omega$. Then $L \cap (0,1) \neq \emptyset$.
We have that:
\begin{align}
        \max_{\type \in  H} \frac{1-{\type}}{\lambda_{\type}}&= \frac{-\ln(1-\omega)\max_{\type \in H}(1-{\type})}{\omega}  = \frac{-\ln(1-\omega) \left(1-\min(H)\right)}{\omega} \nonumber \\
        &\leq \frac{-\ln(1-\omega) \left(1-\omega\right)}{\omega}
        \le \max_{\substack{\type \in L \cap (0,1)}}  \frac{-\ln(1-{\type})\left(1-{\type}\right)}{{\type}} \nonumber = \max_{\substack{\type \in L\cap (0,1)}} \frac{1-\type}{\lambda_{\type}}.\numberthis\label{eq:iv-derivation}
\end{align}
    The first inequality holds by the definition of $H$. The second inequality holds by the definition of $L$ and due to the fact that the function $h(x)=\frac{-\ln(1-x)(1-x)}{x}$ is non-increasing in $(0,1)$. Using \eqref{eq:iv-derivation}, we conclude that:
    \begin{align*}
        \max_{\substack{\type \in T:\cap (0,1]}}\frac{1-{\type} }{\lambda_{\type}}
        & = \max \left\{\max_{\type \in  H} \frac{1-{\type} }{\lambda_{\type}}, \max_{\substack{\type \in L \cap (0,1)} }\frac{1-{\type} }{\lambda_{\type}}\right\} 
        =  \max_{\substack{\type \in L: \cap (0,1)} }\frac{1-\type }{\lambda_{\type}} \\
        & =  \max_{\substack{\type \in L:\cap (0,1)} }\frac{ -\ln(1-{{\type}})(1-{\type})}{{\type}}
        = \frac{ -\ln(1-{t_{\min_{+}}})(1-t_{\min_{+}})}{t_{\min_{+}}}.  
    \end{align*}
The last equality holds because the function $h(x)$, as defined above, is non-increasing in $(0,1)$, and thus the maximum is attained for $ \min(T \setminus \{0\})$.
\medskip

\noindent \textit{(iv)} If $L = \emptyset$, the statement follows by statement \textit{(i)}. Otherwise,  we have that:
\begin{equation*}
    \max_{\substack{\type \in T \cap (0,1]}} \frac{\mu_{\type}}{\lambda_{\type}}=\max\left\{\max_{\substack{\type \in H}} \frac{\mu_{\type}}{\lambda_{\type}}, \max_{\substack{\type \in L \cap (0,1)}} \frac{\mu_{\type} }{\lambda_{\type}}\right\}
    = \max\left\{\max_{\substack{\type \in H}} \frac{\type}{\omega},1 \right\}
    = \frac{\max(T)}{\omega},
\end{equation*}
where, the last equality follows from the definition of $H$ and because $\omega \le \max(T)$ holds by assumption. 
\end{proof}

\subsection{Proof of Lemma~\ref{lemma:general-omega-lb}} 
\label{app:lowerbound-poarmp-general}
\generalomegalb*
\begin{proof}[Proof of Lemma~\ref{lemma:general-omega-lb}]
Let $\mu^*_{\type}(\omega, T)$ and $\lambda^*_{\type}(\omega, T)$ be as previously defined in \eqref{eq:mu-star-params} and Lemma~\ref{lemma:technical-lemma}, respectively. 
For notational convenience, we use $\lambda_{\type} := \lambda^*_{\type}(\omega, T)$ and $\mu_{\type} := \mu^*_{\type}(\omega, T)$ for each $\type \in T$. 
As argued in Corollary \ref{cor:mu:feasible}, $(\mu_{\type})_{\type \in T}$ is a feasible solution to $\poarmp(T)$. 
We bound the two terms of the $\min$ expression in the objective function of $\poarmp(T)$ separately.

First, we show that:
\begin{equation}\label{eq:global-lambda-min}
    \min_{\type \in  T}\lambda_{\type} = \frac{\omega}{-\ln(1-\omega)}.
\end{equation}
When $\min(T)> 0$, \eqref{eq:global-lambda-min} immediately follows from property \textit{(ii)} of Lemma \ref{lemma:technical-lemma}. Otherwise, when $\min(T)=0$, we have that
\begin{equation*}
    \min_{\type \in  T}\lambda_{\type}= \min\left\{\min_{
    {\type \in  T: \type>0}}\lambda_{\type},\min_{
    {\type \in  T:  \type}=0}\lambda_{\type} \right\} = \min\left\{\frac{\omega}{-\ln(1-\omega)},1 \right\} = \frac{\omega}{-\ln(1-\omega)}.
\end{equation*}
The last equality holds since $\frac{z}{-\ln(1-z)}<1$, for all $z \in (0,1)$. We obtain \eqref{eq:global-lambda-min}.

Next, we show that:
\begin{equation}\label{eq:second-quantity-general-bound}
    \max_{\type \in T} \frac{\mu_{\type}}{\lambda_{\type}} 
    + \max_{\type \in T} \frac{1 - {\type}}{\lambda_{\type}} \leq \frac{\max(T) }{\omega} + 1.
\end{equation}
When $\min(T)>0$, we have:
\begin{align*}
    \max_{\type \in T} \frac{\mu_{\type}}{\lambda_{\type}} 
    & + \max_{\type \in T} \frac{1 - {\type}}{\lambda_{\type}}
    = \max_{
    {\type \in T \cap (0,1]}} \frac{\mu_{\type}}{\lambda_{\type}} 
    + \max_{
     {\type \in T \cap (0,1]}} \frac{1 - {\type}}{\lambda_{\type}} 
    =\frac{\max(T)}{\omega}  + \max_{
     {\type \in T \cap (0,1]}} \frac{1 - {\type}}{\lambda_{\type}} \\
    & 
    = 
    \frac{\max(T)}{\omega}   + \max\left\{
    \frac{-\ln(1-\omega) \left(1-\omega\right)}{\omega},    
    \max_{\substack{\type \in L\cap(0,1)}}  \frac{-\ln(1-{\type} )\left(1-{\type}\right)}{{\type}}
    \right\} \leq
    \frac{\max(T)}{\omega} 
    + 1.
\end{align*}
Here, the second and third equality follow from properties \textit{(iv)} and \textit{(i)} of Lemma \ref{lemma:technical-lemma}, respectively.
Then, the inequality holds by the fact that $\frac{-\ln(1-z)(1-z)}{z} < 1$ holds for all $z \in (0,1)$.      

Similarly, when $\min(T)=0$, we obtain: 
\begin{align*}
    \max_{\type \in T} \frac{\mu_{\type}}{\lambda_{\type}} 
    + \max_{\type \in T} \frac{1 - {\type}}{\lambda_{\type}} &= \max\left\{\max_{\substack{\type \in T: \cap (0,1]}} \frac{\mu_{\type}}{\lambda_{\type}},  \frac{\mu_{0}}{\lambda_{0}}\right\} + \max\left\{\max_{\substack{\type \in T\cap (0,1]}} \frac{1 - {\type} }{\lambda_{\type}},  \frac{1}{\lambda_{0}}\right\}\\
    &=\max\left\{\frac{\max(T)}{\omega}, 1 \right\}
    + \max\left\{\max_{\substack{\type \in T:\cap (0,1]}} \frac{1 - {\type} }{\lambda_{\type}}, 1\right\}
    =\frac{\max(T)}{\omega} 
    + 1,
\end{align*}
where the second equality holds by properties \textit{(i)} and \textit{(iv)} of Lemma \ref{lemma:technical-lemma}, i.e., for $t=0$, it holds that $\lambda_{0}=1$. 
The last equality holds because $\omega \leq \max(T)$ (by definition) and by property \textit{(iii)} of Lemma \ref{lemma:technical-lemma}, as $\frac{-\ln(1-z)(1-z)}{z} < 1$ for all $z \in (0,1)$. We thus obtain \eqref{eq:second-quantity-general-bound}.
         
By combining \eqref{eq:global-lambda-min} and \eqref{eq:second-quantity-general-bound}, we obtain that the optimal value of $\poarmp(T)$ is at least:
\begin{align*}
    \min \left\{ 
    \min_{\type \in T} \lambda_{\type}, 
    \left( 
    \max_{\type \in T} \frac{\mu_{\type}}{\lambda_{\type}} 
    + \max_{\type \in T} \frac{1 - {\type}}{\lambda_{\type}} 
    \right)^{-1} \right\} 
    &\geq \min \left\{ 
    \frac{\omega}{-\ln(1-\omega)}, 
    \frac{{\omega}}{\omega + \max(T)} 
    \right\}.
\end{align*}  
\end{proof}

\subsection{Proof of Claim \ref{claim:LW:prop0}} \label{app:math-claims-theorem-no-reserve-sigma-max}
\claimLWpropzero*
\begin{proof}[Proof of Claim \ref{claim:LW:prop0}]
Let $z_0 = 1+ \nicefrac{W_{0}\left(-2e^{-2} \right)}{2}$. Consider the function $h(z) = z - 1 - W_{0}\left(-e^{-z-1} \right)$ in $\left[z_0, 1\right]$.
By Definition \ref{def:lambert}, we have:
\[
e^{W_{0}(-2e^{-2})} W_{0}(-2e^{-2}) = -2e^{-2} \quad \Leftrightarrow \quad  -e^{2} = e^{W_{0}(-2e^{-2})} \frac{W_{0}(-2e^{-2})}{2}.
\]
Using this, we obtain:
\begin{equation}\label{eq:lambert-trick}
    W_{0}\left(-e^{-\frac{W_{0}(-2e^2)}{2}-2}\right)
    = W_{0}\left(-e^{-\frac{W_{0}(-2e^2)}{2}} e^{-2}\right)
    = W_{0}\left(e^{\frac{W_{0}(-2e^2)}{2}} \frac{W_{0}(-2e^2)}{2} \right).
\end{equation}
We now compute $h(z_0)$:
\begin{align*}
&h(z_0) = z_0 - 1 - W_{0}(-e^{-z_0-1}) 
    = \frac{W_{0}\left(-2e^{-2} \right)}{2} - W_{0}\left(-e^{-\frac{W_{0}(-2e^2)}{2}-2} \right) \\
&= \frac{W_{0}\left(-2e^{-2} \right)}{2} - W_{0}\left(e^{\frac{W_{0}(-2e^2)}{2}} \frac{W_{0}(-2e^2)}{2} \right) 
    = \frac{W_{0}\left(-2e^{-2} \right)}{2} - \frac{W_{0}\left(-2e^{-2} \right)}{2} 
    = 0.
\end{align*}
Here, the third equality follows from \eqref{eq:lambert-trick} and the fourth equality follows by Definition \ref{def:lambert}.
Having established that $h(z_0) = 0$, we now show that the derivative $h'(z) > 0$ for all $z > z_0$. By Fact \ref{fact:lambert-derivative}, we have:
\[
    h'(z) = 1 - \left(W_{0}\left(-e^{-z-1}\right) \right)'
    = 1 + \frac{W_{0}\left(-e^{-z-1}\right)}{e^{-z-1} \left(1+W_{0}\left(-e^{-z-1}\right)\right)} \cdot \left(-e^{-z-1} \right)',
\]
and simplifying leads to:
\[
    h'(z) = \frac{1+2W_{0}\left(-e^{-z-1}\right)}{1+W_{0}\left(-e^{-z-1}\right)}.
\]
The strict inequality $h'(z) > 0$ holds because $W_{0}$ is the principal branch of the Lambert $W$ function, and since $W_{0}$ is non-decreasing on $\left[-\frac{1}{e},\infty\right)$, we have:
\[
    1 + 2W_{0}(-e^{-z-1}) > 1 + 2W_{0}(-e^{-z_0-1}) > 0.59,
\]
where the last inequality is verified numerically. The argument for the denominator is identical i.e., we numerically verify that $1 + W_{0}(-e^{-z_0-1})> 0.79$.

Since $h(z)$ is strictly increasing in $(z_0, 1]$ and $h(z_0) = 0$, it holds that $h(z) > 0$ for all $z > z_0$ or, equivalently, that $f(z) < z$.
For the equality, since $1-f(z)=-W_{0}(-e^{-z-1})>0$ (as $-e^{-z-1} > -\frac{1}{e}$ holds for $z>0$), $f(z)+z = -\ln(1-f(z))$ is equivalent to:
\begin{equation*}
    e^{f(z)+z} = \frac{1}{1-f(z)} \quad \Longleftrightarrow \quad e^{f(z)+z}\left( 1-f(z)\right)=1. 
\end{equation*}
Substituting $f(z)$ in this expression leads to:
\begin{align*}
    e^{f(z)+z}\left( 1-f(z)\right) & =-e^{z+1}\left(W_{0}\left(-e^{-z-1} \right)e^{W_{0}(-e^{-z-1})} \right)\\
    &=-e^{z+1}\left(-e^{-z-1} \right)=1,
\end{align*}
where the second equality follows by Definition \ref{def:lambert}.
\end{proof}

\subsection{Lower bounds}\label{app:LB-no-reserve-prices}
\lemuniversalLBbudget*
\begin{proof}[Proof of Theorem~\ref{lem:universal-LB-budget}]
Consider a $\fpar{2}{}$ with 2 agents of any types $\type_1$, $\type_2$ whose valuations are $\vec v_1=(1,1)$ and $\vec v_2=(0,1)$, respectively. Let the budget profile be $\budget = (1,\infty)$ and let ties be broken in favor of agent 1 .  Consider the deterministic bid profile $\b=(\b_1,\b_2)$ with $\b_1=(0,1)$ and $\b_2=(0,1)$. 

We note that $\b\in \mathcal R_1$ and $\b\in \mathcal R_2$. For agent 2, it is clear that there is no mixed deviation that could improve their gain while still satisfying the ROI and budget constraint. Agent 1 wins both items and their gain is $\z_1(\b)=2-\type \geq1$. Thus, for agent 1 to bid deterministically $\b_1$ is a best response to $\b_2$, showing that $\b$ is a mixed Nash equilibrium. The liquid welfare at this equilibrium we obtain:
$$\sw(\b)=\min(v_1(\vec x_1(\b)),\budget_1)+\min(v_2(\vec x_2(\b)),\budget_2)=1,$$
while optimally, agent 1 receives item 1 and agent 2 receives item 2 yielding: 
$$\opt=\min(v_1(1,0),\budget_1)+\min(v_2(0,1),\budget_2)=2.$$ 
Hence, we conclude that the $\text{\mne-\poa}(\mathcal I_{\add}^{T})\geq 2$.
\end{proof}

In proving the lower bounds of Theorems \ref{thm:LB-POA-budget-commontype} and \ref{thm:LB-POA-budgetfree-hybrid} the following claim is helpful.

\begin{restatable}{claim}{claimdomain}\label{claim-domain}
    For every $\type\in \left(1+\frac{W_0(-2e^{-2})}{2}, 1\right]$, it holds that $-W_0(-e^{-\type-1})\in (1-\type,\nicefrac{1}{e}]$. Furthermore, for $f(a)=\frac{1-a+a\ln(a)+\type}{1-a+a\ln(a)+a\type}$ with $a\in (1-\type,1)$, we have that:
    $$f(-W_0(-e^{-\type-1}))= 1 + \frac{\type }{1+W_{0}\left(-e^{-\type-1}\right)}$$
    \end{restatable}

\begin{proof}[Proof of Claim~\ref{claim-domain}]
The fact that $-W_0(-e^{-\type-1}) > 1 - \type$ holds for every $\type \in  \left(1+\frac{W_0(-2e^{-2})}{2}, 1\right]$ follows directly from Claim \ref{claim:LW:prop0}. In the interest of brevity, we show that $-W_0(-e^{-t-1})\leq \frac1e$ for $t\in \left(1+\frac12W_0(-2e^{-2}),1\right]$. In Figure~\ref{fig:claimdomain} we indeed see that for $\type\in \left(1+\frac{W_0(-2e^{-2})}{2}, 1\right]$ that $\frac1e>-W_0(e^{-\type-1})$. For completeness sake, the line $1-\type$ is also plotted in the graph from which we visually confirm the first part of the claim.

\begin{figure}[h]
    \centering
    \begin{tikzpicture}
\centering
      \begin{axis}[
      scale only axis,
      width=0.4\textwidth,
      height=0.2\textwidth,
      xmin=0.36788, xmax=1,
      ymin=0, ymax=0.5,
      axis x line=bottom,
      axis y line=left,
      axis line style={line width=0.8pt},
      tick style={semithick},
      ticklabel style={font=\scriptsize},
      xlabel={\small$\type$},
      xlabel style={yshift=-2pt},
      ylabel style={xshift=1pt},
      xtick={0.36788,0.79681213,1},
      xticklabels={\scriptsize$\frac1e$,\scriptsize$1+\frac12W_0(2e^{-2})$,\scriptsize$1$},
      ytick={0,0.36788},
      yticklabels={\scriptsize$0$,\scriptsize$\frac1e$},
      clip=false,
      legend style={at={(0.65,1.21)},anchor=north west},
    ]
    \addplot[red!80!black, domain=0.1589:0.36788] (x - ln(x) -1, x);
    \addlegendentry{\scriptsize\(-W_0(e^{-\type-1})\)}
    \addplot[thick,blue] coordinates {
        (0.36788,          0.36788)
        (1, 0.36788)
      };
      \addlegendentry{\scriptsize\( 1/e\)}
    \addplot[green, domain=0.632120558829:1, samples = 100]{1-x};
    \addlegendentry{\scriptsize\(1-\type\)}
  \end{axis}
\end{tikzpicture}
    \caption{A graph to compare $-W_0(-e^{-\type-1})$ (in red) to the constant $\frac1e$ (in blue) to the line $1-\type$ (in green).}
    \label{fig:claimdomain}
\end{figure}

For the second part of the claim, we use the Claim~\ref{claim:LW:prop0} and obtain the following:

{\allowdisplaybreaks\begin{align*} f\left(-W_0\left(-e^{-\type-1}\right)\right)&=\frac{1+W_0\left(-e^{-\type-1}\right)-W_0\left(-e^{-\type-1}\right)\ln\left(-W_0\left(-e^{-\type-1}\right)\right)+\type}{1+W_0\left(-e^{-\type-1}\right)-W_0\left(-e^{-\type-1}\right)\ln\left(-W_0\left(-e^{-\type-1}\right)\right)-\type W_0\left(-e^{-\type-1}\right)}\\
    &=\frac{1+W_0\left(-e^{-\type-1}\right)-W_0\left(-e^{-\type-1}\right)\left(-\type-1-W_0\left(-e^{-\type-1}\right)\right)+\type}{1+W_0\left(-e^{-\type-1}\right)-W_0\left(-e^{-\type-1}\right)\left(-\type-1-W_0\left(-e^{-\type-1}\right)\right)-\type W_0\left(-e^{-\type-1}\right)}\\
    &=\frac{1+2W_0\left(-e^{-\type-1}\right)+\type W_0\left(-e^{-\type-1}\right)+W_0\left(-e^{-\type-1}\right)^2+\type}{1+2W_0\left(-e^{-\type-1}\right)+ W_0\left(-e^{-\type-1}\right)^2}\\
    &=\frac{\left(1+\type W_0\left(-e^{-\type-1}\right)\right)^2+\type \left(1+W_0\left(-e^{-\type-1}\right)\right)}{\left(1+\type W_0\left(-e^{-\type-1}\right)\right)^2}=1+\type\left(1+W_0\left(-e^{-\type-1}\right)\right)^{-1}.
    \end{align*}}

\end{proof}

\LBPOAbudgetcommontype*
\begin{proof}[Proof of Theorem~\ref{thm:LB-POA-budget-commontype}]
We distinguish between the two cases.\label{proof:lb-POA-budget-commontype}
\medskip 

\noindent
\textbf{Case 1:} $\type\in \left[0,1 + \frac{W_{0}(-2e^{-2})}{2}\right]$. This case follows immediately from Theorem~\ref{lem:universal-LB-budget}. 
\medskip

\noindent
\textbf{Case 2:} $\type\in \left(1 + \frac{W_{0}(-2e^{-2})}{2},1\right]$.  Consider a $\fpar{2}{}$ with 2 agents of type $\type>1 + \frac{W_{0}(-2e^{-2})}{2}$, valuations $\vec v_1=(1,1)$ and $\vec v_2=(0,1)$. Next, fix $a\in(1-\type, 1)$ (to be determined later) and consider the budget profile $\budget=\left(\frac{1-a+a\ln(a)}{\type},\infty\right)$. Let ties in the auction of item 1 be broken in favor of agent 1 and in the auction of item 2 only be broken in favor of agent 2 when the tie occurs at 0. Consider now the randomized bid profile $\B$ where in the auction of item 1 both agents bid deterministically 0, but in the auction of item 2 they bid $(x,x)$, where $x$ is sampled from a distribution $X$ with CDF $F_X(x)=\frac{a}{1-\type x}$ with $\supp (X)=\left[0,\frac{1-a}{\type}\right]$. We note that the PDF is $f_X(x)=\frac{a\type}{(1-\type x)^2}$ for $x\in \left[0,\frac{1-a}{\type}\right]$. 

First we show that the ROI and budget constraints are satisfied. Since the payment of agent 2 never exceeds 0, we see that for this agent, the ROI and budget constraint are satisfied. For agent 1, the expected payment is:
    $$\EX[p_1(\B)]=\int_0^{\frac{1-a}{\type}}xf_X(x)dx=\frac{1-a+a\ln(a)}{\type}=\budget_1.$$
    For the expected value, note that agent 1 will always win item 1 and wins item 2 with probability $1-a$. Thus, the expected value is $2-a$. Since $a\in (1-\type,1)$, we have that $1-a\leq \type$. Combining this with $a<1$, we find that $a\ln(a)<0$ showing that $\EX_{\b\sim \B}[p_1(\b)]\leq 1\leq 2-a=\EX_{\b\sim \B}[v_1(\vec x_1(\b))]$. Thus, we conclude that $\B\in \mathcal R_1$.

    We now verify that $\B$ is a CCE. Agent 2 only wins item 2 when both agents bid 0 in auction 2. This occurs with probability $a$ and their payment would be 0 yielding an expected gain:
    $$\EX[\z_2(\B)]=a.$$
    Since agent 2's value of item 1 is 0, they cannot strictly improve their gain by deviating in this auction. So, suppose they deviate by bidding $\b_2'=(0,b)$ for some $b\in \left[0,\frac1\type\right]$. For any $b>\frac1\type$, their gain would be negative if they win the item. So, those deviations we need not to consider. If $b\in\left(\frac{1-a}{\type},\frac1\type\right]$, then agent 2 will always win the item and their gain would be $\EX[\z_2(\B_1,\b_2')]=1-\type b<a$, as $b>\frac{1-a}{\type}$. If instead $b\in \left[0,\frac{1-a}{\type}\right]$, then the expected gain of agent 2 becomes:
    $$\EX[\z_2(\B_1,\b_2')]=(1-\type b)F_X(b)=a.$$
    Hence, any pure deviation of agent 2, not just ROI restricted pure deviations, cannot improve their gain showing that no mixed deviation, not just ROI restricted, could do so as well. More formally, let $\B_2'$ be any deviation such that $(\B_1,\B_2')\in \mathcal R_2$, then:
    $$\EX\left[\z_2(\B_1,\B_2')\right]\leq a=\EX\left[\z_2(\B)\right].$$

    For agent 1, the expected gain under $\B$ is:
    $$\EX[\z_1(\B)]=\EX[v_1(\vec x_1(\B))] -\type \EX[p_1(\B)]=(2-a) - \type\cdot \frac{1-a+a\ln(a)}{\type} =1-a\ln(a).$$
    Since agent 1 wins auction 1 regardless of their bid, they cannot improve their gain by deviating in that auction. Thus, suppose that they deviate by bidding $\b_1'=(0,b)$ for some $b\in [0,\frac1\type]$. For an analogous reason as before, we need not to consider $b>\frac1\type$. If $b\in \left(\frac{1-a}{\type},\frac1\type\right]$, then agent 1 wins both items with probability 1 and their gain is:
    $$\EX[\z_1(\b_1',\B_2)]=2-\type b<1+a.$$
    Regardless of whether $(\b_1',\B_2)\in \mathcal R_1$, we find that when $1-a\ln(a)\geq 1+a$, then agent 1 would not deviate in this manner. We note that this inequality holds whenever $a\leq \frac1e$. Similarly, if $b\in \left[0,\frac{1-a}{\type}\right]$, then:
    $$\EX[\z_1(\b_1',\B_2)]=1\cdot (1-F_X(b))+(2-tb)F_X(b)=1+a.$$
    Hence, for $a\leq \frac1e$, we conclude that there exists no pure deviation for agent 1 that could strictly improve their gain. Thus, similarly to agent 2, no mixed deviation could achieve this and we conclude that $\B$ is a CCE. The liquid welfare at this equilibrium is:
    $$\sw(\B)=\min\left(\EX \left[v_1(\vec x_1(\B))\right],\budget_1\right)+\min\left(\EX\left[v_2(\vec x_2(\B))\right],\budget_2\right)=\frac{1-a+a\ln(a)}{\type}+a.$$
    However, optimally, item 1 would be allocated to agent 1, while item 2 would be allocated to agent 2. This would yield a welfare of:
    $$\opt=\min(v_1(1,0),\budget_1)+\min(v_2(0,1),\budget_2)=\frac{1-a+a\ln(a)}{\type}+1.$$
    Therefore, we obtain:
    $$\cce\text{-}\poa(\mathcal I_{\add}^{\{t\}})\geq \frac{\opt}{\sw(\B)}=\frac{1-a+a\ln(a)+\type}{1-a+a\ln(a)+a\type},$$
    for any $a\in (1-\type,1)$ and $a\leq \frac1e$. Using Claim \ref{claim-domain}, we find that $a^\ast=-W_0(-e^{-\type-1})\in (1-\type,1)$ and $a^\ast\leq \frac1e$. Plugging in this choice for $a$ above, we obtain the desired bound by Claim~\ref{claim-domain}.
\end{proof}

\LBPOAbudgetfreehybrid*
\begin{proof}[Proof of Theorem~\ref{thm:LB-POA-budgetfree-hybrid}]
    We distinguish between two cases.
\label{proof:LB-POA-budgetfree-hybrid}
    \medskip

    \noindent
    \textbf{Case 1:} $\type\in \left(0,1 + \frac{W_{0}(-2e^{-2})}{2}\right]$. Consider a $\fpar{2}{}$ with 2 agents. Agent 1 being a value maximizer and agent 2 having a price sensitivity of $\type$. Let $\vec v_1=(1,0)$ and $\vec v_2=(0,1)$ and suppose ties are broken in favor of agent 1. Consider the pure bid profile $\b=(\b_1,\b_2)$ with $\b_1=(0,1)$ and $\b_2=(0,1)$. It is clear that both players under $\b$ satisfy the ROI constraint. Furthermore, agent 1 has no pure deviation that would improve their value and thus there exists no mixed deviation that could improve their value. For agent 2, we note that no pure deviation satisfying the ROI constraint could yield them a non-zero gain. Therefore, no mixed deviation satisfying the ROI constraint exists for agent 2 that would improve their gain. Thus, $\b$ is a mixed Nash equilibrium and for the welfare at this equilibrium we see:
    $$\sw(\b)=v_1(\vec x_1(\b))+v_2(\vec x_2(\b))=v_1(1,1)+v_2(0,0)=1.$$
    The optimal welfare, however, is achieved when item 1 is allocated to agent 1 and item 2 to agent 2. This gives:
    $$\opt=v_1(1,0)+v_2(0,1)=2.$$
    Hence, we see that $\opt/\sw(\b)=2$, giving the first lower bound.

    \medskip

    \noindent
    \textbf{Case 2:} $\type \in \left(1 + \frac{W_{0}(-2e^{-2})}{2},1\right]$. Consider a $\fpar{2}{}$ with 2 agents. Agent 1 being a value maximizer and agent 2 having price sensitivity $\type$. Let $a\in (1-\type,1)$, $\vec v_1=\left(\frac{1-a+a\ln(a)}{\type},0\right)$ and $v_2=(0,1)$. In auction 1, let ties be broken in favor of agent 1 and in auction 2 we let ties be broken in favor of agent 2 if and only if the tie occurs at 0. 

    Next, we consider the following mixed bid profile. Agent 1 bids $\b_1=(0,x)$ where $x$ is sampled from $X$ with CDF $F_X(x)=\frac{a}{1-\type x}$ and $\supp(X)=\left[0,\frac{1-a}{\type}\right]$. Agent 2 bids deterministically $\b_2=(0,0)$. Let $\B$ denote this randomized bid profile. From the analysis in Theorem \ref{thm:LB-POA-budget-commontype}, we know that the expected payment of agent 1 in auction 2 is $\frac{1-a+a\ln(a)}{\type}$, thus showing that the ROI constraint is satisfied by both agents. Since agent 1 is a value maximizer, we easily see that there exists no deviation that would improve their value. As $\type>0$ and agent 2 does not value item 1, we only have to consider pure deviations of the form $\b_2'=(0,b)$ for some $b\in [0,1]$. An analogous analysis to the second case in the proof of Theorem~\ref{thm:LB-POA-budget-commontype}, shows that there is no pure deviation that satisfying the ROI constraint that would improve the gain of agent 2. As the gain is already 0 of agent 2, we see for agent 2 that there is no mixed deviation satisfying the ROI constraint that would improve their gain. Hence, $\B$ is a mixed Nash-equilibrium. 

    We note that $\opt = \frac{1-a+a\ln(a)}{\type}+1$ and that the liquid welfare at equilibrium is $\frac{1-a+a\ln(a)}{\type}+a$, as $a$ is the expected gain of agent 2 under $\B$. Thus, we find for the coarse correlated price of anarchy the following bound:
    $$\mne\text{-}\poa(\mathcal I_{\add}^{\{0,t\},\infty})\geq\frac{\opt}{\sw(\B)}= \frac{1-a+a\ln(a)+\type}{1-a+a\ln(a)+a\type},$$
    for $a\in (1-\type,1)$. Similar to Theorem \ref{thm:LB-POA-budget-commontype}, applying Claim \ref{claim-domain} and plugging in $a=-W_0(-e^{-\type-1})$ yields the desired bound.
    \end{proof}
\subsection{Proof of Theorem~\ref{theorem:not-so-hybrid}}\label{app:not-so-hybrid}
\notsohybrid*
\begin{proof}[Proof of Theorem~\ref{theorem:not-so-hybrid}]
Set $\omega = 1 - e^{-\nicefrac{1}{\min(T)}}$. 
Note that $\omega \leq \min(T)$.
To see this, consider $h(z) = z - 1 + e^{-\frac{1}{z}}$. For all $z \in [\beta, 1]$, it holds that $h(z) \geq 0$. Firstly, it holds that$h(\beta) = \beta -1 +e ^{-\frac{1}{\beta}}=0$ by definition of $\beta$. Also, the function $h$ is non-decreasing for all $z \geq \beta$, as $h'(z)=1 +\frac{e^{-\frac{1}{z}}}{z^2}>0$ in this case. Therefore, $h(\min(T)) = \min(T) - \omega \ge 0$.

For $\vec{\mu} := \vec{\mu}^*(\omega)$ and $\vec{\lambda} := \vec{\lambda}^{*}(\omega)$ as defined in \eqref{eq:mu-star-params} and property \textit{(i)} of Lemma \ref{lemma:technical-lemma}, respectively, it holds that: 
\begin{align}
\max_{\substack{\type \in T}} \frac{\mu_{\type}}{\lambda_{\type}} 
    + \max_{\substack{\type \in T}} \frac{1 - t }{\lambda_{\type}} &=\max_{\substack{\type \in T:\\ t > 0}} \frac{\mu_{\type} }{\lambda_{\type}} 
    + \max_{\substack{\type \in T:\\ t > 0}} \frac{1 - t }{\lambda_{\type}}= \frac{\max(T)-\ln(1-\omega)(1-\min(T))}{\omega} \nonumber\\
    &\leq \frac{1-\ln(1-\omega)(1-\min(T))}{\omega} = \frac{-\ln(1-\omega)}{\omega} \nonumber + \frac{1+\ln(1-\omega)\min(T)}{\omega} 
 =\frac{-\ln(1-\omega)}{\omega}.\numberthis \label{eq:bound-second-part-O}
\end{align}
Here, the first equality holds as there are no value maximizers in $T$ by assumption. 
The second equality follows from properties \textit{(iii)} and \textit{(iv)} of Lemma \ref{lemma:technical-lemma}, as $\omega \leq \min(T)$. Finally, the last equality follows since $1 + \ln(1-\omega)\min(T) = 1 + \ln\left(e^{\frac{-1}{\min(T)}}\right)\min(T)=0$ holds by our choice of $\omega$. 
We conclude that:
\begin{align*}
    \left(\cce\textsc{-}\poa)(\mathcal{I}_{\xos}^{T, \infty})\right)^{-1} \geq \poarmp(T^+)= \poarmp(T) &\geq \min \left\{ 
    \min_{\substack{\type \in T}} \lambda_{\type}, 
    \left( \max_{\substack{\type \in T}} \frac{\mu_{\type}}{\lambda_{\type}} 
    + \max_{\substack{\type \in T}} \frac{1 - t}{\lambda_{\type}} 
        \right)^{-1} \right\}\\
    &=  \min \left\{ 
        \frac{\omega}{-\ln(1-\omega)}, 
        \left( 
            \max_{\substack{\type \in T}} \frac{\mu_{\type}}{\lambda_{\type}} 
            + \max_{\substack{\type \in T}} \frac{1 - t}{\lambda_{\type}} 
        \right)^{-1}\right\}\\
         &=  \frac{\omega}{-\ln(1-\omega)} = \min(T)\cdot \left(1 - e^{-\frac{1}{\min(T)}}\right) .
\end{align*}
Here, the first inequality follows by Theorem \ref{theorem:templateII}. The first equality holds since, by assumption, we only consider budget-free instances. Then, the second equality follows from property \textit{(ii)} of Lemma \ref{lemma:technical-lemma}. Finally, the third equality follows from \eqref{eq:bound-second-part-O}. The proof follows. 
\end{proof}

\section{Missing material of Section~\ref{sec:reserve-prices}}

\subsection{Proof of Theorems~\ref{theorem:single-type}, \ref{lemma:LB-POA-Reserve-value-maximizers} and \ref{lemma:LB-BigTauSens-withReserve}}

\theoremsingletype*
\begin{proof}[Proof of Theorem \ref{proof:theorem:single-type}]\label{proof:theorem:single-type}
To determine a lower bound on $\poarmp(\{t\})$, we need to determine a pair $(\lambda, \mu)$ that satisfies \eqref{eq:key-lemma-0-constraint}, \eqref{eq:key-lemma-1-constraint} or \eqref{eq:key-lemma-2-constraint}, depending on the type $\type \in[0,1]$ and the value of the parameter $\eta \in [0,1)$. We distinguish three cases.

\medskip
\noindent \textbf{Case 1:} $t \in (1-\nicefrac{1}{e},1)$ and $\eta \in \left[0, \frac{1-e(1-t)}{t} \right)$. Set $(\lambda, \mu) =\left(1-\frac{1-t \eta}{e}, t\right)$. Note that for this choice of $(\lambda, \mu)$ it holds that, by the definition of \textbf{Case 1}, that
\begin{equation*}
    \mu = t \cdot \ln\left( \frac{1-t \eta}{1-t}\right)^{-1} \cdot \ln\left( \frac{1-t \eta}{1-t}\right) \geq  t \cdot \ln\left( \frac{1-t \eta}{1-t}\right)^{-1}.
\end{equation*}
Therefore, $(\lambda, \mu)$ satisfy the constraints in \eqref{eq:key-lemma-0-constraint} and \eqref{eq:key-lemma-1-constraint} and $(\lambda, \mu)$ is a feasible solution of $\poarmp(\{t\})$. Hence, we obtain that
\begin{equation*}
    \poarmp(\{t\}) \geq \min \left(\lambda, \frac{\lambda}{\mu+1-t} \right)= \lambda = 1-\frac{1-t \eta}{e} = P_t(\eta)^{-1}
\end{equation*}
\medskip
\noindent
\textbf{Case 2:} $t \in (0,1)$ and  $\eta \in \left(\max\left(0, \frac{1-e(1-t)}{t} \right), 1\right)$. Set: $ (\lambda, \mu)= \left( t\ln\left( \frac{1-t \eta}{1-t}\right)^{-1},  t\ln\left( \frac{1-t \eta}{1-t}\right)^{-1}\right).$ Observe that $(\lambda, \mu)$ satisfy the constraint in \eqref{eq:key-lemma-1-constraint} (with equality). We therefore obtain:
\begin{equation*}
    \poarmp(\{t\}) \geq \min \left( \lambda, \frac{\lambda}{\mu + 1 -t} \right)= \frac{\lambda}{\mu+1-t} = \frac{t}{t+(1-t)\ln\left(\frac{1-t \eta}{1-t} \right)}=P_t(\eta)^{-1}.
\end{equation*}
Here, the second equality follows since, by the definition of \textbf{Case 3}, it holds that 
\begin{equation*}
    \mu = \frac{t}{\ln \left(\frac{1 - t \eta}{1-t} \right)} \geq \frac{t}{\ln \left(\frac{1 - t \max\left(0, \frac{1-e(1-t)}{t} \right)}{1-t} \right)} \geq t.
\end{equation*}

\medskip
\noindent
\textbf{Case 3:} $t= 0$. 
Set $\mu= (1-\eta)^{-1}$, which clearly satisfies the constraint in $\eqref{eq:key-lemma-2-constraint}$ and prescribes that $\lambda = (1- \eta)^{-1}$. We obtain:
\begin{align*}
    \poarmp(\{t\}) \geq \min \left\{ \lambda, \frac{\lambda}{\mu  +1}\right\} &= \min\left\{\frac{1}{1-\eta}, \frac{(1-\eta)^{-1}}{1+ (1-\eta)^{-1}} \right\} 
    = \min\left\{\frac{1}{1-\eta}, \frac{1}{2 -\eta} \right\} = \frac{1}{2-\eta} = P_{t}(\eta)^{-1}.
\end{align*}
Finally, by Theorem \ref{theorem:templateII} we have that $\cce\textit{-}\poa(\mathcal{I}_{\xos}^{\{t\}, \infty}) \leq \left(\poarmp(\{t\})\right)^{-1}$ holds for well-supported equilibria. Combining this fact with the lower bounds on $\poarmp(\{t\})$ we obtained for each of the three cases above, the claim follows.
\end{proof}

\LBPOAReservevaluemaximizers*

\begin{proof}[Proof of Theorem~\ref{lemma:LB-POA-Reserve-value-maximizers}]\label{proof:LB-POA-Reserve-value-maximizers}
Consider a $\fpar{2}{\vec{r}}$ among two value maximizing agents with additive valuations and feasible reserve prices. Assume that ties are broken in favor of agent 1 for both auctions. 
For auctions $1$ and $2$, consider a small $\varepsilon >0$ and let $\eta_1 \in [0,1-\varepsilon]$ and $\eta_2 = 1 - \varepsilon \ge 0$. Note that $\eta = \eta_1$. 
Let $v_{11} = 1$ and $v_{12} = 0$ and let $v_{21} = 0$ and $v_{22} = 1-\eta_1$. 
Note that the reserve prices are feasible and equal to $r_1 = \eta_1$ and $r_2 = (1 - \varepsilon)(1-\eta_2)$. 
Consider a bid profile $\vec{b}$ with $\vec{b}_1 = (\eta_1, 1-\eta_1)$ and $\vec{b}_2 = (0, 1 - \eta_1)$. 
We show that $\vec{b}$ is a MNE.

First note that under $\vec{b}$ the ROI constraints are satisfied for both agents, as agent 1 wins both items with a total payment of 1. 
Also note that agent 2 cannot win anything by deviating without violating their ROI constraint. This is because agent 1 deterministically bids the value $v_{22}$ of agent 2 for item $2$ and ties are broken in favor of agent 1, and agent 2 has 0 value for item $1$. As agent 1 wins both items under $\b$, agent 1 cannot improve their gain with any (random) unilateral deviation. Therefore, $\b$ is a MNE.

It is easy to see that the optimal liquid welfare in this instance is achieved when agent 1 is allocated item $1$ and agent 2 is allocated item $2$, and is equal to:
$$\opt=v_1(1,0)+v_2(0,1)=2 - \eta_1 = 2 - \eta.$$
The liquid welfare of the MNE $\vec{b}$ is equal to $\sw(\b)=v_{11} + v_{12} = 1$. 
Therefore, the \mne\text{-}\poa\ of this instance is to $2 - \eta$, concluding the proof.
\end{proof}

\LBBigTauSenswithReserve*
\begin{proof}[Proof of Theorem~\ref{lemma:LB-BigTauSens-withReserve}]
\label{proof:LB-BigTauSens-withReserve}
Consider a first-price auction with reserve price $r$ among $n = 2$ agents of the same type, i.e., $\type := \type_1 = \type_2$.
Let $\vec{v} = (1, \eta)$, with $\eta \in [0,1)$, and let $\type = \frac{e-1}{e - \eta}$. 
The feasible reserve price is $r = \eta < v_1$.
We assume that ties are broken in favor of agent $1$ for every bid profile $\vec{b} = (b, b)$ with $b > \eta$, and in favor of agent $2$ for $\vec{b} = (\eta, \eta)$.
Let $X$ be a random variable with the following CDF and PDF:
\[
F_X(x) = \frac{1}{e} \cdot \frac{1 - \type \eta}{1-\type x} \ \  \text{ and } \ \  f_X(x) = \frac{\type}{e} \cdot \frac{1 - \type \eta}{(1-\type x)^{2}} \ \  \text{ with } \ \  x \in \left[\eta ,\frac{e-1 + \type \eta}{e\type}\right].
\]
Note that the domain is well-defined as $ x \ge 0$, and it is easy to verify that $F_{X}(\cdot)$ is non-negative and increasing over the domain and $F_{X}( (e - 1 + \type \eta)/ e \type ) = 1$.
Now consider the randomized bid profile $\B$ in which the agents always bid identically according to $X$, i.e., draw some value $x$ from $X$ and let both agents bid $x$. 
We prove that $\B$ is a CCE.

First, we verify that $\vec{B}$ satisfies the ROI constraints of the two agents. Note that this is trivially true for agent 2 as, by the tie-breaking rule, agent $2$ only wins the item for the bid profile $(\eta, \eta)$ and $\eta = v_2$.  
Note that $\vec{B} \in \mathcal{R}_1$ since:
\begin{equation*}
\EX [ p_1(\B)] \le \frac{e-1 + \type \eta}{e\type} \cdot \EX[x_{1}(\B)] \leq \EX [x_{1}(\B)] 
= \EX [v_1(x_1(\B))],
\end{equation*}
where the first inequality holds as agent 1 never bids above $\frac{e-1 + \type\eta}{e\type}$, while the second inequality holds by assumption. More precisely, note that the expression $\frac{e-1 + \type \eta}{e\type}$ is decreasing in $\type$ and plugging in $\type=\frac{e-1}{e-\eta}$ gives 1, as desired.

Secondly, we show that no agent can improve by unilaterally deviating. 
First, note that under $\B$ agent 2 only wins when bidding bid $\eta$, and that the gain of agents 2 is non-negative in this case, i.e., $v_2 - \type \eta = \eta(1- \type) \ge 0$.
We consider deterministic deviations of agent 2.
Agent 2 would never win when bidding $b_2 < \eta = r$. 
If agent 2 wins when bidding $b_2 > \eta$, agent 2 violates their ROI constraint.
And if agents 2 wins when bidding $b_2 = \eta$, their ROI constraint is tight.
Therefore, agent 2 has no randomized unilateral deviation that is beneficial and satisfies the ROI constraint. 

As ties are broken in favor of agent $1$ for every bid profile $\vec{b} = (b, b)$ with $b > \eta$, the expected gain of agent $1$ under $\vec{B}$ is:
\begin{align*}
    \EX [\z_1(\B)] 
    &= \int_{\eta}^{\frac{e-1 + \type \eta}{e \type}} (1 - \type x) f_{X}(x) d x
    = \int_{\eta}^{\frac{e-1 +\type \eta}{e \type}} \frac{\type}{e} \frac{1 - \type \eta}{1 - \type x} d x \\
    &= \frac{1 - \type \eta}{e} \bigg[ - \ln(1 - \type x) \bigg]_{\eta }^{\frac{e-1 + \type \eta}{e \type}} = \frac{1 - \type \eta}{e}.
\end{align*}
Consider a deterministic deviation $b_1'$ of agent $1$.
Clearly, $b_1' \le \eta$ leads to an expected gain of $\EX [\z_1(b'_1,X)] = 0$, as agent $2$ always wins the item at $(\eta,\eta)$. 
Moreover, note that bidding $b_1' > \frac{e-1 + \type \eta}{e\type}$ is not beneficial. In this case agent 1 would always win the item but:
\begin{equation*}
    \EX [\z_1(b'_1, X)] = 1 - \type b'_1 < 1 - \type \frac{e-1 + \type \eta}{e\type} = \frac{1 - \type \eta}{e} = \EX [\z_1(\B) ].
\end{equation*}
Finally, consider $b_1' \in  (\eta, \frac{e-1 + \type \eta}{e\type}]$.
We obtain that:
\begin{equation*}
    \EX [\z_1(b'_1, X)] = \ppr [b'_1 \ge X] ( 1 - \type b_1') = F_X(b'_1)(1-\type b_1') = \frac{1 - \type \eta}{e} = \EX [\z_1(\B)].
\end{equation*}
Therefore, we can conclude that for any randomized unilateral deviation $Y$ with PDF $f_{Y}$ of agent 1 satisfying the ROI constraint, it holds that:
\[ \EX [ \z_1(Y, X) ] = \int_{0}^{\infty} \EX [\z_1(y,X)] \cdot f_{Y}(y) d y  \le \int_{0}^{\infty} \frac{1 - \type \eta}{e} f_{Y}(y) d y = \frac{1 - \type \eta}{e},
\]
proving that $\vec{B}$ is a CCE.

We conclude the proof by showing how inefficient the CCE $\vec{B}$ is.
Clearly, the optimal liquid welfare is $v_1 = 1$, and therefore:
\begin{equation*}
    \cce\text{-}\poa \geq \frac{\opt}{\sw(\B)} =\frac{1}{\EX [v_1(x_1(\B)) + v_2(x_2(\B))]} = \frac{1}{1 - F_X(\eta) + \eta F_X(\eta) } = \frac{e}{e-1+ \eta}.
\end{equation*}
\end{proof}

\subsection{Proof of Theorem \ref{thm:hybrid-reserve-prices}} \label{app:thm:hybrid-reserve-prices}
\thmhybridreserveprices*

The steps of our approach to proving Theorem~\ref{thm:hybrid-reserve-prices} are similar to those in Section~\ref{subsec:solving-poarmp}. In particular, the focus is to obtain a feasible solution to $\poarmp(T)$ that yields strong POA upper bounds, this time concentrating on the mixed-agent setting, i.e., $T = \{0,1\}$. However, handling the relative gap parameter $\eta$ analytically requires a somewhat different perspective. 

For $\eta \in [0,1)$, let $\zeta(\eta) = 2 - \eta + W_{0}\left(-\frac{(1-\eta)^2}{e^{2-\eta}}\right)$.
In Definition~\ref{def:mu:bar}, we present the feasible solution to the mathematical program that we work with.

\begin{definition}\label{def:mu:bar}
Given a set of types $T = \{0,1\}$ and a parameter $\eta \in [0,1)$, define $\bar{\vec{\mu}}(\eta) \in \mathbb{R}_{>0}^{2}$ such that:
\begin{equation} \label{eq:mu-for-mixed-agents-with-reserve}
    \bar{\mu}_{\type}(\eta) =
    \begin{cases}
        \frac{1}{1 - \eta},  &\text{if } \type = 0, \\
        \frac{1}{\zeta(\eta)}, & \text{if } \type = 1.
    \end{cases}
\end{equation}
\end{definition}
We first show in Lemma \ref{cor:mu-bar:feasible} that $\vec{\bar{\mu}}(\eta)$ is indeed feasible for all $\eta \in [0,1)$, using Fact~\ref{claim:zeta-range}.

\begin{fact} \label{claim:zeta-range}
For all $\eta \in [0,1)$, it holds that $\zeta(\eta) > 0$.
\end{fact}

Note that Fact~\ref{claim:zeta-range} holds because the function $-\frac{(1-z)^2}{e^{2-z}}$ is non-decreasing for $z \in [0,1)$, the function $W_{0}(\cdot)$ is also non-decreasing, and we have $W_{0}(-e^{-2}) \ge -0.16$.

\begin{lemma}\label{cor:mu-bar:feasible}
For all $\eta \in [0,1)$, the vector $\bar{\vec{\mu}}(\eta)$ is a feasible solution of $\poarmp(\{0,1\})$.
\end{lemma}

\begin{proof}
For agents of type $\type = 0$, the constraint in~\eqref{eq:key-lemma-2-constraint} is satisfied by construction. For agents of type $\type = 1$, the constraint in~\eqref{eq:key-lemma-0-constraint} holds because, by Fact~\ref{claim:zeta-range}, we have $\zeta(\eta) > 0$ for all $\eta \in [0,1)$.
\end{proof}

In Lemma~\ref{claim:technical-claim-hybrid-reserves}, we establish additional properties useful for the proof of Theorem~\ref{thm:hybrid-reserve-prices}, which follows.

\begin{restatable}{lemma}{claimtechnicalclaimhybridreserves}\label{claim:technical-claim-hybrid-reserves}
The following properties hold for $T=\{0,1\}$ and every $\eta \in [0,1)$. We write $\zeta:= \zeta(\eta)$ for brevity. 
\begin{compactenum}[(i)] 
    \item 
    $
    \min_{\type \in T}\bar{\lambda}_{\type}(\eta) \geq \frac{1-(1-\eta)e^{-\zeta}}{\zeta}.
    $
    \item 
    $
    \max_{\type \in T}\frac{\bar{\mu}_{\type}(\eta)}{\bar{\lambda}_{\type}(\eta)} \leq \frac{1}{1-(1-\eta)e^{-\zeta}}.
    $
    \item 
    $
    \max_{\type \in T}\frac{1 -{\type }}{\bar{\lambda}_{\type}(\eta)} \leq 1-\eta.
    $
    \item
    $
    \zeta = 2-\eta - (1-\eta)^2e^{-\zeta}.
    $
\end{compactenum}
\end{restatable}
\begin{proof}[Proof of Lemma \ref{claim:technical-claim-hybrid-reserves}]\label{proof:technical-claim-hybrid-reserves}
By \eqref{eq:key-lemma-0-constraint} and \eqref{eq:key-lemma-2-constraint}, it holds that:
\[
\bar{\lambda}_{\type}(\eta) =
\begin{cases} 
    \frac{1}{1-\eta} & \text{if } t = 0, \\
    \frac{1-(1-\eta)e^{-\zeta}}{\zeta} & \text{if } t = 1.
\end{cases} 
\]
For brevity, let $\bar{\lambda}_{\type} = \bar{\lambda}_{\type}(\eta)$ and $\bar{\mu}_{\type} = \bar{\mu}_{\type}(\eta)$ for each $\type \in T$.

\noindent \textit{(i)} For $\type=1$, we have:
\begin{align*}
        \bar{\lambda}_{\type} &= \frac{1-(1-\eta)e^{-\zeta}}{\zeta}=\frac{1}{1-\eta} \cdot \frac{1-\eta-(1-\eta)^2e^{-\zeta}}{\zeta} \\
        &=\frac{1}{1-\eta}\cdot \frac{1-\eta-(1-\eta)^2e^{-W_{0}\left(-\frac{(1-\eta)^2}{e^{2-\eta}} \right)} e^{-2+\eta}}{\zeta}\\
        &=\frac{1}{1-\eta}\cdot \frac{1-\eta+W_{0}\left(-\frac{(1-\eta)^2}{e^{2-\eta}} \right)}{\zeta} \\
        &< \frac{1}{1-\eta}\cdot \frac{2-\eta+W_{0}\left(-\frac{(1-\eta)^2}{e^{2-\eta}} \right)}{\zeta}= \frac{1}{1-\eta}\frac{\zeta}{\zeta}=\frac{1}{1-\eta},
\end{align*}
where the fourth equality follows by Definition \ref{def:lambert}. On the other hand, for $\type=0$ we have that $\bar{\lambda}_{\type}=\frac{1}{1-\eta}$. As these are the only two types by assumption, this leads to $\min_{\type \in T}{\bar{\lambda}}_{\type} \geq \frac{1-(1-\eta)e^{-\zeta}}{\zeta}$. 
\medskip

\noindent \textit{(ii)} For $\type=1$ by our choice of parameters we have that:
\[
\frac{\bar{\mu}_{\type}}{\bar{\lambda}_{\type}} = \frac{1}{1-(1-\eta)e^{-\zeta}} > 1,
\]
where the inequality follows by definition of $\zeta$, as $\eta \in [0,1)$ and by Fact \ref{claim:zeta-range}.
For $\type=0$, we have that $\frac{\bar{\mu}_{\type}}{\bar{\lambda}_{\type}}=1$, and the property follows.
\medskip

\noindent \textit{(iii)} For $\type=1$, we have that $ \frac{1 -t}{\bar{\lambda}_{\type}} = 0 < 1 - \eta$ as $\eta \in [0,1)$. On the other hand, for  $\type=0$ with $t = 0$, we have that $\frac{1 - t}{\bar{\lambda}_{\type}}=1-\eta$, and the property follows.
\medskip

\noindent \textit{(iv)} We have:
\begin{align*}
        2-\eta - (1-\eta)^2e^{-\zeta} &= 
        2 - \eta -(1-\eta)^2e^{-W_{0}\left(-\frac{(1-\eta)^2}{e^{2-\eta}} \right)}e^{-2+\eta} = 2 - \eta + W_{0}\left(-\frac{(1-\eta)^2}{e^{2-\eta}}\right)=\zeta.
\end{align*}
Here, the first equality follows by the definition of the Lambert $W$ function (Definition \ref{def:lambert}).
\end{proof}

We can now prove Theorem~\ref{thm:hybrid-reserve-prices}. 

\begin{proof}[Proof of Theorem~\ref{thm:hybrid-reserve-prices}]
We use $\zeta:= \zeta(\eta)$ for notational convenience. Define $\bar{\vec{\mu}} = \bar{\vec{\mu}}(\eta)$ as in \eqref{eq:mu-for-mixed-agents-with-reserve} and $\bar{\vec{\lambda}} = \bar{\vec{\lambda}}(\eta)$. Note that $\bar{\vec{\mu}}$ is a feasible solution by Corollary \ref{cor:mu-bar:feasible}. 
We proceed as in Section \ref{subsec:solving-poarmp} and use $\bar{\vec{\mu}}$ to obtain the following lower bound on the value of the POA-revealing mathematical program: 
\begin{align}
     \poarmp(\{0,1\})& \geq \min \left\{ 
    \min_{\substack{\type \in T}} \bar{\lambda}_{\type}, 
    \left( 
        \max_{\substack{\type \in T}} \frac{\bar{\mu}_{\type}}{\bar{\lambda}_{\type}} 
        + \max_{\substack{\type \in T}} \frac{1 - t }{\bar{\lambda}_{\type}} 
    \right)^{-1} \right\} \nonumber \\
    &\geq \min \left\{\frac{1-(1-\eta)e^{-\zeta}}{\zeta}, \left(\frac{1}{1-(1-\eta)e^{-\zeta}} +1-\eta\right)^{-1} \right\} \nonumber \\
    &=\left(1-(1-\eta)e^{-\zeta} \right) \cdot \min\left\{\frac{1}{\zeta}, \frac{1}{2-\eta -(1-\eta)^2e^{-\zeta}} \right\} \nonumber \\
    &= \frac{1-(1-\eta)e^{-\zeta}}{\zeta}. \label{eq:hybrid-poa-reserves-analysis}
\end{align}
Here, the second inequality follows from properties \textit{(i)}--\textit{(iii)} of Lemma \ref{claim:technical-claim-hybrid-reserves}. 
The final equality holds due to property \textit{(iv)} of the same lemma. Hence, we obtain:
\begin{align*}
    \cce\textit{-}\poa\left(\mathcal{I}^{\{0,1\}}_{\xos}\right) &\leq \poarmp(\{0,1\})^{-1}\\ &\leq \frac{\zeta}{1-(1-\eta)e^{-\zeta}}\\
    &= (1-\eta) \cdot \frac{\zeta}{1-\eta -(1-\eta)^2e^{-\zeta}}\\
    &=(1-\eta) \cdot \frac{\zeta}{\zeta-1}\\
    &=Q(\eta). 
\end{align*}

The first inequality follows from Theorem \ref{theorem:templateII}, noting that $T^+=T=\{0,1\}$ and that, by assumption, we consider well-supported $\cce$. Then, the second inequality follows from \eqref{eq:hybrid-poa-reserves-analysis}. Finally, the second equality holds by property \textit{(iv)} of Lemma \ref{claim:technical-claim-hybrid-reserves}.

We conclude by showing that $\lim_{z \to 1} Q(z) =1$. Let $h(z) = -(1-z)^{2}e^{z-2}$. By using l'H\^opital's rule, the Fact \ref{fact:lambert-derivative} (the expression for the derivative of $W_0$) and simplifying, it holds that:
\begin{align*}
\lim_{z \to 1} Q(z) &=\lim_{z \to 1} \frac{\left(2 - z + W_{0}(h(z)\right)(1 + W_{0}(h(z)))  + (1-z) \left(1 + 2W_{0}(h(z)) \right)}{  1 + 2W_{0}(h(z))) }=\frac{1+0}{1}=1, 
\end{align*}
as $W_{0}(h(1)) =W_0(0) =0$. This concludes the proof.
\end{proof}

\subsection{Missing proofs of Subsection~\ref{sec:well-supported-eq}}\label{app:thm:CE-well-supported}

\thmCEwellsupported*

\begin{proof}[Proof of Theorem~\ref{thm:CE-well-supported}]
Toward a contradiction, suppose that there exists a $\B \in \ce(\vec{I})$ that is not well-supported, i.e., there exists an auction $j \in M$ such that 
\begin{equation} \label{eq:PositivePrNotSold}
\ppr [ x_j(\B) \neq \vec{0} ] <1.    
\end{equation}
Let agent $i \in N$ be the rightful winner of auction $j$, i.e., $i = \rw(j)$. 
We define the deviation of agent $i$ as follows. For all auctions $k \neq j$, agent $i$ follows the recommendation. However, for auction $j$, agent $i$ follows the recommendation only if the recommendation is to bid above or equal to the reserve price $r_j$. Otherwise, agent $i$ ignores the recommendation and bids exactly the reserve price $r_j$ for auction $j$. More formally, for any $\b_{i}$ in the support of $\B$, we define
\begin{equation}
\swap(\b_i) = \b'_i = (b_{i1}, \dots, b_{i(j-1)}, \max \{b_{ij}, r_j\}, b_{i(j+1)}, \dots,  b_{im}).    \label{eq:deviation}
\end{equation}
Note that by \eqref{eq:PositivePrNotSold}, there exist recommendation(s) $\b_i$ in the support of $\B$ that recommend agent $i$ to bid below the reserve price $r_j$ for auction $j$. 

We first show that the deviation as defined in \eqref{eq:deviation} satisfies the ROI constraint for agent $i$. 
The conditional PDF of $\B_{-i}$ given $\B_{i} = \b_{i}$ is denoted as $f_{\B_{-i} \mid \B_{i}} (\b_{-i} \mid \b_{i})$.
Note that $\B \in {\mathcal{R}_i}$ as $\B$ is a correlated equilibrium, i.e., the ROI constraint is satisfied for agent $i$:
\begin{equation}
\pe \left [ \pe \left [ \sum_{ \ k \in M} p_{ik}(\b_{i}, \B_{-i}) \, \middle\vert \, \B_i = \b_i \right ] \right ]
\le \pe \left [ \pe \left [ \sum_{ \ k \in M} v_{ik}(x_{ik}(\b_{i}, \B_{-i})) \, \middle\vert \, \B_i = \b_i  \right ] \right ]. \label{eq:ROIforCE}
\end{equation}
Recall that the outer expectation is for $\B_i$ and the inner expectation is for $\B_{-i}$.
The ROI constraint for the deviation of agent $i$ is:
\begin{equation}
\pe \! \left [ \pe \! \left [ \! \sum_{ \ k \in M} \! p_{ik}(\swap(\b_{i}), \B_{-i}) \, \middle\vert \, \B_i = \b_i \right ] \right ]
\! \! \le \! \pe \! \left [ \pe \! \left [ \! \sum_{ \ k \in M} \! v_{ik}(x_{ik}(\swap(\b_{i}), \B_{-i})) \, \middle\vert \, \B_i = \b_i  \right ] \right ]. \label{eq:ROIforRecommendation}
\end{equation}
Note that for auctions $k \neq j$, agent $i$ always follows the recommendation, so the terms for these auctions $k$ on the left- and right-hand side of equations \eqref{eq:ROIforCE} and \eqref{eq:ROIforRecommendation} coincide. 
Agent $i$ also follows the recommendation for auction $j$ if $b_{ij} \ge r_j$, so in this case the terms also coincide. 
Finally, consider recommendations for auction $j$ with $b_{ij} < r_j$. In this case the terms for auction $j$ on the left- and right-hand side of equation \eqref{eq:ROIforCE} are equal to 0, as agent $i$ never wins auction $j$ when bidding below the reserve price $r_j$. For the deviation, the terms on the left- and right-hand side in \eqref{eq:ROIforRecommendation} satisfy:
\begin{align*}
&\pe \left [ \pe \left [ p_{ij}(\swap(\b_{i}), \B_{-i}) \mid \B_i = \b_i \wedge b_{ij} < r_j \right ] \right ] \le 
\pe \left [ \pe \left [ v_{ij}(x_{ij}(\swap(\b_{i}), \B_{-i})) \mid \B_i = \b_i \wedge b_{ij} < r_j  \right ] \right ],
\end{align*}
as $r_j < v_{ij}$ (reserve prices are feasible) and agent $i$ deviates to the reserve price $r_j$ for auction $j$ if $b_{ij} < r_j$.
Combining this with the reasoning above, it follows that the deviation of agent $i$ satisfies the ROI constraint in \eqref{eq:ROIforRecommendation}. Since $I$ is, by assumption, budget-free, $(\swap(\vec{B}_i), \vec{B}_{-i})$ satisfies the budget-constraint in \eqref{eq:budget}. Therefore we obtain that $(\swap(\vec{B}_i), \vec{B}_{-i}) \in {\mathcal{R}_i}$. 

We now show that agent $i$ is strictly better off with this deviation. First, note that we have:
\begin{align*}
\pe [ & \z_{i}(\b_{i}, \B_{-i}) \mid \B_{i} = \b_{i}] 
= \int \z_{i}(\b_{i},\b_{-i}) 
\cdot f_{\B_{-i} \mid \B_{i}} (\b_{-i} \mid \b_{i}) d \b_{-i} \\
&= \int \left[ \sum_{ \ k \in M; k \neq j} v_{ik} x_{ik}(\b_{i}, \b_{-i}) - \sens_{i} p_{ik}(\b_{i}, \b_{-i}) \right ] 
\cdot f_{\B_{-i} \mid \B_{i}} (\b_{-i} \mid \b_{i}) d \b_{-i} \\
&+ \int \left [ v_{ij} x_{ij}(\b_{i}, \b_{-i}) - \sens_{i} p_{ij}(\b_{i}, \b_{-i}) \right ] 
\cdot f_{\B_{-i} \mid \B_{i}} (\b_{-i} \mid \b_{i}) d \b_{-i}. 
\end{align*}
Next we have that:
\begin{align*}
\pe [ & \z_{i}(\swap(\b_{i}), \B_{-i}) \mid \B_{i} = \b_{i}] 
= \int \z_{i}(\swap(\b_{i}),\b_{-i}) 
\cdot f_{\B_{-i} \mid \B_{i}} (\b_{-i} \mid \b_{i}) d \b_{-i} \\
&=  \int  \left [   \sum_{\ k \in M; k \neq j}  v_{ik} x_{ik}(\b_{i}, \b_{-i}) - \sens_{i} p_{ik}(\b_{i}, \b_{-i}) \right ] 
\cdot   f_{\B_{-i} \mid \B_{i}} (\b_{-i} \mid \b_{i}) d \b_{-i}\\
&+ \int \left [ v_{ij} x_{ij}(\swap(\b_{i}), \b_{-i}) - \sens_{i} p_{ij}(\swap(\b_{i}), \b_{-i}) \right ] 
\cdot f_{\B_{-i} \mid \B_{i}} (\b_{-i} \mid \b_{i}) d \b_{-i},
\end{align*}
where the second equality follows by definition of the deviation. 
Subtracting the former conditional expectation from the latter leads to:
\begin{align}
\pe [ & \z_{i}(\swap(\b_{i}), \B_{-i}) \mid \B_{i} = \b_{i}] 
- \pe [\z_{i}(\b_{i}, \B_{-i}) \mid \B_{i} = \b_{i}]  \nonumber \\
= & \int \Big [ v_{ij} x_{ij}(\swap(\b_{i}), \b_{-i}) - \sens_{i} p_{ij}(\swap(\b_{i}), \b_{-i}) - v_{ij} x_{ij}(\b) + \sens_{i} p_{ij}(\b) \Big ] \nonumber \cdot f_{\B_{-i} \mid \B_{i}} (\b_{-i} \mid \b_{i}) d \vec{b}_{-i}. 
\end{align}
Note that the expression above in brackets evaluates to $0$ for $\b_i$ with $b_{ij} \ge r_j$, as in that case $b'_{ij} = b_{ij}$.
Otherwise, if $b_{ij} < r_j$, then agent $i$ never wins auction $j$ under $\vec{b}_{i}$ and if agent $i$ wins auction $j$ under $(\swap(\b_{i}), \b_{-i})$, the expression evaluates to
\[ 
v_{ij} x_{ij}(\swap(\b_{i}), \b_{-i}) - \sens_{i} p_{ij}(\swap(\b_{i}), \b_{-i}) 
= v_{ij} - \sens_{i} r_{j} > 0, 
\]
as $b'_{ij} = r_j < v_{ij}$ (reserve prices are feasible) and $\sigma_i \le 1$. 
Furthermore, it follows from \eqref{eq:PositivePrNotSold} that there is a positive probability that agent $i$ wins auction $j$ in this case. 
This leads to:
\begin{align*}
\pe [\z_{i}(\swap(\B_{i}), \B_{-i})] &= \pe \left [ \pe [\z_{i}(\swap(\b_{i}), \B_{-i}) \mid \B_{i} = \b_{i}] \right ] > \pe \left [ \pe [\z_{i}(\b_{i}, \B_{-i}) \mid \B_{i} = \b_{i}] \right ] = \pe [\z_{i}(\B)],
\end{align*}
which implies that $\B \not \in \ce(I)$, a contradiction. 
\end{proof}

\propCCEitemNotAlwaysSold*

\begin{proof}[Proof of Theorem~\ref{prop:CCEitemNotAlwaysSold}]
\label{proof:prop:CCEitemNotAlwaysSold}
Consider a single first-price auction with a feasible reserve price $r$ among 2 agents with values $v_1 =1$ and $v_2 =0$. Let $\sens_1 = \sens_2=1$. 
Assume that ties are broken in favor of agent 1 and note that $r < 1$, as the reserve price is feasible. Also, neither agent is budget-constrained i.e., $\budget_1=\budget_2=\infty$. 
Let $X$ be a random variable with the following CDF and PDF:
\[
F_{X}(x) = \frac{1}{e} \cdot \frac{1- r}{1 - x}
\quad \text{and}\quad f_{X}(x) = \frac{1}{e} \cdot \frac{1- r}{(1 - x)^2}
\quad \text{with}\quad x \in \bigg[0, \frac{e - 1 + r}{e} \bigg].
\] 
Note that the domain is well-defined as $0 \le x < 1$, and it is easy to verify that $F_{X}(\cdot)$ is non-negative and increasing over the domain and $F_{X}( (e - 1 + r)/e ) = 1$.
Now consider the randomized bid profile $\B$ for which the agents always bid identically according to $X$, i.e., draw some value $x$ from $X$ and let both agents bid $x$. 
We prove that $\B$ is a CCE.

Since ties are broken in favor of agent 1, the expected gain of agent $2$ under $\B$ is $0$ and agent 2 satisfies the ROI constraint. Moreover, as $v_2 =0$, no unilateral deviation of agent 2 leads to a positive expected gain, and so agent 2 satisfies the equilibrium condition. 
Note that agent 1 satisfies the ROI constraint for $\B$ as $x \le (e -1 +r)/e < 1 = v_1$, so agent 1 never bids above $v_1$. 
As agent 1 does not win when bidding below the reserve price, the expected gain of agent $1$ is:
\[ \EX [ \z_1(\B)] = \int_{r}^{\frac{e - 1 + r}{e}}(1-x)f_{X}(x) d x = \frac{1- r }{e}\bigg[-\ln(1-x) \bigg]_{r}^{\frac{e - 1 + r}{e}} = \frac{1- r}{e}. \]
First, consider any deterministic unilateral deviation $\bsingle'_1$ of agent 1.
If $\bsingle'_1 \in [0, r)$, then agent 1 will never win the auction and always have a gain of $0$. 
Secondly, consider a bid $\bsingle'_1 \in [r, \frac{e - 1 + r}{e}]$. Let $B_{2} = \vec{B}_{-1}$, then:
\[ 
\EX [\z_1(\bsingle'_1,\Bsingle_{2})] = (1 - \bsingle'_1) F_{X}(\bsingle'_1) = (1 - \bsingle'_1) \frac{1- r }{e(1 - \bsingle'_1)} = \frac{1 - r }{e},
\]
where the first equality holds as agent $1$ wins and pays $\bsingle'_1$ if agent $2$ bids at most $b'_1$. Finally, if $\bsingle'_1 > \frac{e - 1 + r}{e}$, agent 1 always wins the auction but this leads to an expected gain of $1-b'_1$ that is strictly smaller than $\frac{1 - r}{e}$. Therefore, we can conclude that for any randomized unilateral deviation $Y$ with PDF $f_{Y}$ of agent 1 satisfying the ROI constraint, it holds that:
\[ \EX [ \z_1(Y, \Bsingle_2) ] = \int_{0}^{\infty} \EX [\z_1(y,\Bsingle_{2})] \cdot f_{Y}(y) d y  \le \int_{0}^{\infty} \frac{1 - r}{e} f_{Y}(y) d y = \frac{1 - r}{e}.
\]
Therefore, agent 1 also satisfies the equilibrium condition, proving that $\B$ is a CCE that is not well-supported.
Note that this holds for any feasible reserve price $r$. Additionally, for any reserve price $r>0$, the CCE $\B$ is not well-supported as in this case the probability that the item is not sold is greater than 0. 

Additionally, note that this can be extended to hold for any simultaneous first-price auction, i.e., $\fpar{m}{\vec{r}}$ with $m>1$, simply by considering $\B$ for an auction $j \in \{1, 2, \dots, m\}$ and letting the randomized bid profiles of the auctions $i\neq j$ be independent of $\B$.
\end{proof}

\corollaryLBwithReservenotallitemssold*

\begin{proof}[Proof of Corollary~\ref{corollary:LB-withReserve-not-all-items-sold}]
\label{proof:corollary:LB-withReserve-not-all-items-sold}
Consider the $\fpar{}{r}$ among two agents with feasible reserve price and the CCE $\B$ as in the proof of Theorem \ref{prop:CCEitemNotAlwaysSold}.
It is easy to see that the optimal liquid welfare in this instance is $1 = v_{\rw}$. 
The expected liquid welfare of $\B$ is $\pe[\sum_{i \in N} v_i(x_i(\B))] = \ppr[ X \ge r ] =  \frac{e-1}{e}$.
Therefore, $\cce\text{-}\poa \geq \frac{e}{e-1}$, concluding the proof.
\end{proof} 

\propPNEsubmodnotwellsupported*

\begin{proof}[Proof of Theorem \ref{prop:PNE-submod-not-well-supported}]\label{proof:prop:PNE-submod-not-well-supported}
Consider a $\fpar{2}{\vec{r}}$ among two agents with feasible reserve prices. 
Let $\sens_1 = 1$ for agent 1 and for agent 2, let $\sens_2 \le 1$ be arbitrary. Also, let $\budget_1= \budget_2= \infty$. 
Consider the submodular valuations with $v_1(\{i\}) = \varepsilon$, $v_1(\{j\}) = 1$ and $v_1(\{i,j\})=1 + \frac{\varepsilon}{2}$.

Let the values of agent 2 for the items be $v_2(\{i\}) = v_2(\{j\}) = v_2(\{i,j\}) = \frac{\varepsilon}{4}$. 
Now consider the feasible reserve prices $r_i =\frac{\varepsilon}{2}$ and $r_j = 1 - \frac{\varepsilon}{2}$. 
The bid profile $\vec{b}$ with $\vec{b}_1 = (\frac{\varepsilon}{2}, 0)$ and $\vec{b}_2 = (0, 0)$ is a MNE.
First, note that under $\vec{b}$ the ROI constraints are satisfied for both agents.
Also note that agent 2 cannot win anything without violating their ROI constraint, as, for both items, the reserve price exceeds their value.  
Agent 1 only wins item $i$ under $\b$ and thus has a gain of $\z_1(\vec{b}) = \frac{\varepsilon}{2}$.
Note that for any $\vec{b}^1$ with $b^{1}_i \ge \frac{\varepsilon}{2}$ and $b^{1}_j < 1 - \frac{\varepsilon}{2}$, $\vec{b}^2$ with $b^{2}_i < \frac{\varepsilon}{2}$ and $b^{2}_j \ge 1 - \frac{\varepsilon}{2}$, $\vec{b}^3$ with $b^{3}_i \ge \frac{\varepsilon}{2}$ and $b^{3}_j \ge 1 - \frac{\varepsilon}{2}$ and $\vec{b}^4$ with $b^{4}_i < \frac{\varepsilon}{2}$ and $b^{4}_j < 1 - \frac{\varepsilon}{2}$, it holds that $\z_1 ( \vec{b}^{k},\vec{b}_{2}) \le \frac{\varepsilon}{2}$ for $k=1,2,3,4$. 
Therefore, agent 1 also satisfies the equilibrium conditions, as, regardless of the ROI constraint, it holds that for any randomized bid profile $B'_1$
\[
\pe [\z_1(B'_1, \vec{b}_2)] = \int_{0}^{\infty} \z_1 ( \vec{b}',\vec{b}_{2})   f_{\B'_1}(\vec{b}') d \vec{b}' \le \frac{\varepsilon}{2} \int_{0}^{\infty} f_{\B'_1}(\vec{b}') d \vec{b}' = \frac{\varepsilon}{2} = \z_1(\vec{b}).
\]
Therefore, the bid profile $\vec{b}$ is a MNE. Furthermore, it holds that $\vec{b}$ is not well-supported, as the item of auction $j$ is not sold. This concludes the proof of the first statement.

To show the second statement, observe that the optimal liquid welfare in this instance is achieved when agent 1 is allocated both items and is equal to $1 + \frac{\varepsilon}{2}$ (recall that the instance is budget-free). The liquid welfare of the MNE $\vec{b}$ is equal to $\varepsilon$. 
Therefore, their ratio is equal to $\frac{1+\frac{\varepsilon}{2}}{\varepsilon}$, and as $\varepsilon \rightarrow 0$, this becomes unbounded i.e., $\frac{1+\frac{\varepsilon}{2}}{\varepsilon} \rightarrow \infty$. This concludes the proof.
\end{proof}

\end{document}